\documentclass[11 pt]{article}

\usepackage[T1]{fontenc}
\usepackage[utf8]{inputenc}
\usepackage{amsmath,amsthm,amssymb,xcolor}
\usepackage[margin=2.54cm]{geometry}
\usepackage{physics}
\usepackage[colorlinks=true,urlcolor=blue,citecolor=blue,linkcolor=blue]{hyperref}
\usepackage{float}
\usepackage{multirow}
\usepackage{enumitem}
\usepackage{ifthen}

\usepackage{lineno}

\usepackage{pgfplots}
\usepackage{xcolor}
\usepackage{graphicx}
\pgfplotsset{compat=1.7}
\usetikzlibrary{tikzmark, arrows}
\usepackage{tkz-euclide}
\usetikzlibrary{calc}

\usepackage[capitalise]{cleveref}

\newcommand{\eps}{\varepsilon}

\newcommand{\subD}[3]{{#1}^{ #2}_{#3}}
\newcommand{\pL}[3]{L_{#1#2}^{#3}}

\newcommand{\algLnorm}{\texttt{EquivalenceSmall}}
\newcommand{\algEquiv}{\texttt{EquivalenceL2}}

\newcommand{\cltwo}{c_{\text{l2}}}
\newcommand{\ceq}{c_{\text{eq}}}

\usepackage{cleveref}
\Crefname{algocf}{Algorithm}{Algorithms}
\usepackage{enumitem}
\usepackage{float}
\usepackage[figurename=Fig.]{caption}

\usepackage{thmtools} 
\usepackage{thm-restate}

\usepackage[ruled,linesnumbered,vlined,noend]{algorithm2e}
\SetKwComment{Comment}{$\triangleright$\ }{}

\newtheorem{theorem}{Theorem}[section]
\newtheorem{lemma}[theorem]{Lemma}
\newtheorem{remark}[theorem]{Remark}
\newtheorem{corollary}[theorem]{Corollary}

\newtheorem{fact}[lemma]{Fact}
\newtheorem{result}{Result}

\theoremstyle{definition}
\newtheorem{definition}[theorem]{Definition}
\newtheorem{problem}{Problem}

\newcommand{\kl}{D}

\newcommand{\E}{\mathbb{E}}
\DeclareMathOperator{\Var} {Var}
\DeclareMathOperator{\Cov} {Cov}

\newcommand{\cC}{\mathcal{C}}
\newcommand{\cA}{\mathcal{A}}
\newcommand{\cB}{\mathcal{B}}
\newcommand{\cD}{\mathcal{D}}

\newcommand{\pr}{\mathrm{Pr}}

\usepackage{todonotes}

\usepackage{bbm}

\newcounter{mynotes}
\newcommand{\algSmallP}{\ensuremath{\mathsf{SimABC}}}
\newcommand{\algSmallQ}{\ensuremath{\mathsf{SimABC_{CI}}}}

\title{Testing (Conditional) Mutual Information\footnote{Accepted for presentation at the
Conference on Learning Theory (COLT) 2025}}

\date{}

\author{
Jan Seyfried\\
Centre for Quantum Technologies,\\
National University of Singapore\\
jan.seyfried@u.nus.edu
\and
Sayantan Sen\\
Centre for Quantum Technologies,\\
National University of Singapore\\
sayantan789@gmail.com
\and
Marco Tomamichel \\
Department of Electrical and Computer Engineering,\\
Centre for Quantum Technologies,\\
National University of Singapore\\
marco.tomamichel@nus.edu.sg
}

\begin{document}

\maketitle
\begin{abstract}
We investigate the sample complexity of mutual information and conditional mutual information testing. For conditional mutual information testing, given access to independent samples of a triple of random variables $(A, B, C)$ with unknown distribution, we want to distinguish between two cases: (i) $A$ and $C$ are conditionally independent, i.e., $I(A\!:\!C|B) = 0$, and (ii) $A$ and $C$ are conditionally dependent, i.e., $I(A\!:\!C|B) \geq \eps$ for some threshold $\eps$. We establish an upper bound on the number of samples required to distinguish between the two cases with high confidence, as a function of $\eps$ and the three alphabet sizes. We conjecture that our bound is tight and show that this is indeed the case in several parameter regimes. For the special case of mutual information testing (when $B$ is trivial), we establish the necessary and sufficient number of samples required up to polylogarithmic terms.

Our technical contributions include a novel method to efficiently simulate weakly correlated samples from the conditionally independent distribution $P_{A|B} P_{C|B} P_B$ given access to samples from an unknown distribution $P_{ABC}$, and a new estimator for equivalence testing that can handle such correlated samples, which might be of independent interest.

\end{abstract}

\newpage

\tableofcontents

\newpage

\section{Introduction}

Distribution testing is a central problem in computer science and statistics \cite{rubinfeld2012taming,goldreich2017introduction,canonne2020survey,canonne2022topics}. Given access to independent samples from an unknown distribution, the goal of distribution testing is to efficiently determine whether the distribution has some specific property. In particular, we are interested in the sample complexity of this problem; the minimal number of samples needed to correctly identify the property with high confidence. Distribution testing has been a vibrant topic of research recently due to the fact that testing algorithms can succeed with far fewer samples than would be needed for completely learning the unknown distribution. Indeed, for many testing problems, the sample complexity grows sub-linearly with the size of the distribution, making such algorithms amenable to problems involving big data sets \cite{rubinfeld2012taming}.

Independence testing and conditional independence testing are two distribution testing problems that are of particular relevance and have found various applications ranging from the study of causality \cite{granger1980testing,pearl2014probabilistic,spirtes2001causation,zhang2011kernel}, information theory \cite{tomamichel15_cmi} and economics \cite{wang2018characteristic,de2014bayesian} to machine learning and graphical models \cite{neapolitan2004learning,canonne2020testing,daskalakis2017square,bhattacharyya_near-optimal_2021,DBLP:conf/stoc/DaskalakisP21,DBLP:conf/alt/ChooY0C24,DBLP:conf/aistats/WangGTA024}.  In particular, in the context of graphical models such as Bayesian Networks, learning tree-structured distributions \cite{chow1968approximating} (an important subclass of Bayesian networks) reduces to the problem of testing the conditional independence of the underlying distribution \cite{DBLP:conf/stoc/DaskalakisP21}. Recently, \cite{DBLP:conf/alt/ChooY0C24} studied the more general class of \emph{polytree} structured distributions (a Bayesian network is a polytree if the underlying undirected graph is a tree), which again uses a reduction to the conditional independence testing of the associated distribution.

In independence testing, we are given access to i.i.d.\  samples of a pair of random variables $(A, C)$ and our goal is to test if the two random variables are independent of each other, or if they are correlated beyond some threshold.
Similarly, in conditional independence testing, we want to test if a triple of random variables $(A, B, C)$ is conditionally independent and forms a Markov chain $A - B - C$, or if they are conditionally dependent above some threshold. We can always consider independence testing as a special case of conditional independence testing where $B$ is trivial, and hence in the following, we will focus our attention on the latter problem.

The first question to ask then is how we measure conditional dependence. An uncontroversial choice is to measure this in terms of the total variation distance to Markov chains, i.e., to devise an algorithm which can distinguish between the case where $A - B - C$ form a Markov chain and the case where 
the distribution of $(A, B, C)$ is $\eps$-far from any Markov chain in total variation distance. This problem has been studied recently \cite{canonne_testing_2018} and more extensively in the special case of independence testing \cite{batu_testing_2001,acharya2015optimal,diakonikolas_new_2016,diakonikolas2021optimal}. We will discuss these results and their relation to our work in more detail in the next section. However, we argue that measuring distance in Kulback-Leibler (KL) divergence is more natural in this context since the distance to the closest Markov chain in KL divergence, the conditional mutual information~\cite{wyner78_cmi}, has well-defined operational meaning in information theory and is a widely used measure of conditional dependence.

This brings us to conditional mutual information (CMI) testing, which was (to the best of our knowledge) first studied in~\cite{canonne_testing_2018,bhattacharyya_near-optimal_2021}. In CMI testing, we are given a promise that the variables $(A, B, C)$ are either conditionally independent, $I(A\!:\!C|B) = 0$, or that they are indeed sufficiently conditionally dependent, $I(A\!:\!C|B) \geq \eps$.
The promise gap, or threshold, $\eps > 0$, plays an important role in the analysis of this problem, and we measure it in terms of conditional mutual information due to its wide use in information theory, statistics and computer science. The goal of CMI testing is to devise an efficient tester that can correctly distinguish these two cases with high probability (for distributions in between, with $0 < I(A\!:\!C|B) < \eps$, we accept either output). Mutual information (MI) testing and CMI testing will be formally introduced in the next section. 

In this work we investigate the sample complexity of both MI and CMI testing, i.e., we find the minimal number of samples required from an unknown distribution to succeed in these tasks with high probability. It is worth noting at this point that any algorithm that can learn the underlying distribution $P_{ABC}$ of $(A, B, C)$ with sufficient precision will also succeed at this task\,---\,but we will see that this is not efficient, as it would necessitate at least $\Omega(d^3)$ samples, where $d$ is the alphabet size, assumed to be the same for all three variables. Hence, the growth is linear in the alphabet size of the triple $(A, B, C)$. On the other hand, a sub-linear scaling of $O(d^{7/4})$ samples is sufficient to solve this problem~\cite{canonne_testing_2018}.

Our results are much more fine-grained in terms of their dependence on the three alphabet sizes, $d_A$, $d_B$ and $d_C$, as well as the threshold $\eps$, and reveal a complex landscape of scaling regimes depending on the ratios between the different problem parameters.
We solve this problem completely for MI testing, where we give matching upper and lower bounds on the sample complexity in Result~\ref{res:mi}. We find that it scales as
\begin{align}
    \widetilde{\Theta}\left(\min\left\{\frac{d_A^{3/4}d_C^{1/4}}{\eps}, \frac{d_A^{2/3}d_C^{1/3}}{\eps^{4/3}}\right\}\right) .
\end{align}
For CMI testing we show an upper bound on the sample complexity in Result~\ref{res:cmi-upper}. Namely, the sample complexity scales as
\begin{align}
\widetilde{O}\left( \max \left\{ \frac{d_A^{1/2}d_B^{3/4}d_C^{1/2}}{\eps},\min\left\{\frac{d_A^{1/4}d_B^{7/8}d_C^{1/4}}{\eps},\frac{d_A^{2/7}d_B^{6/7}d_C^{2/7}}{\eps^{8/7}}\right\}, \min\left\{\frac{d_A^{3/4}d_B^{3/4}d_C^{1/4}}{\eps},\frac{d_A^{2/3}d_B^{2/3}d_C^{1/3}}{\eps^{4/3}}\right\}\right\} \right) .
\end{align}
We conjecture this to be tight in all regimes but our Result~\ref{res:cmi-lower} only proves this partially. We provide an in-depth discussion of our results in the next section. They exhibit some similarity with the bounds achieved for total variation distance~\cite{canonne_testing_2018}, but differ in some interesting ways that we will discuss. Our treatment is also more complete since we can show matching upper and lower bounds for MI testing and some regimes of CMI testing for general parameters.

To achieve these tight bounds we need to introduce several new techniques that we believe will be of independent interest. For MI testing, we reduce the problem to equivalence testing in Hellinger distance between $P_{AC}$ and $P_A P_C$, and then leverage the fact that our reference distribution is product to improve upon existing techniques~\cite{diakonikolas_new_2016, chan_optimal_2013} to optimally solve this problem by relating it to $\ell_2$ distance equivalence testing. For CMI, we dig deeper as we first need to find a way to efficiently sample from the Markov distribution $P_{A|B} P_{C|B} P_B$ given access only to the joint distribution $P_{ABC}$. This is difficult because sampling from $P_{A|B} P_{C|B} P_B$ requires collisions in the $B$ coordinate, but we fail to see sufficiently many collisions for all low probability outcomes in $B$. We devise such an efficient sampler, which however gives us samples that are not independent. This then necessitates the introduction of a new equivalence tester that can deal with such weakly correlated samples. We describe our methods in detail in the next section.  We note that there have also been some works in the statistics literature, where it is assumed that the number of samples is asymptotically large, independent of the other parameters~\cite{neykov2021minimax,marx2019testing,kim2022local,kim2024conditional}, which is a different setting than our work.

The remainder of this paper is structured as follows.
In \Cref{sec:overview}, we present a detailed overview of our results and the methods we employed to prove them. In \Cref{sec:prelim}, we discuss the necessary preliminaries for this work. Then, in \Cref{sec:kl_hellinger_connection}, we show the reduction from (conditional) mutual information testing to independence testing in squared Hellinger distance. In \Cref{sec:ind_test_hellinger}, we present our tester for independence testing of bipartite distributions, followed by the proof of the lower bound of our independence tester in \Cref{sec:ind_test_lb}. Next, in \Cref{sec:cmi_ub}, we present our conditional mutual information tester, and in \Cref{sec:cmi_lb}, we discuss the lower bounds for conditional mutual information testing. In \Cref{sec:equiv_testing_general}, we present our result for equivalence testing of distributions. Finally, in \Cref{sec:calculation}, we present the proofs of some lemmas and claims that are omitted previously.

\section{Overview of Results}
\label{sec:overview}

\subsection{Problem Setup}
Before we can formally state our results we will introduce our model and formal problem statements. To do this, let us first introduce some notation.

\begin{definition}[Independence and Conditional Independence]\label{defi:cond_ind_intro}
Let $(A, B, C)$ be discrete random variables defined over discrete alphabets $\cA := \{1, 2, \ldots, d_A\}$, $\cB := \{1, 2, \ldots, d_B\}$ and $\cC := \{1, 2, \ldots, d_C\}$, respectively. Then, $A$ and $C$ are said to be \emph{independent} if
\begin{align}
    P_{AC}(a, c) = P_A(a) P_C(c) \qquad \forall a \in \cA, c \in \cC \,,
\end{align}
or, in short, $P_{AC} = P_A P_C$, where $P_A$ and $P_C$ are marginal distributions of $P_{AC}$. We then call $P_{AC}$ a \emph{product distribution}.
Moreover, $A$ and $C$ are said to be \emph{conditionally independent given $B$} if
\begin{align}
    P_{AC|B}(a,c|b) = P_{A|B}(a|b) P_{C|B}(c|b) \qquad \forall a \in \cA, b \in \cB, c \in \cC \,,
\end{align}
or, in short, $P_{AC|B} = P_{A|B} P_{C|B}$, where $P_{AC|B}$, $P_{A|B}$ and $P_{C|B}$ are conditional distributions.
In this case we also say that $P_{ABC}$ forms a \emph{Markov chain $A - B - C$}.
\end{definition}
In the absence of independence, we have correlation between the random variables. Correlation between discrete random variables is naturally measured in terms of mutual information. The mutual information has various operational interpretations in information theory, most prominently in the study of noisy channel coding~\cite{shannon48}. The mutual information, $I(A\!:\!C)$, can be seen as the minimal Kullback-Leibler divergence between the joint distribution of $A$ and $C$ and any product distribution on the same alphabets. It is positive and zero only for product distributions. Similarly, conditional mutual information, $I(A\!:\!C|B)$, is the minimal Kullback-Leibler divergence between the joint distribution of $A$, $B$ and $C$ and any Markov chain $A - B - C$. It is positive and zero only for Markov chains. We will use these observations for our formal definitions.

\begin{definition}[Mutual information and Conditional Mutual information]
        Let $P_{AC}$ be the joint distribution of $A$ and $C$. The \emph{mutual information (MI) of $A$ and $C$} is defined as
        \begin{align}
            I(A:C)_P := \min_{Q_A, Q_C} \kl(P_{AC} \| Q_A Q_C) \,,
        \end{align}
        where the minimization is over all product distributions $Q_{AC} = Q_A Q_C$, and 
        \begin{align}
            \kl(P\|Q) := \sum_x P(x) \log \frac{P(x)}{Q(x)}
        \end{align}
        is the \emph{Kullback-Leibler (KL) divergence between $P$ and $Q$}.\footnote{We will assume here and throughout that the support of $Q$ contains the support of $P$, and use the convention $0 \log 0 = 0$ to deal with zeros, so that the KL-divergence is always finite.} Moreover, the \emph{conditional mutual information (CMI) of $A$ and $C$ given $B$} is defined as
        \begin{align}
            I(A:C|B)_P := \min_{Q_B, Q_{A|B}, Q_{C|B}} \kl(P_{ABC} \| Q_{A|B} Q_B Q_{C|B} ) \,,
        \end{align}
        where the minimization is over all Markov chains $Q_{ABC} = Q_{A|B}Q_B Q_{C|B}$ of the form $A - B - C$.
\end{definition}
We will omit the subscript $P$ from $I(A:C)_P$ and $I(A:C|B)_P$ when it is clear from the context.

We note that the minimum in the above expression is achieved by the marginal distributions of $P_{AC}$, i.e., it holds that $I(A\!:\!C)_P = \kl(P_{AC} \| P_A P_C)$. Similarly, one can show that $I(A\!:\!C|B)_P = \kl(P_{ABC} \| P_{A|B} P_B P_{C|B})$. While these expressions are more commonly used as the definition of (conditional) mutual information, our alternative choice highlights that:
\begin{align*}
    I(A:C)_P \geq \eps &\iff \textnormal{$P_{AC}$ is $\eps$-far from the set of product distributions} \,, \qquad \qquad \textnormal{and} \\
    I(A:C|B)_P \geq \eps &\iff \textnormal{$P_{ABC}$ is $\eps$-far from the set of $A - B - C$ Markov chains} \,.    
\end{align*}
Here and throughout ``$P$ is $\eps$-far from $Q$'' refers to the condition $D(P\|Q) \geq \eps$.
With this in mind we can now state the sample complexity problems this paper is concerned with.
We will now define different testing problems, in which we always want to correctly assign a distribution to one of two classes. We refer to an algorithm performing such a test as a \textit{tester}. Throughout, we say that a tester succeeds if it correctly identifies when the distribution $P$ is in one of the two classes. If neither condition is satisfied, either outcome is accepted. The \textit{sample complexity} of a problem is the minimal number $N$ of samples for which there exists a tester that succeeds with probability at least $2/3$, where $N$ depends on the parameters of the problem.
\begin{problem}[Independence testing] \label{prob:indep} 
    Fix a distance measure $\Delta$, a threshold $\eps > 0$ and alphabet sizes $d_A$ and $d_C$. Consider the following decision problem:
    Given access to $N$ i.i.d.\ samples from an unknown distribution $P_{AC}$, distinguish between the classes
    \begin{enumerate}[label=(\roman*)]
        \item $\Delta(P_{AC}\|P_AP_C) = 0$, and
        \item $\Delta(P_{AC}\|P_AP_C) \geq \eps$.
    \end{enumerate}
    We denote the sample complexity of this problem by $\textnormal{SC}_{\textnormal{Indep}}(\Delta, \eps, d_A, d_C)$.
\end{problem}
We are interested in solving a special case of this problem, where the distance measure $\Delta$ is chosen to equal the Kullback-Leibler divergence.
\begin{problem}[MI testing] \label{prob:MI} 
    Consider the setting of independence testing (\Cref{prob:indep}), where $\Delta(\cdot\|\cdot)=D(\cdot\|\cdot)$, which is equivalent to deciding between the classes
    \begin{enumerate}[label=(\roman*)]
        \item $I(A\!:\!C)_P = 0$, and
        \item $I(A\!:\!C)_P \geq \eps$.
    \end{enumerate}
    We denote the sample complexity by $\textnormal{SC}_{\textnormal{MI}}(\eps, d_A, d_C):=\textnormal{SC}_{\textnormal{Indep}}(D,\eps, d_A, d_C)$.
\end{problem}
This problem was first studied in this form in~\cite{bhattacharyya_near-optimal_2021}, albeit only for the case when $d_A = d_C$. We would like to note that a variant of this problem has been studied in~\cite{batu_testing_2001, levi2013testing,acharya2015optimal,diakonikolas_new_2016,diakonikolas2021optimal}. There, the authors designed an independence testing algorithm that tests if $P_{AC}$ is a product distribution or if $P_{AC}$ is at least $\eps$-far in total variation distance from any product distribution. 

Now we consider the problem of conditional mutual information testing, which we first introduce in a more general form:
\begin{problem}[Conditional independence testing] \label{prob:cindep}
    Fix a threshold $\eps$ and alphabet sizes $d_A$, $d_B$ and $d_C$. Consider the following decision problem:
    Given access to $N$ i.i.d.\ samples from an unknown distribution $P_{ABC}$, distinguish between the cases\footnote{Note that for the definition of the decision gap, we fix the `reference distribution' to be $P_{AB}P_{C|B}$, which might not necessarily be the closest conditionally independent distribution. In general,  $\Delta(P_{ABC}\|P_{AB}P_{C|B})\neq \min_{Q_{ABC}}\Delta(P_{ABC}\|Q_{AB}Q_{C|B})$. However, as mentioned before, this equality does hold if we set $\Delta(\cdot\|\cdot)=D(\cdot\|\cdot)$.} 
    \begin{enumerate}[label=(\roman*)]
        \item $\Delta(P_{ABC}\|P_{AB}P_{C|B}) = 0$ and
        \item $\Delta(P_{ABC}\|P_{AB}P_{C|B})\geq \eps$.
    \end{enumerate}
    We denote the sample complexity by $\textnormal{SC}_{\textnormal{CIndep}}(\Delta, \eps, d_A, d_B, d_C)$.
\end{problem}

Similar to before, when the distance measure $\Delta$ is the Kullback-Leibler divergence, we have the following problem.

\begin{problem}[CMI testing] \label{prob:CMI}
    Consider the setting of conditional independence testing (\Cref{prob:cindep}), where $\Delta(\cdot\|\cdot)=D(\cdot\|\cdot)$, which is equivalent to deciding between the classes
    \begin{enumerate}[label=(\roman*)]
        \item $I(A\!:\!C|B)_P = 0$ and
        \item $I(A\!:\!C|B)_P \geq \eps$.
    \end{enumerate}
    We denote the sample complexity by $\textnormal{SC}_{\textnormal{CMI}}(\eps, d_A, d_B, d_C)=\textnormal{SC}_{\textnormal{CIndep}}(D, \eps, d_A, d_B, d_C)$.
\end{problem}
This problem was first studied in~\cite{bhattacharyya_near-optimal_2021} for $d_A = d_B = d_C$ in the context of learning a special case of Bayesian networks, namely learning tree-structured distributions. A related work in this context is the work of \cite{canonne_testing_2018}. The authors in \cite{canonne_testing_2018} studied a related problem of testing whether $I(A:C|B)=0$, or the underlying distribution is far in variation distance from all such distributions. It is clear that MI testing is a special case of CMI testing when $d_B = 1$, i.e., when $B$ is trivial.

Now we consider the problem of equivalence testing of distributions.

\begin{problem}[Sample complexity of equivalence testing] \label{prob:equiv}
Fix a threshold $\eps$, a distance measure $\Delta$, an integer $n$, and a positive real number $b$. Consider the following decision problem: Given access to $N$ i.i.d.\ samples from two unknown distributions $P$ and $Q$ each defined over $[n]$, with the additional guarantee that $\max\{\|P\|_2,\|Q\|_2\} \leq b$, distinguish between the classes
    \begin{enumerate}[label=(\roman*)]
        \item $P=Q$ and
        \item $\Delta(P\|Q) \geq \eps$.
    \end{enumerate}
    We denote the sample complexity of this problem by $\textnormal{SC}_{\textnormal{EQIV}}(\Delta,\eps,b, n)$.
\end{problem}
Equivalence testing has been studied for a long time in the literature, starting with the work of \cite{batu2013testing}, who studied the problem when the distance measure is the $\ell_1$ distance and there is no bound on the $2$-norm of the distributions. There have been several follow-up works such as \cite{bhattacharya2015testing,acharya2014sublinear,diakonikolas2021optimal}. For the particular setting of \Cref{prob:equiv}, where $b$ is treated as an input parameter, it is known that $\textnormal{SC}_{\textnormal{EQIV}}(\ell_2,\eps,b, n)=\Theta(b/\eps^2)$, $\textnormal{SC}_{\textnormal{EQIV}}(\ell_1,\eps,b, n)=O(nb/\eps^2)$ \cite[Thm.\ 2]{chan_optimal_2013}. We note that we can replace the `$\max$' with a `$\min$' at the cost of an additive factor $O(\sqrt{d})$ in the sample complexity, in which case the algorithm first tests whether $\|P\|_2=\Theta(\|Q\|_2)$ \cite[Lemma 2.3]{diakonikolas_new_2016}.

Note that in all of these problem statements, the probability $2/3$ can be replaced by any target probability larger than $1/2$ as we are not interested in the dependence on this confidence parameter and the success probability can be amplified by repeating the tester a constant number of times.

\subsection{Our Results in Context}

In the following, we present our main results. The first result establishes the exact sample complexity for MI testing.

\begin{result}[MI testing]\label{res:mi}
Consider Problem~\ref{prob:MI} and assume $d_A \geq d_C$. Then,
\begin{align}
    \textnormal{SC}_{\textnormal{MI}}(\eps, d_A, d_C) = \widetilde{\Theta}\left(\min\left\{\frac{d_A^{3/4}d_C^{1/4}}{\eps}, \frac{d_A^{2/3}d_C^{1/3}}{\eps^{4/3}}\right\}\right) .
\end{align}
\end{result}

Here $\tilde{\Theta}$ hides terms that are poly-logarithmic in the problem parameters. The initial study of this problem \cite{bhattacharyya_near-optimal_2021} considered only the case $d_A = d_C = d$ with a sample complexity in $\widetilde{O}(d^2/\eps)$. Later,~\cite[Theorem 1.18]{flammia_quantum_2023} improved this bound to $\widetilde{O}(d/\eps)$. If we adopt the computations in \cite{flammia_quantum_2023} with independent alphabet sizes, this gives us a bound of $\widetilde{O}((d_A + d_C)/\eps)$ samples. Using the techniques described in \Cref{sec:kl_hellinger_connection}, it is easy to see that equivalence testing in the squared Hellinger distance can also be used for (conditional) independence testing. For two unknown distributions of dimension $d$, equivalence testing has a sample complexity of $\Theta(\min\{{d_A^{3/4}d_C^{3/4}}/{\eps},{d_A^{2/3}d_C^{2/3}}/{\eps^{4/3}}\})$, \cite[Thm.\ 4.2]{sublinearly_Testable} and \cite[Prop.\ 3.8]{diakonikolas_new_2016}, which results in a bound of $\tilde O(\min\{{d_A^{3/4}d_C^{3/4}}/{\eps},{d_A^{2/3}d_C^{2/3}}/{\eps^{4/3}}\})$ for independence testing. The similarity between this sample complexity and our  \Cref{res:mi} is not surprising, as independence testing is equivalent to equivalence testing where one of the distributions is guaranteed to having product structure (see \Cref{cor:indtesthellinger_same}).

Moreover, combining the independence testing result of \cite{diakonikolas_new_2016} in total variation distance along with~\cite[Lemma A.1]{canonne_testing_2018} and assuming $d_A\geq d_C$, this bound can be further improved to 
\begin{align}
    \widetilde{O} \left(\min\left\{\frac{d_A}{\eps},\frac{d_A^{3/4}d_C^{3/4}}{\eps},\frac{d_A^{2/3}d_C^{2/3}}{\eps^{4/3}},\max\left\{\frac{d_A^{1/2} d_C^{1/2}}{\eps^2}, \frac{d_A^{2/3} d_C^{1/3}}{\eps^{4/3}} \right\}\right\} \right) .
\end{align}
Comparing this best previous bound with our bound in Result~\ref{res:mi}, we notice an improvement unless $d_C=\Theta(d_A)$, in which case the result reported in \cite{flammia_quantum_2023} is already optimal. 
Moreover, we also prove that the dependencies on $d_A$ and $d_C$ are tight, ignoring poly-logarithmic factors. For the lower bounds, we use an information theoretic approach akin to the technique used in \cite{diakonikolas_new_2016} to prove lower bounds on the problem in variation distance. To the best of our knowledge, this is the first sample optimal mutual information tester in the literature.

Now let us proceed to present our main results on CMI testing.

\begin{result}[CMI testing, upper bound]\label{res:cmi-upper}
Consider Problem~\ref{prob:CMI} and assume $d_A \geq d_C$. Then,
\begin{align}
\textnormal{SC}_{\textnormal{CMI}}(\eps, d_A, d_B, d_C) = \widetilde{O}\left( \max \big\{ f_{\textnormal{sym}}(\eps, d_A, d_B, d_C) , f_{\textnormal{asym}}(\eps, d_A, d_B, d_C) \big\} \right),
\end{align}
where we distinguish between terms that are symmetrical and asymmetrical in $A$ and $C$, 
\begin{align}
    f_{\textnormal{sym}}(\eps, d_A, d_B, d_C) &:= \max\left\{\frac{d_A^{1/2}d_B^{3/4}d_C^{1/2}}{\eps},\min\left\{\frac{d_A^{1/4}d_B^{7/8}d_C^{1/4}}{\eps},\frac{d_A^{2/7}d_B^{6/7}d_C^{2/7}}  {\eps^{8/7}}\right\}\right\} , \quad \text{and} \label{eq_def_f_sym}\\
f_{\textnormal{asym}}(\eps, d_A, d_B, d_C) &:=
    \min\left\{\frac{d_A^{3/4}d_B^{3/4}d_C^{1/4}}{\eps},\frac{d_A^{2/3}d_B^{2/3}d_C^{1/3}}{\eps^{4/3}}\right\}.\label{eq_def_f_asym}
\end{align}
\end{result}

The prior work in \cite{bhattacharyya_near-optimal_2021} showed a sample complexity of $\widetilde{O}(d^3/\eps)$, assuming all alphabets to be of the same size. As for mutual information testing, prior work on testing for conditional independence in variation distance \cite{canonne_testing_2018} can again be translated to a bound on conditional mutual information testing \cite[Lemma A.1]{canonne_testing_2018}, and yields a sample complexity of
\begin{equation}   \label{eq:canonne_cmi_trace_dist_upper_intro}
    \widetilde{O}\left(\max \left\{ \frac{d_B^{1/2}d_A^{1/2}d_C^{1/2}}{\eps^2}, \frac{d_B^{2/3}d_A^{2/3}d_C^{1/3}}{\eps^{4/3}}, \frac{d_B^{3/4}d_A^{1/2}d_C^{1/2}}{\eps}, \min \left\{\frac{d_B^{7/8}d_A^{1/4}d_C^{1/4}}{\eps}, \frac{d_B^{6/7}d_A^{2/7}d_C^{2/7}}{\eps^{8/7}}\right\}\right\}\right).
\end{equation}
The resulting algorithm will in principle show an $\eps$-scaling in $\widetilde{O}(1/\eps^2)$, but has the advantage of a sublinear dimensional scaling compared to \cite{bhattacharyya_near-optimal_2021}. Our algorithm manages to materialize the best of both worlds, achieving a linear scaling in $\widetilde{O}(1/\eps)$, as well as a sublinear dimensional scaling.

An interesting immediate observation is that in the regimes which are not dominated by $\eps$, our sample complexity coincides with \eqref{eq:canonne_cmi_trace_dist_upper_intro}. It is important to stress that, analogous to the case of independence testing, the underlying approaches to solve the two problems vary significantly:  the authors in \cite{canonne_testing_2018} studied the behavior of a polynomial defined over a probability distribution, and used this polynomial to design an efficient tester for their problem. In a sense, they follow an algebraic approach.  While their technique provides a certain generality and can in principle be applied to other learning problems, it is currently not known whether their approach can be modified to obtain bounds for our problem which improve beyond \eqref{eq:canonne_cmi_trace_dist_upper_intro}. This is not unexpected when we compare to the simpler problem of independence testing: the known (sample optimal) algorithms for MI-testing and testing independence with respect to variational distance use very different approaches as well. Our approach reduces the problem to a polylogarithmic number of instances of equivalence testing in $\ell_2$ distance between an unknown distribution and a reference distribution simulated from the unknown distribution.

To obtain lower bounds, we can combine a reduction to MI-testing with existing lower bounds on independence testing coming from variation distance \cite{canonne_testing_2018} to derive the following lower bound for CMI testing.

\begin{result}[CMI testing, lower bound]\label{res:cmi-lower}
    Consider Problem~\ref{prob:CMI} and assume $d_A \geq d_C$. Then,
    \begin{align}
    \textnormal{SC}_{\textnormal{CMI}}(\eps, d_A, d_B, d_C) = \widetilde{\Omega}\left( f_{\textnormal{sym}}(\eps, 1, d_B, 1) , f_{\textnormal{asym}}(\eps, d_A, d_B, d_C) \big\} \right) .
    \end{align}
\end{result} 

Comparing this with Result~\ref{res:cmi-upper}, we see that in the symmetric term, we are missing the dependence on $d_A$ and $d_C$. For the special case where $\eps=\Omega(1)$ and $d=d_A=d_B=d_C$, \cite{canonne_testing_2018} provides a lower bound of $\Omega(d^{7/4})$. While there remain open questions on the exact form of the lower bounds, these findings strongly indicate that the complicated structure of the sample complexity we report in Result~\ref{res:cmi-upper} is indeed required.

Our algorithm for CMI testing uses a new estimator, and we show that we can recover the aforementioned bounds for equivalence testing in the $\ell_1$ and $\ell_2$ distance \cite[Thm.\ 2]{chan_optimal_2013} and \cite[Lemma 2.3]{diakonikolas_new_2016}, \Cref{prob:equiv}, in \Cref{sec:equiv_testing_general}.
\begin{align}
    &\textnormal{SC}_{\textnormal{EQIV}}(\ell_1, \eps,b,n) = O\left(\frac{bn}{\eps^2}\right),\qquad\textnormal{SC}_{\textnormal{EQIV}}(\ell_2, \eps,b,n) = O\left(\frac{b}{\eps^2}\right).
\end{align}
We expect that our estimator will show favorable properties when considering data which violates the i.i.d.\ assumption.

\subsection{Methods}\label{sec:overview_method}

Here we would like to give an overview of the methods we used to prove our results.

\paragraph*{Reduction to testing in Hellinger distance:}
The first step in solving the problem of testing for (C)MI is to reduce it to the testing for (conditional) independence with respect to the squared Hellinger distance ($D_H^2$ in short). For this purpose, we use the following inequalities from \cite[Proposition 2.12]{flammia_quantum_2023} and \cite[p.\ 429]{bounding_prob_metrics} (see \Cref{lemma:kl_to_hell_flaOd}):
\begin{equation}\label{eqn:klhellinger_intro}
D_H^2(P,Q)   \leq \kl(P\|Q)\leq \left(2+\log\left(\max_{i\in[d],P(i)\neq 0}\frac{P(i)}{Q(i)}\right)\right) D_H^2(P,Q).
\end{equation}
In order to use this inequality gainfully, we need to ensure a lower bound on $Q_{\min}:=\min_{i\in [d]}Q[i]$. We achieve this by taking a mixture of the reference distribution with the uniform distribution such that $Q_{\min}$ does not become too small, while also preserving (conditional) independence and $\eps$-farness between distributions (up to a constant factor). This requires a specific continuity result for both the KL-divergence and the conditional mutual information which we prove in \Cref{sec:kl_hellinger_connection}.

\paragraph*{MI Testing:} To perform independence testing in $D_H^2$, we note that sample access to $P_{AC}$ allows us to directly simulate samples from $P_{A}P_C$. The question is then whether these two distributions are the same or far from each other. In fact, we solve a more general problem by testing for equivalence between $P_{AC}$ and an arbitrary product distribution $Q_AQ_C$. Our matching lower bounds show that the problems have the same asymptotic sample complexity. A known idea for equivalence testing is to reduce the problem from the $D_H^2$ metric to testing in $\ell_2$ distance (see \cite[Theorem\ 2]{chan_optimal_2013}, \cite{diakonikolas_new_2016}), which is achieved by splitting a distribution into a logarithmic number of buckets such that all elements in a bucket have \emph{roughly} the same weight. This allows us to tightly bound $D_H^2$ in terms of $\ell_2$ distance. Each bucket can then individually be tested for equivalence in $\ell_2$ distance. One bucket, $S_{\text{small}}$ contains all small weights and requires special treatment. 

We introduce a more fine grained version of this approach by leveraging the fact that our reference distribution is a product distribution, illustrated in \Cref{fig:mi_grid_intro}. This allows us to refine the bucketing, by effectively creating a two-dimensional grid of buckets, one dimension for $A$ and $C$ each, which ultimately leads to an improved sample complexity. Again, buckets with small weights (i.e.\ either $p_a$ or $p_c$ are small) need to be treated separately.
The entire argument is laid out and formally proved in \Cref{sec:ind_test_hellinger}. Our approach explains the origin of the two regimes in the sample complexity: the difficulty of equivalence testing on a bucket $S$ depends on the $\ell_2$ norm of $P$ on $S$, $(\sum_{i\in S}P(i)^2)^{1/2}$, which we can in general bound in two ways, $(\sum_{i\in S}P(i)^2)^{1/2}\leq \min\{\sqrt{|S|}p_{\max},\sqrt{p_{\max}}\}$, where $p_{\max}:=\max_{i\in S}\{P(i)\}$, giving rise to the two regimes in our results. Since the threshold for which indices are placed in the bucket of small weights, $S_{\text{small}}$, depends on the number of samples $N$ we use, the $\ell_2$ norm of $P$ on $S_{\text{small}}$ depends on $N$ again: solving the resulting expression for $N$ then gives us the sample complexity in \Cref{res:mi} and explains the seemingly counterintuitive exponents in the sample complexity of our results.

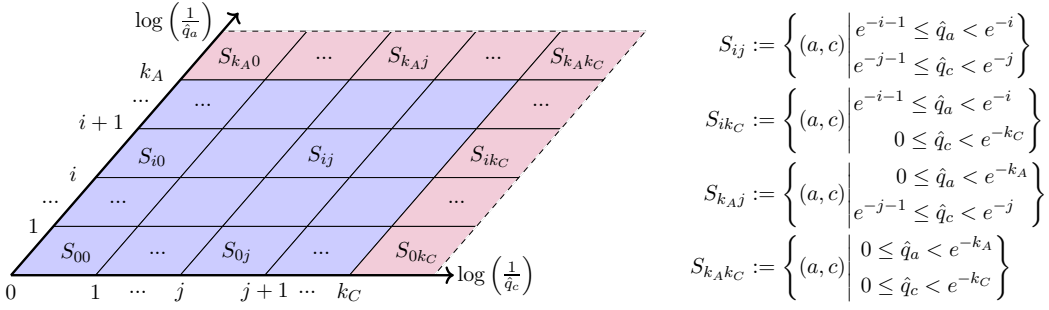
\begin{figure}[H]
        \begin{center}
\scalebox{0.75}{
\begin{tikzpicture}

\def\binX{1}
\def\tcol{blue}

\def\angleP{pi/3}
\def\sf{1.5}
\def\hX{\binX*cos(\angleP)} \def\hY{\binX*sin(\angleP)}
    
\foreach \j in {0,...,4}{
\foreach \i in {4,...,0}
{
    \def\spx{(\j*\binX+\hX*\i)}
    \def\spy{\hY*\i}
    \tkzDefPoint(\spx*\sf,\spy){A} 
    
    \tkzDefPoint((\spx+\binX)*\sf,\spy){B} 
    \tkzDefPoint((\spx+\binX+\hX)*\sf,\spy+\hY){C} 
    \tkzDefPointWith[colinear= at C](B,A) \tkzGetPoint{D}
    \ifthenelse{\i = 4}{
        \ifthenelse{\j = 4}{\def\tcol{purple}}{\def\tcol{purple}}
    }{
        \ifthenelse{\j = 4}{\def\tcol{purple}}{\def\tcol{blue}}
    } 
    \tkzDrawPolygon[fill=\tcol!20](A,B,C,D)

}
}

\tkzDefPoint(5*\binX*\sf,0){X};
\tkzDefPoint(5*\binX*\sf+5*\hX*\sf, 5*\hY){Y};
\draw[dashed, white, thick] (X) -- (Y);
\tkzDefPoint(5*\hX*\sf, 5*\hY){X};
\tkzDefPoint(5*\hX*\sf+5*\binX*\sf, 5*\hY){Y};
\draw[dashed, white, thick] (X) -- (Y);

\draw[->, very thick] (0,0) -- (5*\binX*\sf+0.35,0);
\tkzDefPoint(5*\hX*\sf*1.075,5*\hY*1.075){E}
\draw[->, very thick] (0,0) -- (E);

\node at (2.9, 4.5) {$\log\left(\frac{1}{\hat q_a}\right)$};
\node at (8.6, 0) {$\log\left(\frac{1}{\hat q_c}\right)$};

\tkzDefPoint(\hX-0.1, \hY){P};
\node at (P) {$1$};
\tkzDefPoint(1.5*\hX-0.05, 1.5*\hY){P};\node at (P) {$...$};
\tkzDefPoint(1*\hX+\binX*\sf-0.1, 1.5*\hY){P};\node at (P) {$...$};
\tkzDefPoint(1*\hX+5*\binX*\sf-0.1, 1.5*\hY){P};\node at (P) {$...$};
\tkzDefPoint(2*\hX-0.05, 2*\hY){P};
\node at (1.1, 1.8) {$i$};
\tkzDefPoint(3*\hX-0.05, 3*\hY){P};
\node at (1.6, 2.7) {$i+1$};
\node at (2.3, 3.1) {$...$};
\tkzDefPoint(4*\hX+\binX*\sf-0.1, 3.5*\hY){P};\node at (P) {$...$};
\tkzDefPoint(4*\hX+5*\binX*\sf-0.1, 3.5*\hY){P};\node at (P) {$...$};
\tkzDefPoint(4*\hX-0.05, 4*\hY){P};
\node at (2.5, 3.6) {$k_A$};

\node at (0, -0.3) {$0$};
\node at (\binX*\sf, -0.3) {$1$};
\node at (1.5*\binX*\sf, -0.3) {$...$};
\node at (2*\binX*\sf, -0.3) {$j$};
\node at (3*\binX*\sf, -0.3) {$j+1$};
\node at (3.5*\binX*\sf, -0.3) {$...$};
\node at (4*\binX*\sf, -0.3) {$k_C$};
\node at (5*\binX*\sf, -0.3) {$ $};

\node at (1.15, 0.4) {$S_{00}$};

\node at (4.05, 3.85) {$S_{k_A0}$};
\node at (5.55, 3.85) {$...$};
\node at (7.05, 3.85) {$S_{k_Aj}$};
\node at (8.55, 3.85) {$...$};
\node at (10.05, 3.85) {$S_{k_Ak_C}$};
\node at (2.6, 0.4) {$...$};
\node at (4, 0.4) {$S_{0j}$};
\node at (5.5, 0.4) {$...$};
\node at (7.2, 0.4) {$S_{0k_C}$};
\node at (5.5, 2.1) {$S_{ij}$};
\node at (2.5, 2.1) {$S_{i0}$};
\node at (8.5, 2.1) {$S_{ik_C}$};

\node at (15.2,2.1) {$\begin{aligned}
S_{ij}&:=\left\{(a,c)\middle|{\scriptscriptstyle\begin{aligned}e^{-i-1}&\leq \hat q_a< e^{-i}\\ e^{-j-1}&\leq \hat q_c< e^{-j}\end{aligned}}\right\}
\\
S_{ik_C}&:=\left\{(a,c)\middle|{\scriptscriptstyle\begin{aligned}e^{-i-1}&\leq \hat q_a< e^{-i}
\\ 0&\leq \hat q_c< e^{-k_C}\end{aligned}}\right\}
\\
S_{k_Aj}&:=\left\{(a,c)\middle|{\displaystyle\begin{aligned}0&\leq \hat q_a< e^{-k_A} \\ e^{-j-1}&\leq \hat q_c< e^{-j}\end{aligned}}\right\}
\\
S_{k_Ak_C}&:=\left\{(a,c)\middle|{\textstyle\begin{aligned} ~0&\leq \hat q_a< e^{-k_A} \\  ~0&\leq \hat q_c< e^{-k_C}\end{aligned}}\right\}
\end{aligned}$};

\end{tikzpicture}
}
\caption{\label{fig:mi_grid_intro} Partition of $d_A\times d_C$ based on $\hat Q_A$ and $\hat Q_C$. Indices $(a,c)$ of similar weight $\hat q_a\hat q_c$ are grouped together in buckets $S_{ij}$, which are used to perform piecewise equivalence testing with $P_{AC}$. The axes are labeled according to the corresponding category, which is inverse logarithmic to the weight of the probabilities. The color of the categories indicate a different analysis of the sample complexity of the categories. The red regimes dominate the sample complexity.}
    \end{center}
\end{figure}

\paragraph*{CMI Testing:} Similar to MI testing, we reduce CMI testing to testing for conditional independence in $D_H^2$. However, unlike for MI, we can no longer sample directly from the reference distribution $Q_{ABC}:=P_{AB}P_{C|B}$, which is crucial for us. We overcome this bottleneck by performing a case distinction depending on whether $p_b$ for a given $b \in B$ is large or small. Both regimes are tested separately to see whether the conditional mutual information is zero or at least $\eps/2$. The sample complexity will be dominated by the fact that we need to generate enough samples for both regimes. This is illustrated in \Cref{fig:cmi_sampling}.
\begin{itemize}
    \item To generate samples from the first, `\emph{large}' regime, we use two phases of taking samples. In the first phase, we take $\widetilde{O}(N)$ samples and sort their $A$-coordinates into queues depending on their $b$-values. In the second phase, we take $O(N)$ more samples: if we draw a sample $(a,b,c)$ from the large regime, we remove an element from the corresponding $b$-queue, say $a'$, and we output $(a',b,c)$. The probability to observe $(a',b,c)$ is easily seen to equal $p_{a'|b}p_{bc}$, which corresponds exactly to our reference distribution $Q_{ABC}$. Concentration bounds guarantee that with high probability, enough large samples are available in the respective queues from the reference distribution.\newline
    For testing, we follow a similar approach as for independence testing. Now, instead of the two-dimensional bucketing from MI testing, here the bucketing needs to be performed in three dimensions, as shown in \Cref{fig:cmi_partition_intro}. As for the MI-testing, the sample complexity will be dominated by accounting for elements with smaller probability masses. The case distinction in the sample complexity follows from bounding the $\ell_2$ norm of specific buckets in two different ways.
    \begin{figure}[H]
        \begin{center}
    \scalebox{0.75}{
    \begin{tikzpicture}

\def\binX{1}
\def\tcol{blue}


\def\angleP{pi/3}
\def\sf{1.5}
\def\hX{\binX*cos(\angleP)} \def\hY{\binX*sin(\angleP)}

\def\spYtop{4.5}
    
\foreach \j in {0,...,4}{
\foreach \i in {4,...,0}
{
    \def\spx{(\j*\binX+\hX*\i)}
    \def\spy{\spYtop+\hY*\i}
    
    \tkzDefPoint(\spx*\sf,\spy){A} 
    
    \tkzDefPoint((\spx+\binX)*\sf,\spy){B} 
    \tkzDefPoint((\spx+\binX+\hX)*\sf,\spy+\hY){C} 
    \tkzDefPointWith[colinear= at C](B,A) \tkzGetPoint{D}

    \ifthenelse{\i = 4}{
        \ifthenelse{\j = 4}{\def\tcol{red}}{\def\tcol{purple}}
    }{
        \ifthenelse{\j = 4}{\def\tcol{purple}}{\def\tcol{blue}}
    } 
    \tkzDrawPolygon[fill=\tcol!20](A,B,C,D)

}
}

\tkzDefPoint(\hX*\sf*5-0.8,4.5+5*\hY+0.2){X};
\node at (X) {$\log\left(\frac{1}{\hat p_{ab}}\right)$};

\tkzDefPoint(\hX*\sf*5,4.5+5*\hY){X};
\tkzDefPoint((\binX+\hX)*\sf*5,4.5+5*\hY){Y};

\draw[dashed, white, thick] (X) -- (Y);
\tkzDefPoint((\binX+\hX)*\sf*5,5*\hY-1){X};

\tkzDefPoint(\hX*\sf+0.05,4.5+\hY+0.35){U};
\node[anchor=east] at (U) {$1$};
\tkzDefPoint(2*\hX*\sf+0.05,4.5+2*\hY+0.35){U};
\node[anchor=east] at (U) {$i$};
\tkzDefPoint(3*\hX*\sf+0.05,4.5+3*\hY+0.35){U};
\node[anchor=east] at (U) {$i+1$};
\tkzDefPoint(4*\hX*\sf+0.15,4.5+4*\hY+0.35){U};
\node[anchor=east] at (U) {$k_A$};

\tkzDefPoint(2.5*\binX*\sf+\hX*\sf*2.5,4.5+2.5*\hY){U};
\node at (U) {$\pL{i}{j}{0}$};

\tkzDefPoint(\hX*\sf*5.5,4.5+4.5*\hY){U};
\node at (U) {$\pL{k_A}{0}{0}$};

\tkzDefPoint(\hX*\sf*4.5,4.5+3.5*\hY){U};
\node at (U) {$...$};
\tkzDefPoint(2*\binX*\sf+\hX*\sf*4.5,4.5+3.5*\hY){U};
\node at (U) {$...$};

\tkzDefPoint(\hX*\sf*3.5,4.5+2.5*\hY){U};
\node at (U) {$\pL{i}{0}{0}$};
\tkzDefPoint(\binX*\sf+\hX*\sf*3.5,4.5+2.5*\hY){U};
\node at (U) {$...$};

\tkzDefPoint(\hX*\sf*2.5,4.5+1.5*\hY){U};
\node at (U) {$...$};
\tkzDefPoint(2*\binX*\sf+\hX*\sf*2.5,4.5+1.5*\hY){U};
\node at (U) {$...$};

\tkzDefPoint(2.5*\binX*\sf+\hX*\sf*4.5,4.5+4.5*\hY){U};
\node at (U) {$\pL{k_A}{j}{0}$};

\tkzDefPoint(4.5*\binX*\sf+\hX*\sf*4.5,4.5+4.5*\hY){U};
\node at (U) {$\pL{k_A}{k_C}{0}$};


\def\binX{1.5}
\def\angleP{pi/2}
\def\sf{1}
\def\hX{\binX*cos(\angleP)} \def\hY{\binX*sin(\angleP)}

\foreach \j in {0,...,4}{
\foreach \i in {2,...,0}
{
    \def\spx{(\j*\binX+\hX*\i)}
    \def\spy{\hY*\i}
    
    \tkzDefPoint(\spx*\sf,\spy){A} 
    
    \tkzDefPoint((\spx+\binX)*\sf,\spy){B} 
    \tkzDefPoint((\spx+\binX+\hX)*\sf,\spy+\hY){C} 
    \tkzDefPointWith[colinear= at C](B,A) \tkzGetPoint{D}
    
    \ifthenelse{\j = 4}{
        \def\tcol{purple}
    }{
        \def\tcol{blue}
    } 
    \tkzDrawPolygon[fill=\tcol!20](A,B,C,D)

}
}

\node[anchor=east] at (0.05, 4.85) {$0$};
\node[anchor=east] at (\binX*\sf+0.1, 4.85) {$1$};
\node[anchor=east] at (2*\binX*\sf+0.1, 4.85) {$j$};
\node[anchor=east] at (3*\binX*\sf+0.15, 4.85) {$j+1$};
\node[anchor=east] at (4*\binX*\sf+0.1, 4.85) {$k_C$};

\node[anchor=east] at (0.0, 3.35) {$1$};
\node[anchor=east] at (0.05, 1.85) {$k_B$};

\node at (1.5*\sf*0.5, 0.75) {$\pL{0}{0}{k_B}$};
\node at (1.5*\sf*0.5, 1.5+0.75) {$...$};
\node at (1.5*\sf*0.5, 3+0.75) {$\pL{0}{0}{0}$};
\node at (1.5*\sf*1.5, 3+0.75) {$...$};
\node at (1.5*\sf*3.5, 3+0.75) {$...$};

\node at (1.5*\sf*1.5, 0.75) {$...$};
\node at (1.5*\sf*3.5, 0.75) {$...$};

\node at (1.5*\sf*2.5, 0.75) {$\pL{0}{j}{k_B}$};
\node at (1.5*\sf*2.5, 1.5+0.75) {$...$};
\node at (1.5*\sf*2.5, 3+0.75) {$\pL{0}{j}{0}$};

\node at (1.5*\sf*4.5, 0.75) {$\pL{0}{k_C}{k_B}$};
\node at (1.5*\sf*4.5, 1.5+0.75) {$...$};
\node at (1.5*\sf*4.5, 3+0.75) {$\pL{0}{k_C}{0}$};

\node at (-1,-0.6) {$\log\left(\frac{1}{\hat p_{b}}\right)$};

\def\binX{1.5}
\def\spXs{7.5}
\def\sf{2/3}
\def\angleP{pi/3}
\def\hX{\binX*cos(\angleP)} \def\hY{\binX*sin(\angleP)}

\foreach \j in {0,...,4}{
\foreach \i in {2,...,0}
{
    \def\spx{\spXs+\j*\hX}
    \def\spy{\binX*\i+\j*\hY*\sf}
    
    \tkzDefPoint(\spx,\spy){A} 
    
    \tkzDefPoint((\spx+\hX),\spy+\hY*\sf){B} 
    \tkzDefPoint((\spx+\hX),\spy+\hY*\sf+\binX){C} 
    \tkzDefPointWith[colinear= at C](B,A) \tkzGetPoint{D}
    
    \ifthenelse{\j = 4}{
        \def\tcol{red}
    }{
        \def\tcol{purple}
    } 

    \tkzDrawPolygon[fill=\tcol!20](A,B,C,D)

}
}

\tkzDefPoint(5*(3/2)+1,4){L} 
\node[fill=white,opacity=0.9, text opacity=1] at (L) {$\log\left(\frac{1}{\hat p_{bc}}\right)$};

\def\tcol{yellow}
\tkzDefPoint(0,-1){A} 
\tkzDefPoint(1.5*5,-1){B} 
\tkzDefPoint(1.5*5,0){C}
\tkzDefPointWith[colinear= at C](B,A) \tkzGetPoint{D}
\tkzDrawPolygon[fill=\tcol!20](A,B,C,D)

\tkzDefPoint(1.5*5,-1){A} 
\tkzDefPoint(1.5*5+5*\hX,-1+5*\hY*\sf){B} 
\tkzDefPoint(1.5*5+5*\hX,5*\hY*\sf){C}
\tkzDefPointWith[colinear= at C](B,A) \tkzGetPoint{D}
\tkzDrawPolygon[fill=\tcol!20](A,B,C,D)

\node at (1.5*2.5, -0.5) {Small regime $S$};

\draw[dashed, white, thick] (Y) -- (X);
\draw[dashed, white, thick] (\binX*5,-1) -- (X);
\draw[dashed, white, thick] (\binX*5,-1) -- (0,-1);

\def\angleP{pi/3}
\def\sf{1.5}
\def\hX{\binX*cos(\angleP)} \def\hY{\binX*sin(\angleP)}
\draw[->, very thick] (0,4.5) -- (5*1.5+0.35,4.5);
\draw[->, very thick] (0,4.5) -- (0,-1.35);
\tkzDefPoint(3.55*\hX*\sf,3.55*\hY+4.5){E}
\draw[->, very thick] (0,4.5) -- (E);

\node at (15,3.8) {$\begin{aligned}
\pL{i}{j}{k}&:=\left\{(a,b,c)\middle|{\scriptscriptstyle\begin{aligned}e^{-i-1}&\leq \hat p_{ab}< e^{-i}\\ e^{-j-1}&\leq \hat p_{bc}< e^{-j}\\e^{-k-1}&\leq \hat p_b <e^{-k}\end{aligned}}\right\}
\\
\pL{i}{k_C}{k}&:=\left\{(a,b,c)\middle|{\scriptscriptstyle\begin{aligned}e^{-i-1}&\leq \hat p_{ab}< e^{-i}
\\ 0&\leq \hat p_{bc}< e^{-k_C}\\e^{-k-1}&\leq \hat p_b <e^{-k}\end{aligned}}\right\}
\\
\pL{k_A}{j}{k}&:=\left\{(a,b,c)\middle|{\displaystyle\begin{aligned}0&\leq \hat p_{ab}< e^{-k_A} \\ e^{-j-1}&\leq \hat p_{bc}< e^{-j}\\e^{-k-1}&\leq \hat p_b <e^{-k}\end{aligned}}\right\}
\\
\pL{k_A}{k_C}{k}&:=\left\{(a,b,c)\middle|{\textstyle\begin{aligned} ~0&\leq \hat p_{ab}< e^{-k_A} \\  ~0&\leq \hat p_{bc}< e^{-k_C}\\e^{-k-1}&\leq \hat p_b <e^{-k}\end{aligned}}\right\}
\\
S&:=\left\{(a,b,c)\middle| 0\leq \hat p_b\leq e^{-k_B}\right\}
\end{aligned}$};

\end{tikzpicture}
    }
    \caption{\label{fig:cmi_partition_intro} Partition of $d_A\times d_B\times d_C$ based on $\hat P_{AB}$, $\hat P_{BC}$, and $\hat P_{B}$. Indices $(a,b,c)$ of similar weight $\hat p_{ab}\hat p_{bc}/\hat p_b$ are grouped together in categories $L_{ij}^k$, which are used to perform piecewise equivalence testing with $P_{ABC}$. The axes are labeled according to the corresponding category, which is inverse logarithmic to the weight of the probabilities. The color of the categories indicate a different analysis of the sample complexity of testing the categories. The red and orange regimes dominate the sample complexity. The small regime is treated separately.}
        \end{center}
    \end{figure}
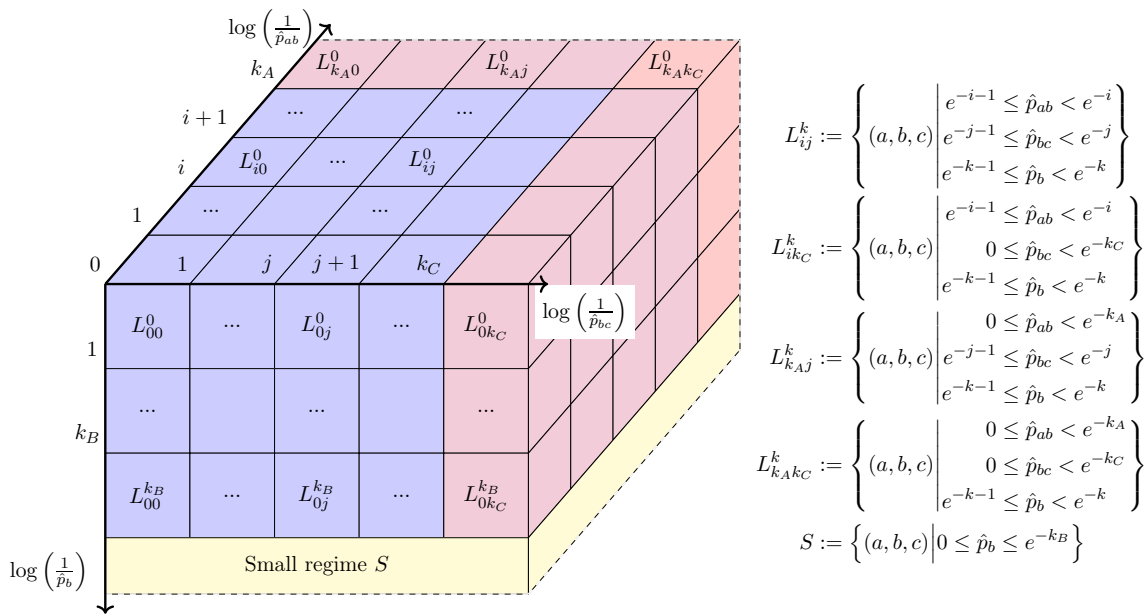

    \item Unlike the previous case, in the `\emph{small}' regime,  
    we can no longer guarantee how often a \emph{particular} rare $b$-value will appear. However, we know that if we take enough samples, \emph{some} of the rare elements will appear more than once, and the \emph{overall} number of rare collisions can be quantified using concentration bounds. This motivates us to take enough samples to obtain multiple collisions with rare elements, and use them to (approximately) simulate samples from $P_{AB}P_{C|B}$ by combining $(a,b,c)$, $(a',b,c')$ into $(a,b,c')$. The procedure is described in \Cref{fig:cmi_sampling}, and comes with a few caveats when compared to the approach used in the large regime, both due to the conditioning on seeing the respective $b$ twice. 

    First, the resulting samples are not independent, but slightly correlated. As a result, we can no longer use the same technique for equivalence testing which we used in the large regime. Instead, we introduce a novel tester for equivalence testing which is robust with respect to the correlations induced by our sampling approach. This is one of the technical novelties in this work, and we believe it will be of independent interest to the community, which we describe in more detail below. Second, the frequency at which we see a particular $b$ will be biased. This is due to the fact that a sample with a rare $b$ element appears with probability $p_b$, but a collision appears only with a probability proportional to $p_b^2$. Unsurprisingly, skewing $P_{B}$ will affect the statistics of the resulting samples, which makes testing for conditional independence more costly, as an improved precision is required. However, the conditional probabilities of witnessing a specific $(a,c)$ given $b$, $P_{AC|B=b}$, remain unchanged.
    
    One might try to avoid using a case distinction, and extend the first regime to cover all $p_b$ which might be relevant in testing for conditional mutual information. However, it is easy to construct examples where this would necessitate the approximate learning of probabilities $p_b$ in $O(\eps/d_B)$, which would require $O(d_B/\eps)$ samples and is as costly as learning $P_B$. Our case distinction allows us to achieve a sample complexity with sublinear scaling in $d_B$. 
\end{itemize}

\begin{figure}[H]
        \begin{center}
\scalebox{0.75}{
\begin{tikzpicture}

\def\binX{0.75}

\node at (2, 4.6) {\Large Large Regime};

\draw [draw=orange,fill=orange, fill opacity = 0.5, draw opacity = 0.5] (-2.5, -1*\binX-1) rectangle (6.5, -1*\binX+0.4);

\draw [draw=lime,fill=lime, fill opacity = 0.5, draw opacity = 0.5] (-2.5, 5*\binX-0.4) rectangle (6.5, 5*\binX+0.4);

\foreach \j in {0,1,2}{
    \def\spx{3.5*\j*\binX}
    \def\spy{0}
    \tkzDefPoint(\spx-0.8,\spy-0.2){A} 
    \tkzDefPoint(\spx+\binX+0.2,\spy-0.2){B};
    \tkzDefPoint(\spx+\binX+0.2,\spy+4*\binX+0.2){C};
    \tkzDefPointWith[colinear= at C](B,A) \tkzGetPoint{D};

    \tkzDrawPolygon[fill=blue!10](A,B,C,D)

    \ifthenelse{\j = 0}{
        \node at (\spx-0.35, \spy+2*\binX) {$Q_{1}$};
    }{
        \node at (\spx+0.2-2*\binX, \spy+2*\binX) {$...$};
        
        \ifthenelse{\j = 1}{
            \node at (\spx-0.35, \spy+2*\binX) {{\color{red}$Q_{i}$}};
        }{
            \node at (\spx-0.35, \spy+2*\binX) {$Q_{d_B}$};
        }   
    }

\foreach \i in {3,...,0}
{
    \tkzDefPoint(\spx,\spy+\i*\binX){A} 
    \tkzDefPoint(\spx,\spy+\i*\binX+\binX){B};
    \tkzDefPoint(\spx+\binX,\spy+\i*\binX+\binX){C};
    \tkzDefPointWith[colinear= at C](B,A) \tkzGetPoint{D};

    \tkzDrawPolygon[fill=white!20](A,B,C,D)
    \ifthenelse{\j = 0}{
        \ifthenelse{\i = 0}{
            \node at (\spx+0.5*\binX, \spy+0.5*\binX) {$a_y$};
        }{\ifthenelse{\i=1}{
            \node at (\spx+0.5*\binX, \spy+1.5*\binX) {$a_x$};
        }{}
        }
    }{
        \ifthenelse{\j = 1}{
            \ifthenelse{\i=0}{
            \node at (\spx+0.5*\binX, \spy+0.5*\binX) {${\color{blue}a_z}$};
            \node at (\spx+0.5*\binX, \spy+1.5*\binX) {${\color{violet}a_u}$};
            \draw[->, thick, violet] (\spx+0.5*\binX, \spy+4.8*\binX) -- (\spx+0.5*\binX, \spy+1.75*\binX);
            \draw[->, thick, blue] (\spx+0.5*\binX, \spy+0.25*\binX) -- (\spx+0.5*\binX, \spy-0.75*\binX);
            }{}
        }{
            \ifthenelse{\i = 0}{
            \node at (\spx+0.5*\binX, \spy+0.5*\binX) {$a_w$};
            }{}
        }
    } 
}
}

\node at (-1.7, 5*\binX+0.6) {\bf Phase 1};
\node at (-1.7, -1*\binX+0.6) {\bf Phase 2};

\node[anchor=west] at (-2.5, 5*\binX) {$\forall(a_u,b_i,c_v)\in S_1$};
\draw[->] (0.5, 5*\binX) -- (2.5, 5*\binX);
\node at (3.5, 5*\binX) {$({\color{violet}a_u},{\color{red}b_i},c_v)$};

\node[anchor=west] at (-2.5, -1*\binX) {$\forall (a_r,b_i,c_s)\in S_2|b_i\in B_L$};
\node[anchor=west] at (-2.5, -1.8*\binX) {$\forall (a_r,b_i,c_s)\in S_2|b_i\notin B_L$};
\draw[->] (2, -1*\binX) -- (2.5, -1*\binX);
\draw[->] (4.5, -1*\binX) -- (6, -1*\binX) -- (6, -1.35*\binX);
\draw[->] (2, -1.8*\binX) -- (5.5, -1.8*\binX);
\node at (3.5, -1*\binX) {$({\color{blue}a_z},{\color{red}b_i},c_s)$};
\node at (6, -1.8*\binX) {$S_{\text{out}}$};

\tikzset{shift={(10,0)}}
\draw (-3,-1.5) -- (-3,4.5);

\node at (2, 4.6) {\Large Small Regime};

\def\binX{0.75}

\draw [draw=yellow,fill=yellow] (-2.5, -1*\binX-1) rectangle (6.5, -1*\binX-0.2);

\draw [draw=yellow,fill=yellow] (-2.5, 5*\binX-1.1) rectangle (6.5, 5*\binX+0.4);

\foreach \j in {0,1,2}{
    \def\spx{3.5*\j*\binX}
    \def\spy{0}
    \tkzDefPoint(\spx-\binX-0.2,\spy-0.2){A} 
    \tkzDefPoint(\spx+\binX+0.2,\spy-0.2){B};
    \tkzDefPoint(\spx+\binX+0.2,\spy+2*\binX+0.8){C};
    \tkzDefPointWith[colinear= at C](B,A) \tkzGetPoint{D};

    \tkzDrawPolygon[fill=blue!10](A,B,C,D)

    \ifthenelse{\j = 0}{
        \node at (\spx-0.4, \spy+2.55*\binX) {$T_{1}$};
    }{
        \node at (\spx+0.15-2*\binX, \spy+2*\binX) {$...$};
        
        \ifthenelse{\j = 1}{
            \node at (\spx-0.4, \spy+2.55*\binX) {{\color{red}$T_{i}$}};
        }{
            \node at (\spx-0.4, \spy+2.55*\binX) {$T_{d_B}$};
        }   
    }

\foreach \i in {1,...,0}
{
    \tkzDefPoint(\spx-\binX,\spy+\i*\binX){A} 
    \tkzDefPoint(\spx-\binX,\spy+\i*\binX+\binX){B};
    \tkzDefPoint(\spx+\binX,\spy+\i*\binX+\binX){C};
    \tkzDefPointWith[colinear= at C](B,A) \tkzGetPoint{D};

    \tkzDrawPolygon[fill=white!20](A,B,C,D)
    \ifthenelse{\j = 0}{
        \ifthenelse{\i = 0}{
            \node at (\spx+0*\binX, \spy+0.5*\binX) {$a_y~~~c_m$};
        }{
        }
    }{
        \ifthenelse{\j = 1}{
            \ifthenelse{\i=0}{
            \node at (\spx, \spy+0.5*\binX) {${\color{black}a_z~~~c_s}$};
            \node at (\spx, \spy+1.5*\binX) {${\color{violet}a_u~~~c_v}$};
            \draw[->, thick, violet] (\spx, \spy+3.8*\binX) -- (\spx, \spy+1.75*\binX);
            \draw[->, thick, blue] (\spx, \spy-0.1*\binX) -- (\spx, \spy-1.55*\binX);

            \tkzDefPoint(\spx-\binX-0.1,\spy-0.1){A} 
            \tkzDefPoint(\spx+\binX+0.1,\spy-0.1){B};
            \tkzDefPoint(\spx+\binX+0.1,\spy+2*\binX+0.1){C};
            \tkzDefPointWith[colinear= at C](B,A) \tkzGetPoint{D};
            \tkzDrawPolygon[very thick, dashed, blue](A,B,C,D)
            }{}
        }{}
    } 
}
}

\node[anchor=west] at (-2.5, 5*\binX) {$\forall(a_u,b_i,c_v)\in S|b_i\notin B_S$};
\node[anchor=west] at (-2.5, 5*\binX-0.7) {$\forall(a_u,b_i,c_v)\in S|b_i\in B_S$:};
\draw[->] (2, 5*\binX) -- (4.5, 5*\binX);
\node at (5.5, 5*\binX) {discard};
\node at (2.6, 5*\binX-0.7) {$({\color{violet}a_u},{\color{red}b_i}, {\color{violet}c_v})$};

\draw[->] (4, -1.8*\binX) -- (5.5, -1.8*\binX);
\node at (2.7, -1.8*\binX) {$({\color{blue}a_u},{\color{red}b_{i}},{\color{blue}c_s})$};
\node at (6, -1.8*\binX) {$S_{\text{out}}$};

\end{tikzpicture}
}
\caption{\label{fig:cmi_sampling}(Comparison of sampling methods). In the large regime ($b\in B_L$), sampling is split into two phases using separate sets $S_1$ and $S_2$ of samples. In the first phase, we assign $A$-coordinates of samples to queues $Q_i$, depending on the $B$-coordinate. In the second phase, the actual output is generated by taking samples and replacing the $A$-coordinate with an element from the respective $B$-queue, if $b_i\in B_L$. The samples in $S_{\text{out}}$ can be seen as drawn i.i.d.\ from a distribution which coincides with $P_{AB}P_{C|B}$.
In the small regime ($b\in B_S$), we process only samples from $B_S$, which we sort into tuples $T_i$, dependent on the $B$-coordinate. As soon as a tuple is filled, we generate an output. This way of sampling skews the statistics of $B$ and the resulting samples are not independent.}
    \end{center}
\end{figure}
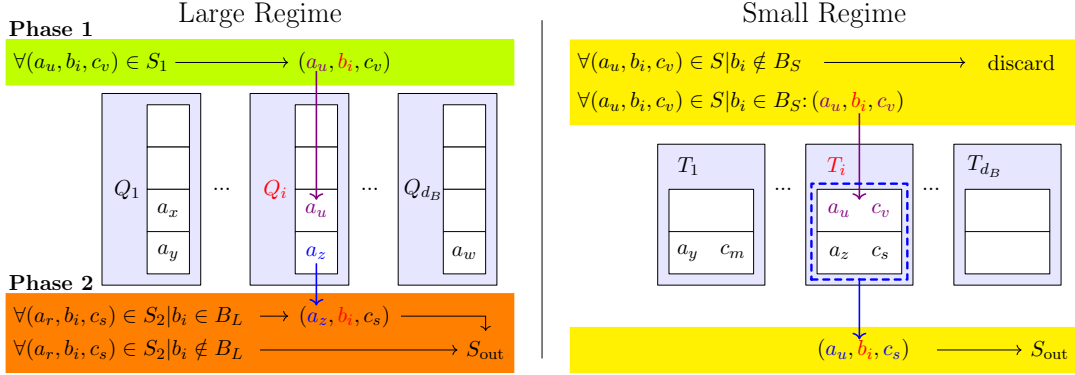

In general, our new tester for equivalence testing takes two pairs of independent multisets of samples $X, X'$ and $Y, Y'$ generated by two random processes $\mathcal{X}$ and $\mathcal{Y}$, respectively, which in particular might cause correlations within the respective multisets. Let $X_i$ denote the number of times the element $i \in A\times B\times C$ has appeared in set $X$, and define $Y_i,X_i'$ and $Y_i'$ analogously. Our estimator then constructs the following variable $Z$, which is compared against a threshold.
\begin{equation}\label{eqn:estimator_intro}
    Z=\sum_{i}Z_{i},\quad Z_{i}:=X_{i}X'_{i}-2X_{i}Y_{i}+Y_{i}Y'_{i}.    
\end{equation}
A crucial property of this estimator is that if $X$ and $Y$ follow the same statistics, then $\mathbb{E}[Z]=\sum_i(\mathbb{E}[X_i]-\mathbb{E}[Y_i])^2=0$. This property is very useful, as it allows us to build a tester by bounding $\gamma\leq \mathbb{E}[Z]$ for the case where $\mathcal{X}$ and $\mathcal{Y}$ differ, and defining the threshold at $\gamma/2$. The sample complexity is determined by the variance of $Z$. In our application in the small regime, a detailed analysis of the specific correlations between the samples is necessary to tightly bound the variance of $Z$. 

This tester can also be used for general equivalence testing, in which case $\mathcal{X}$ and $\mathcal{Y}$ correspond to sampling from distributions $P$ and $Q$. This is discussed in \Cref{sec:equiv_testing_general}, recovering the best known bounds in $\ell_1$ distance, \cite[Lemma 2.3]{diakonikolas_new_2016}, and sample optimal bounds in $\ell_2$ distance \cite[Theorem\ 2]{chan_optimal_2013}. To the best of our knowledge, this estimator is not present in the current literature and will be of independent interest.

\paragraph*{Lower Bounds:}
The primary idea behind proving lower bounds for MI testing is to use Le Cam's two-point method. We define two sets of distributions over distributions $S_{\mathsf{yes}}$ and $S_{\mathsf{no}}$ such that the distributions in $S_{\mathsf{yes}}$ are independent, and the distributions in $S_{\mathsf{no}}$ are ``\emph{far}'' from being independent in the squared Hellinger distance. Moreover, the distributions are designed such that given a random distribution from $S_{\text{yes}}\cup S_{\text{no}}$, deciding whether it came from $S_{\text{yes}}$ or $ S_{\text{no}}$ is difficult. Lower bounds to this task imply bounds on MI testing, since an MI tester could be used to solve this problem. 

To obtain a bound on the sample complexity, we follow the mutual information route, introduced in \cite{diakonikolas_new_2016}. In a nutshell, the goal is to prove that the mutual information between a set of samples obtained from the unknown distribution $P$ (chosen randomly from either $S_{\text{yes}}$ or $ S_{\text{no}}$), and the distribution itself is small. As a consequence, unless a sufficient number of samples is drawn, it is difficult to distinguish whether $P$ was drawn from either $S_{\mathsf{yes}}$ or $S_{\mathsf{no}}$. To connect the mutual information between the set of samples, we rely on the following inequality:
\begin{align}
\label{eq:mi_bound_proof_intro}
    2I(X:A)&\leq  \sum_{a\in S_A}\frac{(\Pr[A=a|X=0]-\Pr[A=a|X=1])^2}{\Pr[A=a|X=0]+\Pr[A=a|X=1]}\leq 12I(X:A).
\end{align}
where $X$ is a uniform random bit, and $A$ is a random variable taking values in a set $S_A$. The first inequality was known (e.g.\  \cite{diakonikolas_new_2016}), but to our knowledge, the second inequality does not appear in the literature in this form, and might be of independent interest. The specific construction we use extends distributions utilized in \cite{diakonikolas_new_2016} for proving lower bounds on equivalence testing. There, each index of the distribution is assigned one of two types, with one being identical for both instances, and thus not carrying information (this, in a sense, acts as noise which makes distinguishing them harder). We keep this structure in the dimension $d_A$, and simply expand it in an additional direction $d_C$, such that, for a fixed $a$, all $(a,c)$ belong to the same type. The asymmetry of $d_A$ and $d_C$ in the sample complexity is thus also reflected in the construction leading to our lower bounds.

Our lower bound result for CMI testing, presented in \Cref{res:cmi-lower}, is a reduction to our approach of independence testing described above.

\section{Preliminaries}\label{sec:prelim}
In this work, we will be using the following notations. We denote the set  $\{1, \ldots,n\}$ as $[n]$, and $[n]\cup \{0\}$ as $[n]_0$. Similarly, let $\mathbb{N}_0$ denote the set $\mathbb{N} \cup \{0\}$. For a positive integer $t \in \mathbb{N}^+$, and a set $\mathcal{S}\subset \mathbb{N}^+_0$, a vector $a \in \mathcal{S}^{\times t}$ denotes an ordered list $(a_1, \ldots, a_t)$ where for every $i \in [t]$, $a_i \in \mathcal{S}$. For a vector $v \in \mathcal{S}^{\times t}$, $\|v\|_1:=\sum_{i \in [t]} v_i$ denotes the sum of the entries in $v$. For a parameter $p \in (0,1)$, $\textnormal{Ber}(p)$ denotes the Bernoulli distribution over $\{0,1\}$, where for a random variable $X$ drawn from $\textnormal{Ber}(p)$, it holds that $\Pr[X=1]=p$.

Throughout this work, $P_{ABC}$ denotes an unknown discrete distribution defined over $A \times B \times C$, and we denote $|A|=d_A$, $|B|=d_B$, $|C|=d_C$. When studying conditional independence, we denote the conditioning system by $B$. Due to the symmetry of the problem with respect to the non-conditioning systems, we may assume $d_A\geq d_C$ without loss of generality. For a set $\mathcal{S}$, $U_\mathcal{S}$ denotes the uniform probability distribution over $\mathcal{S}$. For a probability distribution $P_{\cD}$ defined on a set $\cD$, we will often interchangeably denote $p_x:=P_{\cD}(x)$ for some $x \in \cD$. 
By stating i.i.d.\ samples from a distribution $P_{\cD}$, we mean independently and identically distributed samples from $P_{\cD}$. We also introduce the following notation for a probability distribution reduced to a subset $\mathcal{S}$ of its domain $\cD$:
\begin{equation}
    \subD{P}{\mathcal{S}}{\cD}(x):=
    \begin{cases}
        P_{\cD}(x) & \text{if }x\in \mathcal{S},
        \\
        0 & \text{otherwise.}
    \end{cases}
\end{equation}
The weight of $P_{\cD}$ on $\mathcal{S}$ is denoted by $P_{\cD}[\mathcal{S}]:=\sum_{i\in \mathcal{S}}p_i$. Note in particular that $\subD{P}{\mathcal{S}}{\cD}$ is in general not normalized, i.e., $\|\subD{P}{\mathcal{S}}{\cD}\|_1=P_{\cD}[\mathcal{S}]$.

For concise expressions and readability, we use the asymptotic complexity notion of $\widetilde{O}(\cdot), \widetilde{\Omega}(\cdot)$, and $\widetilde{\Theta}(\cdot)$, where we hide poly-logarithmic dependencies on the parameters. Throughout this work the logarithm $\log(\cdot)$ will denote the natural logarithm, and $\sinh(\cdot)$ and $\cosh(\cdot)$ denote the sine and cosine hyperbolic functions, respectively.

Next we define the different distance measures we will use to relate distributions. Let $\mathcal{P}(\cD)$ be the set of probability distributions over the elements of a set $\cD$.  Let $P$ and $Q$ be two such distributions over $\cD$ (we drop the $\cD$ in the subscript when it is clear from the context):
\begin{itemize}
    \item The {\em $\mathrm{KL}$-divergence} between $P$ and $Q$ is defined as:
$\kl(P\|Q) = \sum_{x \in \cD} P(x) \log \frac{P(x)}{Q(x)}$. We assume that the support of $Q$ contains the support of $P$, and use the convention $0 \log 0 = 0$ to deal with zeros, so that the KL-divergence is always finite.

\item  The {\em squared Hellinger} distance between $P$ and $Q$ is defined as:
\begin{equation}
    D_H^2(P,Q)= \frac{1}{2} \sum_{x \in \cD} \left(\sqrt{P(x)}- \sqrt{Q(x)}\right)^2.
\end{equation}

\item The \emph{$\ell_p$ distance} between $P$ and $Q$ is defined as
$\|P-Q\|_p$, where the \emph{$\ell_p$-norm} is given as $\| A \|_p := \left(\sum_{x \in \cD} | A(x) |^p\right)^{1/p}$ for any $p \geq 1$.
\end{itemize}
Let us next state a lemma which connects the KL-divergence and the squared Hellinger distance, which will be crucial for our work.

\begin{lemma}[{\cite[p.\ 429]{bounding_prob_metrics} \& \cite[Proposition 2.12]{flammia_quantum_2023}}]
\label{lemma:kl_to_hell_flaOd}
Let $P$ and $Q$ be two probability distributions over $[d]$. Then we have
\begin{equation}
    D_H^2(P,Q)\leq\kl(P\|Q)\leq \left(2+\log\left(\max_{i\in[d],P(i)\neq 0}\frac{P(i)}{Q(i)}\right)\right)D_H^2(P,Q).
\end{equation}
\end{lemma}
Note that we can simply bound this further to $\kl(P\|Q)\leq 3\log(1/Q_{\min})D_H^2(P,Q)$, where $Q_{\min}:=\min_{i\in [d]}Q(i)$ if we assume $d\geq 3$. We also make use of the following folklore relations between different distance measures (see e.g.\, \cite[Eq.\ 8]{bounding_prob_metrics} for $(i)$, and $(ii)$ follows from Jensen's inequality).
\begin{fact}
\label{fact:relations_distances}
    For arbitrary distributions $P$ and $Q$ of dimension $d$, it holds that
    \begin{enumerate}[label=(\roman*)]
        \item $\frac{1}{2}D_H^{2}(P,Q)\leq \|P-Q\|_1\leq D_H(P,Q)$,
        \item $\|P-Q\|_1\leq \sqrt{d}\|P-Q\|_2$.
    \end{enumerate}
\end{fact}

In our analysis, we will also use different bounds on the $\ell_2$ norm of a distribution.

\begin{fact}\label{fact:prelim:l2}
Let $Q$ be an arbitrary probability distribution defined over $[N]$ and $\mathcal{S} \subseteq [N]$ be a subset of the support. Then, we have
\begin{eqnarray*}
\|\subD{Q}{\mathcal{S}}{}\|_2^2 &\leq& 
\begin{cases}
    |\mathcal{S}|\max\limits_{i\in \mathcal{S}} q_i^2 & [\text{Case 1}],
    \\ \max\limits_{i\in \mathcal{S}} q_i & [\text{Case 2}].    
\end{cases}
\end{eqnarray*}
\end{fact}

We will use the following two concentration bounds (see, e.g., {\cite[Thm.\ 1.1]{DubhashiP09}} and \cite[Thm.\ 3.3, 4.1 \& 4.2]{Motwani_Raghavan_1995}).

\begin{lemma}[Chernoff Bound]\label{lem:chernoff2}
Let $X_1, \ldots, X_n$ be independent random variables such that $X_i \in [0,1]$. For $X=\sum\limits_{i=1}^n X_i$, the following holds for any $0\leq \delta \leq 1$.
\begin{equation}
    \Pr[X\geq (1+\delta)\mathbb{E}[X]]\leq e^{-\delta^2\mathbb{E}[X]/3},\quad \Pr[X\leq (1-\delta)\E[X]]\leq e^{-\delta^2\E[X]/2}.
\end{equation}
\end{lemma}

\begin{lemma}[Chebyshev's inequality]\label{lem:chebyshev}
Let $X$ be a random variable with $\E[X^2] < \infty$. The following holds for any $t>0$.
\begin{equation}
\Pr[|X- \mathbb{E}[X]| \geq t]\leq \frac{\Var[X]}{t^2}.\end{equation}
\end{lemma}

Besides concentration bounds, another crucial ingredient of our techniques is the notion of Poissonization technique. Let us start with the definition of Poisson random variables. 

\begin{definition}[Poisson random variable]
Let $\lambda >0$. A discrete random variable $X$ is said to be a Poisson random variable with parameter $\lambda$, denoted as $\mathsf{Poi}(\lambda)$ if the following holds:
\begin{equation}
\forall k \in \mathbb{N}, \Pr[X=k]=e^{-\lambda} \frac{\lambda^k}{k!}.
\end{equation}
\end{definition}

Now we are ready to describe the Poissonization technique (see, e.g., {\cite[Fact D.10]{canonne2020survey}}). 

\begin{lemma}[Poissonization technique]\label{lem:prelim:poisson}
Let $\Omega$ be a discrete domain, $D \in \Delta(\Omega)$ be a distribution and $m \in \mathbb{N}$.  Suppose $M' \sim \mathsf{Poi}(m)$ independent samples $s_1, \ldots, s_{M'}$ have been obtained from $D$. Suppose $X_t$ denotes the number of times $t \in \Omega$ appears among the samples $s_1, \ldots, s_{M'}$. Then the following hold:
\begin{enumerate}[label=(\roman*)]
    \item $(X_t)_{t\in \Omega}$ are independent.

    \item $X_t \sim \mathsf{Poi}(mD(t))$.
\end{enumerate}
\end{lemma}

Poisson random variables follow the following concentration bound (see, e.g., {\cite[Theorem A.8]{canonne2022topics}}).

\begin{lemma}[Poisson concentration bound]\label{lem:poissionconcentration}
Let $X$ be a $\mathsf{Poi}(\lambda)$ random variable for some $\lambda >0$. Then for any $t>0$, the following holds:
\begin{equation}
\Pr[|X-\lambda| \geq t] \leq 2 e^{-\frac{t^2}{2(\lambda +t)}}.
\end{equation}
\end{lemma}

Finally, we will use the following bounds on the exponential function in our analysis, proved in \Cref{app:prelim}.

\begin{restatable}{lemma}{boundexp}\label{lemma:bounds_exp}
    For $0\leq x<1$, it holds that 
    \begin{equation}
        1-x+\frac{x^2}{4}\leq e^{-x}\leq 1-\frac{x}{2}.
    \end{equation}
    Moreover, for any $x \geq 0$, we have that
    \begin{equation}
        1-x+\frac{x^2}{2}-\frac{x^3}{6}\leq e^{-x}\leq 1-x+\frac{x^2}{2}.
    \end{equation}
\end{restatable}

\section{Reducing CMI Testing to Testing in \texorpdfstring{$D_H^2$}{DH2}}\label{sec:kl_hellinger_connection}
In this section, we present how (conditional) mutual information testing can be reduced to the problem of (conditional) independence testing with respect to the squared Hellinger distance, which we formalize in the following two problems. The reduction itself is the subject of \Cref{theo:cmiheldistreduction}.

\begin{problem}[Conditional independence testing in $D_H^2$] \label{prob:CI_DH2}
    Consider the setting of conditional independence testing (\Cref{prob:cindep}), where $\Delta(\cdot\|\cdot)=D_H^2(\cdot,\cdot)$, that is, for fixed a threshold $\nu > 0$, we want to distinguish between the classes
    \begin{enumerate}[label=(\roman*)]
        \item $D_H^2(P_{ABC},P_{AB}P_{C|B}) = 0$, and
        \item $D_H^2(P_{ABC},P_{AB}P_{C|B}) \geq \nu$.
    \end{enumerate}
    We denote the sample complexity of this problem by $\textnormal{SC}_{\textnormal{CI},H}(\nu, d_A, d_B, d_C)$.
\end{problem}

We can also define the corresponding independence testing problem by choosing $B$ to be trivial.

\begin{problem}[Independence testing in $D_H^2$] \label{prob:I_DH2}
    Consider the setting of conditional independence testing (\Cref{prob:indep}), where $\Delta(\cdot\|\cdot)=D_H^2(\cdot,\cdot)$, that is, for fixed a threshold $\nu > 0$, we want to distinguish between the classes 
    \begin{enumerate}[label=(\roman*)]
        \item $D_H^2(P_{AC},P_{A}P_{C}) = 0$, and
        \item $D_H^2(P_{AC},P_{A}P_{C}) \geq \nu$.
    \end{enumerate}
    We denote the sample complexity of this problem by $\textnormal{SC}_{\textnormal{I},H}(\nu, d_A, d_C)$.
\end{problem}
The reduction from (conditional) independence testing to a testing problem in $D_H^2$ distance only requires a preprocessing of the samples. We can then apply any algorithm to solve \Cref{prob:CI_DH2}, by slightly increasing the precision $\nu$ to which we test. Thus, the sample complexity of CMI testing is, up to logarithmic factors, the same as the sample complexity of conditional independence testing in the squared Hellinger distance.

As we will show in the following, the key idea is to use the following inequality (see \Cref{lemma:kl_to_hell_flaOd}) which relates the KL-divergence and the squared Hellinger distance, with the choice $P=P_{AC|B}$ and $Q=P_{A|B}P_{C|B}$, 

\begin{equation}
    \kl(P\|Q)\leq \left(2+\log\left(\frac{1}{Q_{\min}}\right)\right)D_H^2(P,Q).
\end{equation}
Here, $Q_{\min}$ denotes the minimum probability mass of any element in the distribution $Q$. In order to use this inequality, we need to have a guarantee on the minimum probability mass of every element of  $Q_{AC|B}$. We ensure this by modifying our sampling approach: for both $A$ and $C$, we independently take a weak mixture with the uniform distribution. This action preserves conditional independence and ensures the minimum probability mass on $P_{AC|B=b}$. Our sampling is as follows:

\begin{algorithm}[H]
\LinesNumbered
\DontPrintSemicolon
\setcounter{AlgoLine}{0}
\caption{Sampling from $\tilde{P}$}
\label{alg:sampletildep}
\KwIn{A triplet $(a,b,c)$ from $A\times B\times C$, $\eta \in (0,1)$.}
\KwOut{A triplet from $A\times B\times C$}

$(a',b',c') \gets (a,b,c)$

Sample $X_A, X_C \overset{\$}{\gets} \textnormal{Ber}(\eta)$

\If{$X_A=1$}{
    $a' \overset{\$}{\gets} U_A$
}
\If{$X_C=1$}{
    $c' \overset{\$}{\gets} U_C$
}

\Return $(a',b',c')$. \

\end{algorithm}

It is immediate that the above algorithm preserves Markovianity since we can see $\tilde{P}_{ABC}$ as the result of an application of a stochastic map acting on $A$ and $C$, respectively. While our preprocessing ensures minimum probability mass in the conditional distributions, it might also change the conditional mutual information when it is non-zero. For this reason, we also need to guarantee that $I(A\!:\!C|B)_{\tilde P}\geq I(A\!:\!C|B)_{P}/4$. Our argument will use the following continuity result.
\begin{restatable}{lemma}{continuitylem}\label{lemma:mi:d_continuity}

Let $P_{AC}$ and $T_{AC}$ be two arbitrary distributions over $A \times C$ and $\eps \in (0,1)$ be a threshold. Then, for any $\alpha \leq \left(\frac{\eps}{48 d_Ad_C\log(d_Ad_C/\eps)}\right)^2$, it holds that 
    \begin{equation}
    \label{eq:kl_prod_continuity}
        \kl(P_{AC}\|P_A P_C)\geq \eps\implies D((1-\alpha)P_{AC}+\alpha T_{AC}\|(1-\alpha)P_AP_C+\alpha T_{AC})\geq \eps/2.
    \end{equation}
\end{restatable}
The proof is a direct calculation which bounds the difference between the two terms in \eqref{eq:kl_prod_continuity} using a case distinction and basic bounds on the logarithm. We prove this in \Cref{sec:reduction_app}.
We can now formally argue for the reduction.

\begin{theorem}\label{theo:cmiheldistreduction} 
CMI testing (\Cref{prob:cindep}) can be solved by preprocessing samples of $P$ using  \Cref{alg:sampletildep} with $\eta=\tilde O(\eps^2/(d_Ad_C)^2)$, and then applying any algorithm for \Cref{prob:CI_DH2} to precision $\nu = {\eps}/{(8 \log (d_Ad_C/\eta^2) )}$. In particular, we have
\begin{align}
    \textnormal{SC}_{\textnormal{CMI}}(\eps, d_A, d_B, d_C) \leq O\left(\textnormal{SC}_{\textnormal{CI},H}(\nu, d_A, d_B, d_C)\right).
\end{align}
Analogously, $\textnormal{SC}_{\textnormal{MI}}(\eps, d_A, d_C) \leq O(\textnormal{SC}_{\textnormal{I},H}(\nu, d_A, d_C))$.   
\end{theorem}

\begin{proof}[{Proof of \Cref{theo:cmiheldistreduction}}]
We will only argue for the conditional case, as the result for MI testing follows from the case when $B$ is trivial. We first note that $I(A\!:\!C|B)_P=0$ is equivalent to $D_H^2(\tilde P_{ABC}\|\tilde P_{AB}\tilde P_{C|B})=0$, since \Cref{alg:sampletildep} preserves Markovianity.

Now assume $I(A\!:\!C|B)_P \geq \eps$. The proof consists of two steps. We will first examine the effect of our preprocessing on $I(A\!:\!C|B)_P$ using \Cref{lemma:mi:d_continuity}, and then apply \Cref{lemma:kl_to_hell_flaOd} to arrive at a testing problem in $D_H^2$ distance. For brevity, let us denote $P_{AC}^b:=P_{AC \mid B=b}$, $P_A^b:=P_{A \mid B=b}$ and $P_C^b=P_{C|B=b}$. We first split $B$ into two sets, $B_+:=\{b \in B: \kl(P_{AC}^b \| P_A^b P_C^b) \geq \eps/2\}$ and $B_-:=B\setminus B_+$. Then
\begin{align}
    I(A\!:\!C|B)_P &= \sum_{b \in B} p_b \kl(P_{AC|B=b} \| P_{A| B=b} P_{C \mid B=b}) 
    \\
    &= \sum_{b\in B_+} p_b \kl(P_{AC}^b \| P_A^b P_C^b) + \sum_{b \in B_-} p_b \kl(P_{AC}^b \| P_A^b P_C^b)\label{eqn:reduction:cmi_3}
    \\
    &\leq \sum_{b\in B_+} p_b \kl(P_{AC}^b \| P_A^b P_C^b) + \frac{\eps}{2},
\end{align}
Thus, $I(A\!:\!C|B)_P \geq \eps$ implies
\begin{align}
    & \sum_{b \in B_+ } p_b \kl(P_{AC}^b \| P_A^b P_C^b) \geq \frac{\eps}{2} \label{eqn:reduction:cmi_4}.
\end{align}
For each $b\in B_+$, we would now like to individually apply \Cref{lemma:mi:d_continuity}. In order to do so, we first introduce a probability distribution $T_{ABC}$ over $A\times B\times C$, defined as
\begin{equation}
    t_{abc}:=\frac{p_b}{2-\eta}\left((1-\eta)\left(\frac{p_{c|b}}{d_A} + \frac{p_{a|b}}{d_C}\right) + \frac{\eta}{d_Ad_C}\right),~~~~~\text{for all } a \in A,b \in B, c \in C.
\end{equation}
Note that $T_{ABC}$ is a valid probability distribution, since $\forall (a,b,c) \in A\times B\times C: t_{abc}\geq 0$ and $\sum_{a,b,c}t_{abc}=1$. The conditioned distribution $T_{AC|B}$ has a guarantee on the minimum probability mass since $t_{ac|B=b}\geq \eta/((2-\eta)d_Ad_C)$. We can now rewrite $\tilde Q$ (obtained by applying \cref{alg:sampletildep} to the distribution $Q$) using $T_{ABC}$, 
\begin{align}
    \tilde{P}^b_{AC}&=(1-\eta)^2P_{AC}^b+\eta(1-\eta)\left(\frac{P_{C}^b}{d_A}+\frac{P_{A}^b}{d_C}\right) + \frac{\eta^2}{d_Ad_C}=(1-\eta)^2P_{AC}^b+\eta(2-\eta) T_{AC|B=b},
    \\
    \tilde P_{A}^b\tilde P_{C}^b&=\left((1-\eta)P_{A}^b+\frac{\eta}{d_A}\right)\left((1-\eta)P_{C}^b+ \frac{\eta}{d_C}\right)=(1-\eta)^2P_{A}^bP_{C}^b+\eta(2-\eta) T_{AC|B=b}.
\end{align}
Introducing $\nu=2\eta-\eta^2$ (note that $(1-\eta)^2+\eta(2-\eta)=1$), we can now rewrite 
\begin{align}
    \kl\left(\tilde P_{AC}^b\middle\|\tilde P_{A}^b\tilde P_{C}^b\right)
    &=\kl\left((1-\nu)P_{AC}^b+\nu T_{AC|B=b}\middle\|(1-\nu)P_{A}^bP_{C}^b+\nu T_{AC|B=b}\right).
\end{align}
We choose $\nu=\tilde O(\eps^2/(d_Ad_C)^2)$, and apply \Cref{lemma:mi:d_continuity}, implying
\begin{align}
\forall b\in B_+: D(P_{AC}^b\|P_A^bP_C^b)= \eps_b \implies D(\widetilde{P}_{AC}^b\|\widetilde{P}_A^b\widetilde{P}_C^b)\geq \frac{\eps_b}{2}, 
\end{align}
which, together with \eqref{eqn:reduction:cmi_4}, shows that
\begin{equation}
    I(A\!:\!C|B)_P\geq \eps\implies I(A\!:\!C|B)_{\tilde P}=\kl(\widetilde{P}_{ABC}\|\widetilde{P}_{AB}\widetilde{P}_{C\mid B})\geq \frac{\eps}{4}.
\end{equation}
The decision gap is reduced by only a constant factor, while we obtained a guarantee on the minimum probability mass of $P_{A|B}P_{C|B}$, since
\begin{equation}
    \forall a,b,c: \tilde P_A^b(a)\tilde P_C^b(c)\geq \eta(2-\eta)T_{AC|B=b}(a,c)\geq \frac{\eta^2}{d_Ad_C},
\end{equation}
We now want to use \Cref{lemma:kl_to_hell_flaOd} to move from KL-divergence to $D_H^2$ distance.
\begin{align}
    I(A\!:\!C|B)_{\tilde P}&=\sum_bp_b\kl(\widetilde{P}_{AC}^b\|\widetilde{P}_A^b\widetilde{P}_C^b)
    \\
    &\leq \sum_bp_b\left(2+\log\left(\frac{d_Ad_C}{\eta^2}\right)\right)D_H^2(\widetilde{P}_{AC}^b,\widetilde{P}_A^b\widetilde{P}_C^b)
    \\
    &\leq 2\log\left(\frac{d_Ad_C}{\eta^2}\right)D_H^2(\widetilde{P}_{ABC},\widetilde{P}_{AB}\widetilde{P}_{C|B}),
\end{align}
which completes our reduction, since
\begin{equation}
    I(A\!:\!C|B)_P\geq\eps \implies D_H^2(P_{ABC},P_{AB}P_{C|B})\geq \frac{\eps}{8\log\left(\frac{d_Ad_C}{\eta^2}\right)}.
\end{equation}
\end{proof}

\section{Independence Testing in \texorpdfstring{$D_H^2$}{DH2}}\label{sec:ind_test_hellinger}
In this section, we study the problem of independence testing of a distribution $P_{AC}$, with respect to the squared Hellinger distance. We will first study a related problem of testing the equivalence of distributions when one distribution is promised to be a product distribution. Our problem is formally defined as follows.

\begin{problem}[Sample complexity of equivalence testing for product distribution] \label{prob:eqivprod}
    Fix a threshold $\eps$, distance measure $\Delta$ and alphabet sizes $d_A$ and $d_C$.
    Consider a tester that, given access to $N$ samples from two unknown distributions $P_{AC}$ and $Q_{AC}$ each with $Q_{AC}$ being a product distribution, distinguishes between the classes
    \begin{enumerate}[label=(\roman*)]
        \item $P_{AC}=Q_{A}Q_C$ and
        \item $\Delta(P_{AC},Q_{A}Q_C) \geq \eps$.
    \end{enumerate}
    We denote the sample complexity of this problem by $\textnormal{SC}_{\textnormal{EqProd}}(\Delta,\eps, d_A, d_C)$.
\end{problem}

We first prove that this problem can be solved efficiently. Our result is described below.

\begin{restatable}{theorem}{indtestub}\label{theo:indtesthellinger}
Consider equivalence testing with product structure (\Cref{prob:eqivprod}). Assuming w.l.o.g.\ that $d_A\geq d_C$, we have
\begin{equation}
    \textnormal{SC}_{\textnormal{EqProd}}(D_H^2,\eps, d_A, d_C)=\widetilde{O}\left(\min\left\{\frac{d_A^{3/4}d_C^{1/4}}{\eps}, \frac{d_A^{2/3}d_C^{1/3}}{\eps^{4/3}}\right\}\right).
\end{equation}    
\end{restatable}

Note in particular that two samples from $P_{AC}$ can be used to simulate a sample from $P_{A}P_C$, by taking $(a_1,c_1)$, $(a_2,c_2)$ from $P_{AC}$ and returning $(a_1,c_2)$ as a sample from $P_A P_C$.

As a corollary to the above theorem, we have the following results, where the formalized upper bounds of \Cref{res:mi} come from combining \Cref{theo:indtesthellinger} with \Cref{theo:cmiheldistreduction} in \Cref{sec:kl_hellinger_connection}.

\begin{corollary}\label{cor:indtesthellinger_same}
Consider independence testing in $D_H^2$ as defined in \Cref{prob:I_DH2}. Assuming w.l.o.g.\ that $d_A\geq d_C$, we have
\begin{equation}
    \textnormal{SC}_{\textnormal{I},H}(\eps, d_A, d_C)= O\left(\textnormal{SC}_{\textnormal{EqProd}}(D_H^2,\eps, d_A, d_C)\right).
\end{equation}
Consider MI testing as defined in \Cref{prob:MI}. Assuming w.l.o.g.\ that $d_A\geq d_C$, we have
\begin{align}
    \textnormal{SC}_{\textnormal{MI}}(\eps, d_A, d_C) = \widetilde{O}\left(\min\left\{\frac{d_A^{3/4}d_C^{1/4}}{\eps}, \frac{d_A^{2/3}d_C^{1/3}}{\eps^{4/3}}\right\}\right) .
\end{align}    
\end{corollary}
It is clear that an upper bound for \Cref{prob:eqivprod}  implies an upper bound for \Cref{prob:I_DH2}, and vice versa for the lower bound. We derive upper bounds for the former, and lower bounds for the latter to show equivalence of the sample complexity up to logarithmic factors. For the connection to the mutual information, we note that for any distributions $P$ and $Q$, $D_H^2(P,Q)\leq D(P\|Q)$ such that lower bounds for testing with respect to the squared Hellinger distance carry over to mutual information testing as well (see \Cref{thm:mi_lower}).

\subsection{Some Preliminary Results}
\label{sec:mi_ub_prelim}

We first collect a few results which we will need along the way. Let us start with the result of \cite{chan_optimal_2013} about equivalence testing between distributions.

\begin{lemma}\textnormal{\cite[Theorem\ 2]{chan_optimal_2013}} 
\label{lemma:equivalence_l2_basic}
Let $P$ and $Q$ be two unknown distributions defined over $[d]$ such that $\|P\|_2$, $\|Q\|_2\leq b$ for some $b$. Then in order to distinguish between $\|P-Q\|_2\leq \eps$ and $\|P-Q\|_2 \geq 2\eps$ with probability at least $2/3$, $\Theta(b/\eps^2)$ samples are necessary and sufficient.
\end{lemma}

Now we state a result from \cite{batu_testing_2001} regarding estimating the $\ell_2$ norm of a distribution.

\begin{lemma} \textnormal{\cite[Theorem\ 12]{batu_testing_2001}}
\label{lemma:get_weight_l2_basic}
Let $P$ be an unknown distribution defined over $[d]$. Given $\eta \in (0,1)$, in order to estimate $\|P\|_2^2$ up to a multiplicative precision of $1 \pm \eta$ with probability at least $2/3$, $O(\sqrt{d}/\eta^2)$ i.i.d.\ samples from $P$ are sufficient.    
\end{lemma}

We need to generalize both the results stated above, since we will need analogous statements about subsets of distributions. Note that variants of both our \Cref{lemma:get_weight_l2} and \Cref{lemma:equivalence_l2} results already appeared implicitly in the proofs of \cite{diakonikolas_new_2016}.

\begin{lemma}
    \label{lemma:get_weight_l2}
    Let $P$ be an unknown distribution over $[d]$, $S \subseteq [d]$ and $\eps\in (0,1)$ be a threshold. Given i.i.d.\ sample access to $P$, in order to determine $c$ such that $\frac{ \|\subD{P}{S}{}\|_2}{2} \leq c \leq 2(\|\subD{P}{S}{}\|_2 + \eps)$ holds with probability at least $2/3$, $O(\sqrt{d}+ \frac{1}{\eps})$ samples from $P$ are sufficient.
\end{lemma}
\begin{proof}
Following \Cref{lemma:get_weight_l2_basic}, we know that estimating $\|P\|_2$ up to a multiplicative factor $2$ requires $O(\sqrt{d})$ samples from $P$. However, in our case, we are only interested in a subset of the support of $P$, and samples that are not coming from $S$ might not be directly useful to determine the weight of $\subD{P}{S}{}$.

For this reason we create a dummy distribution $P'$ using $P$, such that $\|\subD{P}{S}{}\|_2 \approx \|P'\|_2$, up to an additive constant $\eps$. In order to define $P'$, we create a set of indices $T= S \cup N$, with $|N|=O(1/\eps^2)$ new indices disjoint from $[d]$. 

The main idea is now to sample from $P'$ via a flattening technique. Let us take a sample $s$ from $P$. If $s \in S$, we keep it as it is. Otherwise, if $s \notin S$, we assign it one of the $N$ new indices, uniformly at random. Our new distribution $P'$ satisfies 
\begin{equation}
\label{eq:bound_l2_dummy}
\|P\|_2\leq\|P'\|_2\leq \|\subD{P}{S}{}\|_2+\sqrt{\sum_{i\notin S} \frac{1}{N^2}} \leq \|\subD{P}{S}{}\|_2 + \frac{1}{\sqrt{N}}=\|\subD{P}{S}{}\|_2 + O(\eps),
\end{equation}
where the last equality follows from the fact that $|N|=O(1/\eps^2)$.

Note that the distribution $P'$ is defined on $S \cup N$. Since $S \subseteq [d]$, the support size of $P'$ is $O(d + {1}/{\eps^2})$. Using \Cref{lemma:get_weight_l2_basic}, we can determine $\|P'\|_2$ up to a small multiplicative constant factor, say $1/2$, using $O(\sqrt{d + 1/\eps^2}) = O(\sqrt{d} + 1/\eps)$ samples from $P$, which directly implies our result. 
\end{proof}

Now let us prove an analogous statement to \Cref{lemma:equivalence_l2_basic} for subsets of distributions. We will first use a trick to reduce the constraint $\max\{\|P\|_2,\|Q\|_2\}\leq b$ to $\|Q\|_2\leq b$, which was used in the proof of \cite[Lemma 2.3]{diakonikolas_new_2016}.

\begin{lemma}
\label{lemma:equivalence_l2}
Let $P$ and $Q$ be two unknown distributions over $[d]$, and $S\subseteq [d]$. Moreover, let us assume that $\|\subD{Q}{S}{}\|_2 \leq b$ for some $b\in \mathbb{R}$. Given i.i.d.\ sample access to $P$ and $Q$, in order to distinguish if $\subD{P}{S}{}=\subD{Q}{S}{}$
or $\|\subD{P}{S}{}-\subD{Q}{S}{}\|_2 \geq \eps$ with probability at least $2/3$, $O(\max\{\frac{b}{\eps^2}, \frac{1}{\eps},\sqrt{d}\})$ samples are sufficient. 

\end{lemma}
\begin{proof} 
We will first use \Cref{lemma:get_weight_l2} to obtain a $c$ such that $\|\subD{P}{S}{}\|_2/2\leq c\leq 2\|\subD{P}{S}{}\|_2+\max\{\eps, b\}$ holds, which takes $O(\sqrt{d}+\min\{1/\eps,1/b\})$ samples. We return `Far' in case $[c-(b+\eps)]/2>b$, i.e., if our lower bound on $\|\subD{P}{S}{}\|_2$ is larger than $b$ (our upper bound on $\|\subD{Q}{S}{}\|_2$). If we did not reject, then $[c-(b+\eps)]/2\leq b$, which implies that $c\leq 3b+\eps$. This guarantees that $\max\{\|\subD{P}{S}{}\|_2,\|\subD{Q}{S}{}\|_2\}\leq 6b+2\eps$. 

We now perform equivalence testing between $\subD{P}{S}{}$ and $\subD{Q}{S}{}$. For this, we will create two dummy distributions $P'$ and $Q'$, consisting of the indices $S$ and a set $N$ of new indices, for $|N|=\Omega(1/\eps^2)$. Let us obtain a sample $s$ from $P$. If $s \in S$, we keep it as it is. Otherwise, if $s \notin S$, we will assign it one of the $N$ new indices, uniformly at random. We perform the analogous operation for $Q$ as well to obtain samples from $Q'$. Note that
\begin{align}
\label{eq:l2_dummy_equiv}
    \|P'-Q'\|_2^2&=\|\subD{P}{S}{}-\subD{Q}{S}{}\|_2^2+N\left(\frac{1-P[S]}{N}-\frac{1-Q[S]}{N}\right)^2
    \\
    &=\|\subD{P}{S}{}-\subD{Q}{S}{}\|_2^2+\frac{(P[S]-Q[S])^2}{N}.
\end{align}
If we choose $N\geq 4/\eps^2$, then
\begin{equation}
    \subD{P}{S}{}=\subD{Q}{S}{}\implies \|P'-Q'\|_2\leq \eps/2, \quad\text{and} \quad \|\subD{P}{S}{}-\subD{Q}{S}{}\|_2\geq \eps \implies \|P'-Q'\|_2\geq \eps.
\end{equation}
We can thus use robust equivalence testing between $P'$ and $Q'$ (\Cref{lemma:equivalence_l2_basic}) to test for equivalence between $\subD{P}{S}{}$ and $\subD{Q}{S}{}$. The precision to which we test is $\eps/2$, and it remains to upper bound $\max\{\|P'\|_2,\|Q'\|_2\}$. By construction, $\|P'\|_2\leq \|\subD{P}{S}{}\|_2+O(\eps)$, and analogous for $Q'$ (see \eqref{eq:bound_l2_dummy}). Since we showed $\max\{\|\subD{P}{S}{}\|_2,\|\subD{Q}{S}{}\|_2\}\leq 6\max\{b,\eps\}$, this results in a sample complexity of $O(\max\{b,\eps\}/\eps^2)$ using \Cref{lemma:equivalence_l2_basic}. Recall that, at the beginning, we use $O(\sqrt{d}+\min\{1/\eps,1/b\})$ samples to obtain a $c$ such that $\|\subD{P}{S}{}\|_2/2\leq c\leq 2\|\subD{P}{S}{}\|_2+\max\{\eps, b\}$ holds. Combining the above, we have the result.
\end{proof}

In order to boost the success probability to $1-\delta$ for any $\delta \in (0,1)$, we run the above algorithm $O(\log(1/\delta))$ times and return the majority answer. The analysis follows from a standard success probability amplification argument. This is formalized in the following corollary.

\begin{corollary}
\label{cor:equalence_l2_algo}
Given a $\delta \in (0,1)$, two multisets $\mathcal{S}_P$, $\mathcal{S}_Q$  of i.i.d.\ samples from $P$ and $Q$, and $S$, $b$ and $\eps$ as in \cref{lemma:equivalence_l2}. There exists an algorithm \textnormal{\algEquiv}$(\mathcal{S}_P,\mathcal{S}_Q, S, b, \eps,\delta)$ which performs equivalence testing between $Q^S$ and $P^S$ in $\ell_2$ distance, which succeeds with probability at least $1-\delta$ if 
\begin{itemize}
    \item $|\mathcal{S}_P|,|\mathcal{S}_Q|\geq\ceq\max\{\frac{b}{\eps^2}, \frac{1}{\eps},\sqrt{d}\}\log(1/\delta)$ in case $|S|<|D|$,
    \item $|\mathcal{S}_P|,|\mathcal{S}_Q|\geq \ceq\frac{b}{\eps^2}$ if $S=D$,
\end{itemize}
where $\ceq$ is an instance independent constant.
\end{corollary}
We will repeatedly use simple concentration bounds based on the Chernoff inequalities, which we prove here for completeness.

\begin{lemma}\label{lemma:learn_approx}
    Given a threshold $\tau$, a parameter $\zeta \in (0,1)$ and $N\geq 8\log(4d/\zeta)/\tau$ i.i.d.\ samples from an unknown distribution $P$ of dimension $d$, the empirical estimator $\hat P$ satisfies the following properties with probability at least $1 - \zeta$,
    \begin{enumerate}[label=(\roman*)]
        \item  $\forall i:p_i\geq \tau\implies1/2\leq \hat p_i/p_i\leq 2$,
        \item $\forall i: p_i\leq \tau\implies\hat p_i\leq 2\tau$.
    \end{enumerate}
where $p_i$ and $\hat p_i$ denote the probability mass of the element $i \in [d]$ in the distributions $P$ and $\hat P$, respectively.    
\end{lemma}

\begin{proof}
    This is a direct application of the Chernoff and union bounds. First, we assume $p_b\leq \tau$. We set $\delta$ such that $(1+\delta)p_b=2\tau$, and denote by $X$ how often $b$ was observed among the $N$ samples. Since $\delta\geq 1$, we have $\delta^2\E[X]\geq (1+\delta)\E[X]/2$ such that
    \begin{equation}
        \Pr[X\geq 2\tau]\leq e^{-\delta^2\mu/3}\leq e^{-\tau N/3}\leq \frac{\zeta}{4d}.
    \end{equation}
    For the other direction, where $p_b\geq \tau$, we set $\delta=1/2$ and $\delta=2$, respectively,
    \begin{align}
        \Pr[X\leq \E[X]/2]&\leq e^{-\E[X]/8}\leq e^{-\tau N/8}\leq \frac{\zeta}{4d},
        \quad\text{and}\quad\Pr[X\geq 2\E[X]]\leq e^{-\E[X]/3}\leq \frac{\zeta}{4d}.
    \end{align}
    Finally, we apply a union bound over all cases.
\end{proof}

\subsection{Our Independence Tester}

We now have all the preliminaries in place to prove the main result of this section:

\indtestub*

The proof of correctness follows from the discussion of \Cref{alg:indtesthellinger} (\Cref{lemma:indep_hellinger_correctness}) and the associated subroutines. Our algorithm takes two multisets $\mathcal{S}_P$ and $\mathcal{S}_Q$ of $\textnormal{SC}_{\textnormal{EqProd}}(D_H^2,\eps, d_A, d_C)$ samples from unknown distributions $P_{AC}$ and $Q_{AC}$, where $Q_{AC}=Q_AQ_C$ is guaranteed to have prodcut structure. With high probability, our algorithm outputs `{\bf Yes}' if $P_{AC}=Q_{AC}$ and `{\bf No}' if $D_H^2(P_{AC}, Q_{AC})\geq \eps$. As described in the overview of our methods in \Cref{sec:overview_method}, the underlying idea is to partition $A\times C$ into a polylogarithmic number of categories, denoted by $S_{ij}$, on which we then individually perform equivalence testing with respect to the $\ell_2$ distance. The partitioning used in \Cref{alg:indtesthellinger} is visualized in \Cref{fig:mi_grid}, where we also define the categories $S_{ij}$.

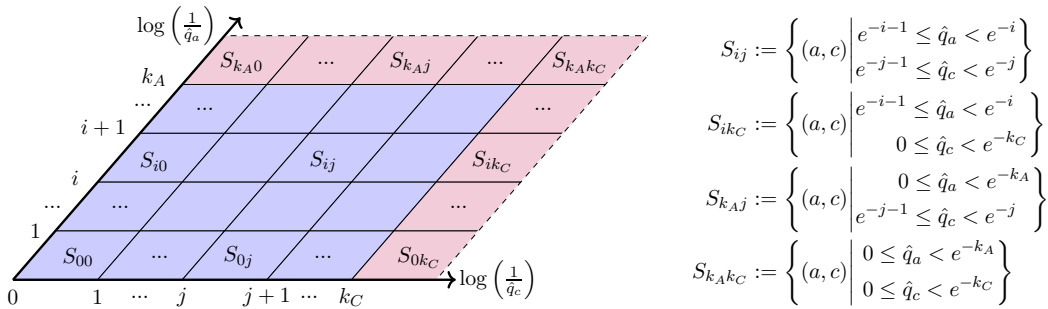
\begin{figure}[H]
        \begin{center}
\scalebox{0.75}{
\begin{tikzpicture}

\def\binX{1}
\def\tcol{blue}

\def\angleP{pi/3}
\def\sf{1.5}
\def\hX{\binX*cos(\angleP)} \def\hY{\binX*sin(\angleP)}
    
\foreach \j in {0,...,4}{
\foreach \i in {4,...,0}
{
    \def\spx{(\j*\binX+\hX*\i)}
    \def\spy{\hY*\i}
    \tkzDefPoint(\spx*\sf,\spy){A} 
    
    \tkzDefPoint((\spx+\binX)*\sf,\spy){B} 
    \tkzDefPoint((\spx+\binX+\hX)*\sf,\spy+\hY){C} 
    \tkzDefPointWith[colinear= at C](B,A) \tkzGetPoint{D}
    \ifthenelse{\i = 4}{
        \ifthenelse{\j = 4}{\def\tcol{purple}}{\def\tcol{purple}}
    }{
        \ifthenelse{\j = 4}{\def\tcol{purple}}{\def\tcol{blue}}
    } 
    \tkzDrawPolygon[fill=\tcol!20](A,B,C,D)

}
}

\tkzDefPoint(5*\binX*\sf,0){X};
\tkzDefPoint(5*\binX*\sf+5*\hX*\sf, 5*\hY){Y};
\draw[dashed, white, thick] (X) -- (Y);
\tkzDefPoint(5*\hX*\sf, 5*\hY){X};
\tkzDefPoint(5*\hX*\sf+5*\binX*\sf, 5*\hY){Y};
\draw[dashed, white, thick] (X) -- (Y);

\draw[->, very thick] (0,0) -- (5*\binX*\sf+0.35,0);
\tkzDefPoint(5*\hX*\sf*1.075,5*\hY*1.075){E}
\draw[->, very thick] (0,0) -- (E);

\node at (2.9, 4.5) {$\log\left(\frac{1}{\hat q_a}\right)$};
\node at (8.6, 0) {$\log\left(\frac{1}{\hat q_c}\right)$};

\tkzDefPoint(\hX-0.1, \hY){P};
\node at (P) {$1$};
\tkzDefPoint(1.5*\hX-0.05, 1.5*\hY){P};\node at (P) {$...$};
\tkzDefPoint(1*\hX+\binX*\sf-0.1, 1.5*\hY){P};\node at (P) {$...$};
\tkzDefPoint(1*\hX+5*\binX*\sf-0.1, 1.5*\hY){P};\node at (P) {$...$};
\tkzDefPoint(2*\hX-0.05, 2*\hY){P};
\node at (1.1, 1.8) {$i$};
\tkzDefPoint(3*\hX-0.05, 3*\hY){P};
\node at (1.6, 2.7) {$i+1$};
\node at (2.3, 3.1) {$...$};
\tkzDefPoint(4*\hX+\binX*\sf-0.1, 3.5*\hY){P};\node at (P) {$...$};
\tkzDefPoint(4*\hX+5*\binX*\sf-0.1, 3.5*\hY){P};\node at (P) {$...$};
\tkzDefPoint(4*\hX-0.05, 4*\hY){P};
\node at (2.5, 3.6) {$k_A$};

\node at (0, -0.3) {$0$};
\node at (\binX*\sf, -0.3) {$1$};
\node at (1.5*\binX*\sf, -0.3) {$...$};
\node at (2*\binX*\sf, -0.3) {$j$};
\node at (3*\binX*\sf, -0.3) {$j+1$};
\node at (3.5*\binX*\sf, -0.3) {$...$};
\node at (4*\binX*\sf, -0.3) {$k_C$};
\node at (5*\binX*\sf, -0.3) {$ $};

\node at (1.15, 0.4) {$S_{00}$};

\node at (4.05, 3.85) {$S_{k_A0}$};
\node at (5.55, 3.85) {$...$};
\node at (7.05, 3.85) {$S_{k_Aj}$};
\node at (8.55, 3.85) {$...$};
\node at (10.05, 3.85) {$S_{k_Ak_C}$};
\node at (2.6, 0.4) {$...$};
\node at (4, 0.4) {$S_{0j}$};
\node at (5.5, 0.4) {$...$};
\node at (7.2, 0.4) {$S_{0k_C}$};
\node at (5.5, 2.1) {$S_{ij}$};
\node at (2.5, 2.1) {$S_{i0}$};
\node at (8.5, 2.1) {$S_{ik_C}$};

\node at (15.2,2.1) {$\begin{aligned}
S_{ij}&:=\left\{(a,c)\middle|{\scriptscriptstyle\begin{aligned}e^{-i-1}&\leq \hat q_a< e^{-i}\\ e^{-j-1}&\leq \hat q_c< e^{-j}\end{aligned}}\right\}
\\
S_{ik_C}&:=\left\{(a,c)\middle|{\scriptscriptstyle\begin{aligned}e^{-i-1}&\leq \hat q_a< e^{-i}
\\ 0&\leq \hat q_c< e^{-k_C}\end{aligned}}\right\}
\\
S_{k_Aj}&:=\left\{(a,c)\middle|{\displaystyle\begin{aligned}0&\leq \hat q_a< e^{-k_A} \\ e^{-j-1}&\leq \hat q_c< e^{-j}\end{aligned}}\right\}
\\
S_{k_Ak_C}&:=\left\{(a,c)\middle|{\textstyle\begin{aligned} ~0&\leq \hat q_a< e^{-k_A} \\  ~0&\leq \hat q_c< e^{-k_C}\end{aligned}}\right\}
\end{aligned}$};

\end{tikzpicture}
}
\caption{\label{fig:mi_grid} Partition of $d_A\times d_C$ based on $\hat Q_A$ and $\hat Q_C$. Indices $(a,c)$ of similar weight $\hat q_a\hat q_c$ are grouped together in categories $S_{ij}$, which are used to perform piecewise equivalence testing with $P_{AC}$. The axes are labeled according to the corresponding category, which is inverse logarithmic to the weight of the probabilities. The color of the categories indicate a different analysis of the sample complexity of the categories. The red regimes dominate the sample complexity.}
    \end{center}
\end{figure}
In the following, we define some quantities which will be used in \Cref{alg:indtesthellinger}. Further, let $k_{AC}:=(\lceil\log(d_Ad_C/\eps)\rceil+1)^2$. 
\begin{align}
    \gamma(i,j)&:=\begin{cases}
        \sqrt{\frac{\eps e^{-(i+j+2)}}{e^6k_{AC}}} & \text{if } i<k_A \text{ and } j<k_C,
        \\
        \frac{\eps}{4k_{AC}\sqrt{|S_{ij}|}} & \text{otherwise},
    \end{cases}
    \\
    b(i,j)&:=\min\left\{\sqrt{|S_{ij}|}e^{-(i+j)+2},\sqrt{e^{-(i+j)+2}}\right\}.
\end{align}
\Cref{alg:indtesthellinger} will use
\begin{equation}
\label{eq:mi_sc_full}
    N:= 10^2k_{AC}\log(k_{AC}) \max\{N_{\text{heavy}}, N_{\text{mixed}} \}
\end{equation}
samples, where 
\begin{align}
N_{\text{heavy}}&:= 10^5k_{AC}\frac{\sqrt{d_A d_C}}{\eps},\quad (\textnormal{\Cref{cl:mi:ub:heavy}}),  
\\
 N_{\text{mixed}}&:= 10k_{AC}^2\min\left\{\frac{d_A^{3/4}d_C^{1/4}}{\eps},\frac{d_A^{2/3}d_C^{1/3}}{\eps^{4/3}}\right\},\quad (\textnormal{\Cref{cl:mi:ub:mixed}}),    \label{eq:mi_ub_mixed}
\end{align}
which gives the sample complexity stated in \Cref{theo:indtesthellinger}.

Finally, we define $k_A,k_C:=\lceil\log(N_{\text{mixed}})\rceil$. Note that $k_{AC}\geq (k_A+1)(k_C+1)$.

\begin{algorithm}[H]
\SetAlgoLined
\DontPrintSemicolon
\LinesNumbered
\setcounter{AlgoLine}{0}
\caption{Equivalence Product Testing in $D_H^2$}
\label{alg:indtesthellinger}
\KwIn{
Multisets $\mathcal{S}_P$ and $\mathcal{S}_Q$ of size $N$ of tuples in $A \times C$ each, $\eps \in (0,1)$ \Comment*[r]{$N$, see \eqref{eq:mi_sc_full}}}

\KwOut{`{\bf Yes}' or `{\bf No}'}

\Comment*[l]{{\textbf{Step $1$: Partition $A\times C$}}}

$M\leftarrow N_{\text{mixed}}$, $\delta\leftarrow \frac{1}{10^3k_{AC}}$ \Comment*[r]{$N_{\text{mixed}}$,~see \eqref{eq:mi_ub_mixed})}

$\mathcal{S}_A$, $\mathcal{S}_C$ $\leftarrow$ remove random multisets of size $8M \log(10^3d_Ad_C)$ each from $\mathcal{S}_Q$

$\hat Q_A\leftarrow$ empirical frequencies of $a\in A$ in $\mathcal{S}_{A}$

$\hat Q_{C}\leftarrow$ empirical frequencies of $c\in C$ in $\mathcal{S}_{C}$

$\forall i\in [k_A]_0,j\in [k_C]_0$: define $S_{ij}$ according to \Cref{fig:mi_grid} \Comment*[r]{partition $A\times C$}

\vspace{5pt}

\Comment*[l]{{\textbf{Step $2$: Equivalence Testing}}}

\For{\textnormal{all} $i,j$}{

    $\mathcal{S}_{p}$, $\mathcal{S}_q$ $\leftarrow$ remove random multisets of size $M$ each from $\mathcal{S}_P$ and $\mathcal{S}_{Q}$, respectively

\If{\textnormal{\algEquiv}$(\mathcal{S}_p, \mathcal{S}_q, S_{ij}, b(i,j), \gamma(i,j),\delta)=$ \textnormal{`{\bf Far}'}}{

\Return `{\bf No}'. \Comment*[r]{\Cref{cor:equalence_l2_algo}, distinction for heavy/mixed implicit in $\gamma$}

}

}
\Return `{\bf Yes}'.
\end{algorithm}

\begin{remark}
    We note that for $i<k_A$, $j<k_C$, we can group together all groups for which $i+j$ gives the same value. It is easy to see that this would reduce the number of equivalence tests from $O(\log^2(k_A))$ to $O(\log(k_A))$. An analogous strategy can be applied later in \Cref{sec:cmi_ub:large} as well.
\end{remark}

We will first approximately learn $P_A$ and $P_C$. This allows us to group indices $a \in A$ into $k_A+1$ many \emph{bins} $A_0, A_1,..., A_k$, such that the weights of indices in the same bin are close to each other up to a constant factor. The only exception are the indices with particularly small weights, which we group together in the last bin, $A_{k_A}$. We perform the same approach for $C$ as well. For a fixed bin $A_i$ and a fixed bin $C_j$, we can define $S_{ij}:=\{(a,c)|a\in A_i,c\in C_j\}$, which we call a \emph{category}. If $i<k_A$ and $j<k_C$, we call $S_{ij}$ a \emph{heavy category}. For any $(a,c),(a',c')$ from the same heavy category $S_{ij}$, the corresponding weights of our reference distribution, $p_ap_c$ and $p_{a'}p_{c'}$, will be close as well. The remaining categories are called \emph{light}, as the weights can be arbitrarily small. A guarantee analogous to heavy categories is not possible here. After partitioning the distribution in a two-dimensional grid, we perform equivalence testing for each of the $k_{AC}$ bins using \Cref{lemma:equivalence_l2}. In the process, we will need to bound terms of the form $\|\subD{P}{S_{ij}}{AC}\|_2$. There are two ways to do so (see \Cref{fact:prelim:l2}), and which option we choose depends on the parameters $d_A,d_C$, and $\eps$ of our problem. The final sample complexity will reflect this choice in terms of the minimization.

The sample complexity of testing a category $S_{ij}$ (or, equivalently, the precision to which we have to test) depends on the type of category.
\begin{enumerate}
    \item \textbf{Heavy categories:} This includes all categories where neither $a$ nor $c$ come from the last bin, such that for all $a,c$, both $p_a$ and $p_c$ are bounded away from zero. That is we will consider categories $S_{ij}$ where $i < k_A$ and $j < k_C$.
    Our result in this setting is as follows.

    \begin{restatable}{claim}{miheavysc}\label{cl:mi:ub:heavy}
        For any heavy category $S_{ij}$ as defined in \Cref{fig:mi_grid}, distinguishing whether $\subD{P}{S_{ij}}{ABC}=\subD{Q}{S_{ij}}{ABC}$ or $D_H^2(\subD{P}{S_{ij}}{ABC},$ $\subD{Q}{S_{ij}}{ABC}) \geq \eps/k_{AC}$  can be done using 
        \begin{equation}
         N_{\textnormal{heavy}}\geq 10^5k_{AC}\frac{\sqrt{d_A d_C}}{\eps}   
        \end{equation}
         samples of $P_{AC}$ and the product distribution $Q_AQ_C$ each. 
    \end{restatable}

    \item \textbf{Mixed categories:} This includes all categories where at least one of $a$ or $c$ come from the last bin, that is, categories $S_{ij}$ where $i= k_A$ or $j = k_C$, or both. 

    \begin{restatable}{claim}{mimixedsc}\label{cl:mi:ub:mixed}
        For any mixed category $S_{ij}$ as defined in \Cref{fig:mi_grid}, distinguishing whether $\subD{P}{S_{ij}}{ABC}=\subD{Q}{S_{ij}}{ABC}$ or $D_H^2(\subD{P}{S_{ij}}{ABC},\subD{Q}{S_{ij}}{ABC})\geq \eps/k_{AC}$  can be done using 
        \begin{equation}
            N_{\textnormal{mixed}}\geq 10k_{AC}^2\min\left\{\frac{d_A^{3/4}d_C^{1/4}}{\eps},\frac{d_A^{2/3}d_C^{1/3}}{\eps^{4/3}}\right\}
        \end{equation}
        samples of $P_{AC}$ and the product distribution $Q_AQ_C$ each. 
    \end{restatable}

\end{enumerate}

These two claims will be proven in the following subsections, \Cref{sec:mi_ub:heavy_regime} and \Cref{sec:mi_ub:mixed_regime}, respectively. Assuming these two claims hold, we can prove the correctness of \Cref{theo:indtesthellinger}, as follows.

\begin{lemma}
\label{lemma:indep_hellinger_correctness}
    \Cref{alg:indtesthellinger} performs equivalence testing for a product distribution with probability of success at least $2/3$ when using the amount of samples specified in \Cref{theo:indtesthellinger}.
\end{lemma}

\begin{proof}
Note that the first step in the algorithm produces a binning such that with probability at least $9/10$,
    \begin{align}
        \forall i\in [k_A]_0,\forall a\in A_i&: e^{-(i+2)} \mathbbm{1}[i<k_A]\leq q_a\leq e^{-i+1},
        \\
        \forall j\in [k_C]_0,\forall c\in C_j&: e^{-(j+2)}\mathbbm{1}[j<k_C]\leq q_c\leq e^{-j+1}.
    \end{align}
The guarantee follows directly from \Cref{lemma:learn_approx} and a union bound, analogously for both $A$ and $C$: since we learn all $q_a\geq 1/N_{\text{mixed}}$ up to a factor two, the smallest possible $q_a$ for which the estimate $\hat q_a$ satisfies $\hat q_a\geq e^{-i+1}$ is $e^{-(i+1)}/2$. This implies that no $a$ with $q_a < e^{-i+1}/2$ can be assigned to $A_i$. Analogously, the largest $p_a$ with $a\in A_i$ has to satisfy $p_a\leq 2e^{-i+1}$. In particular, it holds with high probability that
\begin{equation}
\label{eq:mi_bucketing_guarantee}
    \forall i\in [k_A-1]_0:\max_{a,a'\in A_i}\frac{p_a}{p_{a'}}\leq e^3,\qquad \forall j\in [k_C-1]_0:\max_{c,c'\in C_j}\frac{p_c}{p_{c'}}\leq e^3.
\end{equation}
We will use these inequalities in the following subsections. Note that \Cref{lemma:learn_approx} also upper bounds $\hat q_a$ and $\hat q_c$ for small values from $A_0$ and $C_0$.

Let us now consider the squared Hellinger distance between $P_{AC}$ and $Q_{AC}$. 
\begin{align}\label{eqn:hellingerdistpq}
    D_H^2(P_{AC},Q_{AC}) &= \sum_{S_{ij} \in \mathcal{S}} \sum_{(a,c)\in S_{ij}}(\sqrt{p_{ac}}-\sqrt{q_{ac}})^2 = \sum_{i,j} D_H^2(\subD{P}{S_{ij}}{AC}, \subD{Q}{S_{ij}}{AC}  )=:\sum_{i,j}\eta_{ij}.
\end{align}
The above equation implies that if $P_{AC}=Q_{AC}$, then $\subD{P}{S_{ij}}{AC} = \subD{Q}{S_{ij}}{AC}$ will hold for every $S_{ij} \in \mathcal{S}$. On the other hand, if $D_H^2(P_{AC},Q_AQ_C) \geq \eps$, then by the pigeonhole principle, at least one of the $|\mathcal{S}| = k_Ak_C\leq k_{AC}$ categories, say $S^{i'j'}$, needs to satisfy
\begin{equation}
    \eta_{i'j'}=D_H^2(\subD{P}{S_{i'j'}}{AC}, \subD{(Q_A Q_C)}{S_{i'j'}}{})\geq \frac{\eps}{k_{AC}}.
\end{equation}
We will test this individually for each category. If a category has a value $\eta_{ij}$ such that $0<\eta_{ij}<\eps/k_{AC}$, then the algorithm might either indicate that $\eta_{ij}=0$ or $\eta_{ij}\geq \eps/k_{AC}$. If it indicates $\geq \eps/k_{AC}$, then this is fine, since we ultimately only care whether $\sum_{ij} \eta_{ij}=0$ or not. In both cases, the output can not cause a wrong answer overall, since we know that the output for $S^{i'j'}$ will be correct with high probability.

The choice of $\delta$ guarantees that we succeed with probability $9/10$ over all categories, and a union bound guarantees a probability of success of at least $2/3$ for the entire algorithm. 

To show the sample complexity, we note that
\begin{itemize}
    \item Step 1 takes $16 M \log(10^3d_Ad_C)$ samples,
    \item Step 2 takes at most $k_{AC}\cdot 8\log(10^4k_{AC}^2)\max\{N_{\text{heavy}}, N_{\text{mixed}}\}$ samples.
\end{itemize}
This directly implies the sample complexity of \Cref{theo:indtesthellinger}.
\end{proof}
In the remainder of this section, we prove \Cref{cl:mi:ub:heavy} and \Cref{cl:mi:ub:mixed}. 
\subsubsection{Sample Complexity of the Heavy Regime}\label{sec:mi_ub:heavy_regime}

\miheavysc*

\begin{proof}
Fix an arbitrary heavy category, say $S_{ij}$. Recall that the heavy regime includes all categories where neither $a$ nor $c$ come from the last bin $A_{k_A}$ and $C_{k_C}$, such that both $P_A[a]$ and $P_C[c]$ are bounded away from zero. For better readability, we will simply write $p_a$ for $P_A[a]$ and $p_c$ for $P_C[c]$, as well as $T:=S_{ij}$ (since we fixed $i$, $j$), and analogously $T_A:=(S_{ij})_A$ and $T_C:=(S_{ij})_C$. Since $Q_{AC}$ is a product distribution, we will write $q_{ac}$ and $q_aq_c$ interchangeably.

We will bound $D_H^2(\subD{P}{T}{AC},\subD{Q}{T}{AC})$ as follows:    
\begin{align}
    \label{eq:dh_indep:heavy_bound_mi_proof}
    D_H^2(\subD{P}{T}{AC},\subD{Q}{T}{AC})&=\sum_{(a,c)\in T}(\sqrt{p_{ac}}-\sqrt{q_{ac}})^2\\
    &= \sum_{(a,c)\in T}\frac{(p_{ac}-q_{ac})^2}{(\sqrt{p_{ac}}+\sqrt{q_{ac}})^2}\\
    &\leq \sum_{(a,c)\in T} \frac{(p_{ac}-q_{ac})^2}{q_{ac}}
    \\
    &\leq \sum_{(a,c)\in T} \frac{(p_{ac}-q_{ac})^2}{\min\limits_{(a,c)\in T}\{q_aq_c\}} \\
    &\leq e^6\frac{\|\subD{P}{T}{AC}-\subD{Q}{T}{AC}\|_2^2}{\max\limits_{(a,c)\in T}\{q_aq_c\}}.
    \label{eqn:ind_test_heavy}
\end{align}

Note that the last line, which bounds $\max_{(a,c),(a',c')\in T}q_{a}q_c/(q_{a'}q_{c'})$ follows from \eqref{eq:mi_bucketing_guarantee}. This follows from learning $q_a$ up to a factor $2$, and the definition of the bins. An analogous argument holds for $T_C$ as well.

Following \eqref{eqn:ind_test_heavy}, the problem of testing heavy categories is reduced to equivalence testing between $\subD{P}{T}{AC} = \subD{Q}{T}{AC}$ and $D_H^2(\subD{P}{T}{AC}, \subD{Q}{T}{AC})\geq \eps/k_{AC}$. Moreover, from \eqref{eqn:ind_test_heavy}, we can say that the problem is reduced to the following testing problem (using bounds from \Cref{lemma:indep_hellinger_correctness}):
\begin{equation}\label{eqn:ind_test_heavy2}
    \subD{P}{T}{AC}  = \subD{Q}{T}{AC}\quad\text{or}\quad\|\subD{P}{T}{AC} - \subD{Q}{T}{AC}\|_2\geq \sqrt{\frac{\eps \max\limits_{(a,c)\in T}\{q_a q_c\}}{e^6k_{AC}}}\geq \sqrt{\frac{\eps e^{-(i+j+2)}}{e^6k_{AC}}}=:\eta.
\end{equation}
        
From \Cref{lemma:equivalence_l2} and \Cref{cor:equalence_l2_algo}, we know that equivalence testing with respect to the $\ell_2$ distance with parameter $\eta$ requires $\ceq\max\{\|\subD{(Q_A Q_C)}{T}{}\|_2/\eta^2,1/\eta,\sqrt{d_Ad_C}\}$ samples. Note that 
\begin{equation}\label{eq:dh_indep:l_2_norm}
\|\subD{(Q_A Q_C)}{T}{}\|_2\leq \sqrt{\sum_{(a,c)\in T}q_a^2 q_c^2}\leq \sqrt{|T_A||T_C|}\max_{(a,c)\in T}\{q_a q_c\}\leq \sqrt{|T_A||T_C|}e^{-(i+j)+2}.
\end{equation}

Hence, the sample complexity of the testing problem in \eqref{eqn:ind_test_heavy2} is $e^{10} k_{AC}\sqrt{|T_A||T_C|}/\eps$. Since $|T_A| \leq d_A$ and $|T_C| \leq d_C$, the problem of testing a heavy category requires $e^{10}k_{AC}\sqrt{d_A d_C}/\eps$ samples. This completes the proof of \Cref{cl:mi:ub:heavy}.
\end{proof}

\subsubsection{Sample Complexity of the Mixed Regime}
\label{sec:mi_ub:mixed_regime}
\mimixedsc*

\begin{proof}
Our proof follows in similar line as \Cref{cl:mi:ub:heavy}. We would like to upper bound $D_H^2(\subD{P}{T}{AC}, \subD{Q}{T}{AC})$ by $\|\subD{P}{T}{AC} - \subD{Q}{T}{AC}\|_2$. However, we can no longer use the same approach as \eqref{eqn:ind_test_heavy} as there is no lower bound on $q_i$. Instead, we denote by $T_{-}$ the set of indices in $T$ for which $(\sqrt{p_{ac}}-\sqrt{q_{ac}})^2\leq x$, for some $x\in (0,1)$ to be determined later and consider $T_{+}:=T\setminus T_{-}$. Then we have
\begin{align}
    D_H^2(\subD{P}{T}{AC}, \subD{Q}{T}{AC})&= \sum_{(a,c)\in T_+}(\sqrt{p_{ac}}-\sqrt{q_{ac}})^2+\sum_{(a,c)\in T_-}(\sqrt{p_{ac}}-\sqrt{q_{ac}})^2\\
    &\leq \sum_{(a,c)\in T_+}\frac{(p_{ac}-q_{ac})^2}{(\sqrt{p_{ac}}+\sqrt{q_{ac}})^2}+|T_-| \cdot x \\
    &\leq \frac{1}{x}\|\subD{P}{T_+}{AC} - \subD{Q}{T_+}{AC} \|_2^2+|T_-| \cdot x\\
    &\leq \frac{1}{x}\|\subD{P}{T}{AC} - \subD{Q}{T}{AC}\|_2^2+|T_-| \cdot x,
\end{align}
where the second inequality holds since $(\sqrt{p_{ac}}+\sqrt{q_{ac}})^2 \geq x$. Thus, we can say that
\begin{align}
    & D_H^2(\subD{P}{T}{AC}, \subD{Q}{T}{AC}) \geq \frac{\eps}{k_{AC}}\implies \|\subD{P}{T}{AC} - \subD{Q}{T}{AC}\|_2 \geq \sqrt{\frac{x\eps}{k_{AC}}-|T_-|x^2} \label{eqn:ind_test_heavy3}
\end{align}
We set $x=\eps/(2|T_-|k_{AC})$, and bound $|T_-|\leq |T_A||T_C|$. As a result, following \eqref{eqn:ind_test_heavy3}, our goal is reduced to the following decision problem
\begin{equation}\label{eqn:ind_test_heavy4}
    \subD{P}{T}{AC}=\subD{Q}{T}{AC} \quad \text{or} \quad \|\subD{P}{T}{AC} - \subD{Q}{T}{AC}\|_2\geq \frac{\eps}{4k_{AC}\sqrt{|T_A||T_C|}}=:\eta.
\end{equation}

Similar to the proof of the previous claim, we now need to bound $\|\subD{Q}{T}{AC}\|_2$. Here we will use \Cref{fact:prelim:l2}, which allows us to bound $\| \subD{Q}{T}{AC}\|_2$ in the following two ways
\begin{equation}\label{eqn:ind_test_l2}
\|\subD{Q}{T}{AC}\|_2\leq ~
\begin{cases}
\sqrt{|T_A||T_C|\left(\max\limits_{(a,c)\in T}\{q_a q_c\}\right)^2}
& \text{(Case 1)}
\\
\sqrt{\max\limits_{(a,c)\in T}\{q_aq_c\}}
& \text{(Case 2)}
\end{cases}
\end{equation}
From \Cref{lemma:equivalence_l2}, we know that equivalence testing with respect to the $\ell_2$ distance with parameter $\eta$ requires $\ceq\|\subD{Q}{T}{AC}\|_2/\eta^2$ samples (will we see that this term dominates the other terms $1/\eta$ and $\sqrt{|T|}$). Thus, to decide \eqref{eqn:ind_test_heavy4}, $\ceq\|\subD{Q}{T}{AC}\|_2/\eta^2$ samples are sufficient.
        
Assume without loss of generality that the small probability mass comes from $A$, such that $q_a \leq 1/M$. Since $\forall c\in T_C: |T_C|p_c\leq 1$ (independent of whether $p_c$ is large or small), both options to bound $\|\subD{Q}{T}{AC}\|_2$ presented in \eqref{eqn:ind_test_l2} might be relevant, depending on the relations between the other parameters:
\begin{enumerate}
    \item \textbf{Case 1}: Using \eqref{eqn:ind_test_heavy4} and \eqref{eqn:ind_test_l2}, we have
    \begin{align}
        \ceq\frac{\|\subD{Q}{T}{AC} \|_2}{\eta^2} &\leq 4^2 \ceq \frac{k_{AC}^2(|T_A||T_C|)^{3/2}}{\eps^2}\max_{(a,c)\in T}\{q_a q_c\}
        \\
        & \leq 4^2\ceq\frac{k_{AC}^2|T_A|^{3/2}|T_C|^{1/2}}{\eps^2M}\leq N_{\text{mixed}}.
    \end{align}
    Since $N_L$ will be maximal if $M=N_L$, and since $|T_A|\leq d_A$ and $|T_C|\leq d_C$, such that we find 
    \begin{equation}
        4^2\ceq\frac{k_{AC}^2d_A^{3/2}d_C^{1/2}}{\eps^2N_{\text{mixed}}}\leq N_{\text{mixed}}\implies N_{\text{mixed}}\geq 4\sqrt{\ceq}\frac{k_{AC}d_A^{3/4}d_C^{1/4}}{\eps},
    \end{equation}
    \item \textbf{Case 2}: We use \eqref{eqn:ind_test_heavy4} and \eqref{eqn:ind_test_l2} to obtain:
    \begin{align}
        \ceq\frac{\|\subD{Q}{T}{AC}\|_2}{\eta^2} &\leq 4^2\ceq\frac{k_{AC}^2|T_A||T_C|}{\eps^2}\sqrt{\max_{(a,c)\in T}\{q_a q_c\}} \\
        &\leq 4^2\ceq\frac{k_{AC}^2|T_A||T_C|^{1/2}}{\eps^2\sqrt{M}}\leq N_{\text{mixed}}.
    \end{align}
    Again, assuming $M=N_{\text{mixed}}$ gives us 
    \begin{equation}
        4^2\ceq\frac{k_{AC}^2d_Ad_C^{1/2}}{\eps^2\sqrt{N_{\text{mixed}}}}\leq N_{\text{mixed}}
        \implies N_{\text{mixed}}\geq (4k_{AC})^{4/3}\ceq^{2/3}\frac{d_A^{2/3}d_C^{2/3}}{\eps^{4/3}}.
    \end{equation}
\end{enumerate}
This completes the proof of \Cref{cl:mi:ub:mixed}.
\end{proof}

\section{Lower Bounds for Independence Testing}\label{sec:ind_test_lb}

In this section, we prove lower bounds on independence testing in the squared Hellinger distance. Our result is formally stated below.

\begin{theorem}\label{theo:theo:indtesthellingerlb} 
Consider independence testing in the squared Hellinger distance $D_H^2$  (\Cref{prob:I_DH2}). Assuming w.l.o.g.\ $d_A\geq d_C$, we have
\begin{equation}
\textnormal{SC}_{\textnormal{I},H}(\eps, d_A, d_C)
= \widetilde{\Omega}\left(\min\left\{\frac{d_A^{3/4}d_C^{1/4}}{\eps}, \frac{d_A^{2/3}d_C^{1/3}}{\eps^{4/3}}\right\}\right)  .  
\end{equation}

\end{theorem}

Since $D_H^2(P,Q)\leq D(P\|Q)$ (\Cref{lemma:kl_to_hell_flaOd}), this also implies lower bounds on testing for mutual information.
\begin{corollary}[Formalized lower bounds of \Cref{res:mi}]
\label{thm:mi_lower}
Consider mutual information testing, \Cref{prob:MI}, where w.l.o.g.\ $d_A\geq d_C$. We have
\begin{equation}
\textnormal{SC}_{\textnormal{MI}}(\eps, d_A, d_C) = \widetilde{\Omega}\left(\min\left\{\frac{d_A^{3/4}d_C^{1/4}}{\eps}, \frac{d_A^{2/3}d_C^{1/3}}{\eps^{4/3}}\right\}\right).
\end{equation}
\end{corollary}

In order to prove the above theorem, we construct a pair of distributions that are hard to distinguish, unless we take a sufficiently large number of samples. Our approach follows the techniques used in the proof of 
\cite[Theorem\ 3.1]{diakonikolas_new_2016}, where lower bounds for independence testing in variation distance are derived. The definition of our hard instances is based on the construction used in \cite[Prop.\ 3.8]{diakonikolas_new_2016} to show bounds on equivalence testing in $D_H^2$ distance, with according modifications to accommodate the product structure we require for independence testing.

We will use the following lemma.

\begin{lemma}\textnormal{\cite[Lemma 3.2]{diakonikolas_new_2016}}
\label{lemma:lower_bound_d_h:mi_bound}
Let $X$ be a uniformly random bit and $K$ be a correlated random variable. Then if $f$ is a function such that $\Pr[f(K)=X]>51\%$, then $I(X \!:\! K)\geq 2\cdot 10^{-4}$.
\end{lemma}

In order to apply \Cref{lemma:lower_bound_d_h:mi_bound}, we will construct a decision problem with two important properties: first, independence testing will be a valid strategy to solve the problem, and second, \Cref{lemma:lower_bound_d_h:mi_bound} allows us to derive lower bounds on the sample complexity of the task. This then will also imply lower bounds on independence testing.

We will construct two sets of distributions, one set containing product distributions, and the other set consisting of distributions $P_{AC}$ which satisfy $D_H^2(P_{AC},P_AP_C) \geq \eps$. Depending on a fair random bit $X$, we select one of the two sets, from which we pick a random distribution $P$. We then draw $n$ samples from $P$. Since the samples are i.i.d\, all information about the underlying distribution can be obtained from the frequencies at which we observe the different tuples $(a,c)$ in the samples. We create a matrix $K$, such that $K[a,c]$ indicates how often we observed the tuple $(a,c)$. The task is now to identify the random bit $X$ based on $K$, and \Cref{lemma:lower_bound_d_h:mi_bound} allows us to derive a lower bound for this task: $I(X:K)$ is increasing with $n$, so let us choose $n^*$ such that $I(X:K)$ reaches the threshold as stated in \Cref{lemma:lower_bound_d_h:mi_bound}. We can then only determine the random bit marginally better than random. However, we would like to correctly reconstruct $X$ with high probability, implying that more than $n^*$ samples are necessary to distinguish between $P_{AC}=P_AP_C$ and  $D_H^2(P_{AC},P_AP_C) \geq \eps$, thereby proving the desired lower bound. As mentioned, a valid strategy to reconstruct $X$ would be a general independence tester, such that a lower bound on this problem implies a lower bound for deciding whether $P_{AC}=P_AP_C$ or $D_H^2(P_{AC},P_AP_C) \geq \eps$.

We will use the Poissonization technique (see \Cref{lem:prelim:poisson}), in particular, we will take $\mathsf{Poi}(n)$ samples from the distribution. With the Poissonization technique, we will not get exactly $n$ samples, but with high probability, we will obtain $\Theta(n)$ samples. Moreover, it is sufficient to have a (pseudo) distribution whose total probability mass is $\Theta(1)$ with high probability, instead of exactly $1$.

\subsection{Description of the Hard Distributions}
\label{sec:mi_lower_def_distr}

To construct a distribution, we first assign all letters $a\in \{2,...,d_A\}$ to one of two sets, $S_1$ and $S_2$. With probability $\alpha:=\min\{n/d_A,1/2\}$, we assign $a$ to $S_1$, and to $S_2$ otherwise. Using this assignment, we then define a distribution:
\begin{itemize}
    \item $\forall a\in S_1,\forall c$: $P_{AC}[a,c]:=1/(2nd_C)$
    
    \item $\forall a\in S_2$: for each $c$ individually, we set
    \begin{align}
        P_{AC}[a,c]=
        \begin{cases}
            \frac{\eps}{d_Ad_C}  &\text{if }X=0,\\
            \text{uniformly at random } \frac{\eps}{2d_Ad_C}  \text{ or } \frac{3\eps}{2d_Ad_C} &\text{if }X=1.  \label{eqn:ind_test_lb}
        \end{cases}
    \end{align}
\end{itemize}
The remaining probability mass is distributed uniformly over $(a=1,c)$, for all $c \in C$. Note that $P_{AC}$ is a product distribution if $X=0$, and $P_C$ is always uniform.

We will prove that a distribution sampled with $X=1$, is far from being independent.

\begin{restatable}{lemma}{milbfar}\label{lemma:mi_lower_bounds_farness}

With high probability, for a distribution $P_{AC}$ generated using $X=1$, it holds that
    \begin{equation}
         D_H^2(P_{AC},P_AP_C)\geq \Omega(\eps).
    \end{equation}
\end{restatable}

We will prove the above lemma in \Cref{sec:milbfar_app}.

\subsection{Preliminary Properties of the Mutual Information}
In order to prove our desired lower bounds, we first state a few statements about the mutual information. Note that these do not directly depend on our specific construction for the hardness distribution, and thus might be of independent interest.

\begin{restatable}{lemma}{miub}\label{lemma:mi_bound_proof_app}

Let $X$ be a uniformly random bit, and $A$ be a random variable taking values in a set $S_A$. Then
\begin{align}
\label{eq:mi_bound_proof_app}
    2I(X\!:\!A)&\leq  \sum_{a\in S_A}\frac{(\Pr[A=a | X=0]-\Pr[A=a | X=1])^2}{\Pr[A=a | X=0]+\Pr[A=a | X=1]}\leq 12I(X\!:\!A).
\end{align} 
\end{restatable}

The first inequality is obtained by bounding the KL-divergence by the $\chi^2$-divergence, and appeared in various forms in the literature (e.g.,\  \cite[App.\ A.1]{diakonikolas_new_2016}). We include a proof here for completeness. To the best of our knowledge, the second inequality is new and shows that the first bound is essentially tight, up to constant factors.

Before proceeding to prove \Cref{lemma:mi_bound_proof_app}, let us first prove the following lemma, which will be used in the proof later.

\begin{lemma}\label{lemma:log_chi_squared_reverse_supp}
    Let $a$ and $b$ be non-negative numbers. Then 
    \begin{align}
        a\log\left(\frac{2a}{a+b}\right)+b\log\left(\frac{2b}{a+b}\right)\geq \frac{1}{6}\frac{(a-b)^2}{a+b}.
    \end{align}
\end{lemma}
\begin{proof}
If at least one of $a$ and $b$ is zero, the result follows trivially. Assume without loss of generality that $a>b>0$. Let 
\begin{equation}
    u:=\frac{a-b}{a+b}>0.
\end{equation}
Note that $1-u={2b}/{(a+b)}$ and $u<1$.
We will use the following inequalities \cite[Eq.\ (2) \& (3)]{log_bounds}, which hold for $x>-1$,
\begin{equation}
\log(1+x)\geq \frac{x}{1+x}\quad \text{and}\quad \frac{\log(1+x)}{x}\geq \frac{2}{2+x}.
\end{equation}
The later implies $\log(1+x)\geq 2x/(2+x)$ if $x>0$. We can now bound
\begin{align}
a\log\left(\frac{2a}{a+b}\right)+b\log\left(\frac{2b}{a+b}\right)&=a\log(1+u)+b\log(1-u)
\\
&\geq a\frac{2u}{2+u}-b\frac{u}{1-u}
\\
&=\frac{1}{2+u}\frac{1}{2}\frac{(a-b)^2}{a+b}\geq \frac{1}{6}\frac{(a-b)^2}{a+b}.
\end{align}
This completes the proof of the lemma.
\end{proof}

Now we are ready to prove \Cref{lemma:mi_bound_proof_app}.

\begin{proof}[Proof of \Cref{lemma:mi_bound_proof_app}]
We can use that $\Pr[X=x]=1/2$, as well as $\log(x)\leq x-1$. Then
    \begin{align}
    I(X\!:\!A)&:=\sum_{\substack{a\in S_A\\ x\in \{0,1\}}}\Pr[A=a, X=x]\log\left(\frac{\Pr[A=a, X=x]}{\Pr[A=a]\Pr[X=x]}\right)
    \\
    &= \sum_{\substack{a\in S_A\\ x\in \{0,1\}}}\Pr[A=a |  X=x]\Pr[X=x]\log\left(\frac{\Pr[A=a | X=x]}{\Pr[A=a]}\right)
    \\
    &\leq \sum_{\substack{a\in S_A\\ x\in \{0,1\}}}\frac{\Pr[A=a | X=x]}{2} \cdot \frac{\Pr[A=a | X=x]-\Pr[A=a]}{\Pr[A=a]}\label{eq:mi_chi_squared_bound_ineq}
    \\
    &= \sum_{\substack{a\in S_A\\ x\in \{0,1\}}}\frac{2\Pr[A=a | X=x]^2-\Pr[A=a | X=x](\Pr[A=a | X=0]+\Pr[A=a | X=1])}{2(\Pr[A=a | X=0]+\Pr[A=a | X=1])}
    \\
    &= \frac{1}{2}\sum_{a\in S_A}\frac{(\Pr[A=a | X=0]-\Pr[A=a | X=1])^2}{\Pr[A=a | X=0]+\Pr[A=a | X=1]}.
\end{align}
For the other direction, we use \Cref{lemma:log_chi_squared_reverse_supp} to bound the following
\begin{align}
    &\frac{1}{2}\frac{(\Pr[A=a | X=0]-\Pr[A=a | X=1])^2}{\Pr[A=a | X=0]+\Pr[A=a | X=1]}
    \\
    &\leq 3\sum_{x\in \{0,1\}}\Pr[A=a | X=x]\log\left(\frac{2\Pr[A=a | X=x]}{\Pr[A=a | X=0]+\Pr[A=a | X=1]}\right)
    \\
    &=6\sum_{x\in \{0,1\}}\Pr[A=a | X=x]\Pr[X=x]\log\left(\frac{\Pr[A=a | X=x]}{\Pr[A=a]}\right).
\end{align}
Summing over all $a$ yields the desired bound.
\end{proof}

Another property of the mutual information we will use is the following folklore result (for example, apply the chain rule, $I(X\!:\!Y,Z)=I(X\!:\!Z)+I(X\!:\!Y|Z)$, and use that for Markov chains, $I(X\!:\!Y|Z)\leq I(X\!:\!Y)$ \cite[p.\ 35]{cover_thomas}).
\begin{lemma}\label{lemma:mi_cond_indep_bound}
    Given random variables forming a Markov chain $Y - X - Z$, i.e., $I(Y\!:\!Z| X)=0$, we have
    \begin{equation}
        I(X\!:\!Y,Z)\leq I(X\!:\!Y)+I(X\!:\!Z).
    \end{equation}
\end{lemma} 

\subsection{Bounding the Mutual Information}
As introduced earlier, $K$ denotes the matrix indicating how often we witnessed the different samples. In particular, $K[a,c]$ denotes how often the tuple $(a,c)$ was observed in the provided samples. In the following, for an arbitrary $a\in A$, we denote by $K_a$ the vector in $\mathbb{N}^{ \times d_C}$, where $K_a[c]$ denotes the number of times the pair $(a,c)$ has occurred in $K$. 

By the construction of our distributions defined in Eq. \eqref{eqn:ind_test_lb}, the different $K_a$'s are independent since the $a$'s are assigned independently to either $S_1$ or $S_2$. Further, $P_C$ is uniform, such that we have by \Cref{lemma:mi_cond_indep_bound} that
\begin{equation}
    I(X\!:\!K)\leq \sum_{a}I(X\!:\!K_a).
\end{equation}

Note that due to dependencies within a specific $a$, we cannot directly apply a similar bound for dimension $C$, since knowing about a certain $K_{a}[c]$ gives us information about the type of $a$, that is, whether $a\in S_1$ or $a\in S_2$, which in turn provides information about other elements $K_{ac'}$. 
Following the ideas used by \cite[Lemma 3.7]{diakonikolas_new_2016}, we perform a case distinction for each $a \in A$: $(i)$ if there are less than two samples in $K$ with a specific $a$, the corresponding mutual information is small and we calculate it directly, $(ii)$ if there are at least two samples in $K$ with the same $a$, we want to find a bound roughly of the following form, for some suitable $\beta$:
\begin{equation}
\label{eq:mi_low_mi_condmi_bound}
    I(X\!:\!K_a)\leq \beta I(X\!:\!K_a | a\in S_2).
\end{equation}
The exact statement we will derive, see \eqref{eq:cmi_to_conditional_bound}, will be a bit more technical. Once we know the type of $a$, the elements within $a$ are independent again, and we will be able to bound
\begin{equation}
    I(X\!:\!K_a | a\in S_2)\leq \sum_cI(X\!:\!K_{a}[c] | a\in S_2).
\end{equation}
Note that we intuitively expect the terms with $a\in S_2$ to dominate, as the distribution over $a\in S_1$ is independent of $X$.

In the following, we will define two sets:
\begin{align}
U:=\{\Lambda \in \mathbb{N}_0^{\times d_C}: \|\Lambda\|_1\geq 2\},\quad \text{and}\quad 
V:=\{\Lambda\in \mathbb{N}_0^{\times d_C}: \|\Lambda\|_1< 2\} .   
\end{align}

For any $\Lambda\in \mathbb{N}_0^{\times d_C}$, let us define:
\begin{equation}
s_a(\Lambda):=\frac{(\Pr[K_a=\Lambda|X=0]-\Pr[K_a=\Lambda|X=1])^2}{\Pr[K_a=\Lambda|X=0]+\Pr[K_a=\Lambda|X=1]}.
\end{equation}

We can use \Cref{lemma:mi_bound_proof_app} to bound
\begin{align}
    I(X\!:\!K)
    &\leq \sum_{a \in A} I(X\!:\!K_a)
    \leq \sum_{a \in A} \left[\sum_{\Lambda\in U}s_a(\Lambda)+\sum_{\Lambda\in V}s_a(\Lambda)\right].\label{mi:lb:mibound}
\end{align}

We will separately bound the two terms in \eqref{mi:lb:mibound}, and determine the maximal sample complexity such that they are still upper bounded by a constant, see , which will then directly imply our lower bounds.

To bound the first term involving $U$, we first show two supplementary statements. 
\begin{restatable}{claim}{milbbeta}\label{cl:milbbeta}
There exists a constant $c$ such that for all $a\in A$, $x\in\{0,1\}$, and $\Lambda\in U$,
\begin{equation}
    \frac{\Pr[K_a=\Lambda, a\in S_2|X=x]}{\Pr[K_a=\Lambda, a\in S_1|X=x]}\leq c\frac{\eps^2n^2}{\alpha d_A^2}.
\end{equation}

\end{restatable}

\begin{proof}
We start with finding $\beta$, where we use that $X$ and whether $a\in S_2$ or not are independent of each other. First, note that
\begin{align}
    \Pr[K_a=\Lambda, a\in S_1|X=x]&
    =\prod_{c=1}^{d_C}\frac{\left(\frac{n}{2nd_C}\right)^{K_a[c]}e^{-\frac{n}{2nd_C}}}{K_{a}[c]!}=\left(\frac{1}{2d_C}\right)^{\|K_a\|_1}e^{-\frac{1}{2}}\prod_{c=1}^{d_C}\frac{1}{K_{a}[c]!},
    \\
    \Pr[K_a=\Lambda, a\in S_2|X=0]&=\prod_{c=1}^{d_C}\frac{\left(\frac{n\eps}{d_Ad_C}\right)^{K_a[c]}e^{-\frac{n\eps}{d_Ad_C}}}{K_{a}[c]!}=\left(\frac{n\eps}{d_Ad_C}\right)^{\|K_a\|_1}e^{-\frac{n\eps}{d_A}}\prod_{c=1}^{d_C}\frac{1}{K_{a}[c]!},
    \\
    \Pr[K_a=\Lambda, a\in S_2|X=1]&=\prod_{c=1}^{d_C}\frac{\frac{1}{2}\left(\left(\frac{n\eps}{2d_Ad_C}\right)^{K_a[c]}e^{-\frac{n\eps}{2d_Ad_C}}+\left(\frac{3n\eps}{2d_Ad_C}\right)^{K_a[c]}e^{-\frac{3n\eps}{2d_Ad_C}}\right)}{K_{a}[c]!}
    \\
    &\leq \left(\frac{3n\eps}{2d_Ad_C}\right)^{\|K_a\|_1}e^{-\frac{n\eps}{2d_A}}\prod_{c=1}^{d_C}\frac{1}{K_{a}[c]!}.\label{eq:mi_lower:pr_x_ub}
\end{align}
Note that $\Pr[K_a=\Lambda, a\in S_2|X=0]$ is also upper bounded by \eqref{eq:mi_lower:pr_x_ub}. We can use the above to bound
\begin{align}
    \frac{\Pr[K_a=\Lambda, a\in S_2|X=x]}{\Pr[K_a=\Lambda, a\in S_1|X=x]}&=\frac{\Pr[K_a=\Lambda|a\in S_2,X=x]\Pr[a\in S_2]}{\Pr[K_a=\Lambda| a\in S_1,X=x]\Pr[a\in S_1]}
    \\
    &\leq \frac{(1-\alpha)\left(\frac{3n\eps}{2d_Ad_C}\right)^{\|K_a\|_1}e^{-\frac{n\eps}{2d_A}}\prod_{c=1}^{d_C}\frac{1}{K_{a}[c]!}}{\alpha \left(\frac{1}{2d_C}\right)^{\|K_a\|_1}e^{-\frac{1}{2}}\prod_{c=1}^{d_C}\frac{1}{K_{a}[c]!}}\leq \frac{1}{\alpha}\left(\frac{3\eps n}{d_A}\right)^{\|K_a\|_1}.
\end{align}
Note that we impose the restriction $n\eps/d_A\leq 1$. We will find that this does not contradict with the bound we derive for $n$. This means $\beta$ will take the maximal value of this expression, which by assumption is achieved for $\|K_a\|_1$ minimal, $\|K_a\|_1=2$. Let $c$ be a large enough constant, then the choice
\begin{equation}\label{eq:mi_lower_beta}
\beta:=c\frac{\eps^2n^2}{\alpha d_A^2}
\end{equation} 
satisfies the desired inequality.
\end{proof}

We will also need the following statement on the mutual information between $X$ and $K_a$, conditionend that $a\in S_2$.

\begin{restatable}{claim}{milbis}\label{cl:milbis}
It holds that $I(X\!:\!K_{a}[c]|a\in S_2)=O\left(\frac{\eps^2n^2}{d_A^2d_C^2}\right)$.    
\end{restatable}

\begin{proof}
In the following, use the shorthand notation $\gamma:=\eps n/(d_Ad_C)$. Recall that $K_{a}[c]$ denotes the number of times the pair $(a,c)$ has appeared in $K$. Assuming $n\ll d_Ad_C/\eps$, we find: 
\begin{align}
    \Pr[K_{a}[c]=k | X=0,a\in S_2]&=\frac{e^{-\gamma}}{k!}\gamma^k
    =\frac{\gamma^k}{k!}\left[1-\gamma+O\left(\gamma^2\right)\right]
    \\
    \Pr[K_{a}[c]=k | X=1,a\in S_2]&=\frac{e^{-3\gamma/2}}{2\cdot k!}\left(\frac{3\gamma}{2}\right)^k+\frac{e^{-\gamma/2}}{2k!}\left(\frac{\gamma}{2}\right)^k
    \\
    &=\frac{\gamma^k}{2\cdot k!2^k}\left[3^k+1-\frac{3^k3\gamma}{2}-\frac{\gamma}{2}+O\left(\gamma^2\right)\right]
\end{align}
We then use \Cref{lemma:mi_bound_proof_app} to bound
\begin{equation}
    I(X\!:\!K_{a}[c]|a\in S_2)\leq \sum_{k=0}^{\infty}\frac{(\Pr[K_{a}[c]=k|X=0]-\Pr[K_{a}[c]=k|X=1])^2}{\Pr[K_{a}[c]=k|X=0]+\Pr[K_{a}[c]=k|X=1]}=:\sum_{k=0}^{\infty}s_{ac}(k).
\end{equation}

We will bound the $s_{ac}(k)$ values, and we will see that one of them dominates the other strong enough to bound $I(X\!:\!K_{a}[c]|a\in S_2)\leq O(s_{ac}(k_{\max}))$.
\begin{align}
    s_{ac}(k)&=\frac{\gamma^k}{k!}\frac{(1-\gamma+O(\gamma^2)-[3^k+1-3^k3\gamma/2-\gamma/2+O(\gamma^2)]/2^{k+1})^2}{1-\gamma+[3^k+1-3^k3\gamma/2-\gamma/2]/2^{k+1}+O(\gamma^2)}
    \\
    s_{ac}(0)&=\frac{(1-\gamma-[2-2\gamma]/2+O(\gamma^2))^2}{\Theta(1)}=O(\gamma^4), 
    \\
    s_{ac}(1)&=\gamma\frac{(1-\gamma-[4-5\gamma]/4+O(\gamma^2))^2}{\Theta(1)}=O(\gamma^3),
    \\
    (k\geq 2)\quad s_{ac}(k)&\leq \frac{\gamma^k}{k!}\left(\frac{3}{2}\right)^{2k}=O(\gamma^k).
\end{align}
We see that there is a single dominant term, $\gamma^2$, compared to which the contributions of the other terms can be neglected (without loss of generality, we can assume, say, $\gamma<1/10$). Thus we have the following
\begin{equation}
    I(X\!:\!K_{a}[c]|a\in S_2)=O\left(\frac{\eps^2n^2}{d_A^2d_C^2}\right).
\end{equation}    
\end{proof}
We can now proceed to bound $\sum_{\Lambda \in U}s_a(\Lambda)$.

\begin{lemma}
\label{lemma:mi_lower_n_bound_u}
For 
\begin{equation}
    n\leq \min\left\{\frac{d_A^{3/4}d_C^{1/4}}{\eps},\frac{d_A^{2/3}d_C^{1/3}}{\eps^{4/3}}\right\}
\end{equation}
we have with high probability that $\sum_{\Lambda \in U}s_a(\Lambda)\leq O(1)$. 
\end{lemma}

\begin{proof}
Using \Cref{cl:milbbeta} in the first inequality, as well as \Cref{lemma:mi_bound_proof_app}, we can now bound
\begin{align}
    \sum_{\Lambda \in U}s_a(\Lambda)
    &= \sum_{\Lambda\in U}\frac{\left(\sum\limits_{M\in \{S_1,S_2\}}\Pr[K_a=\Lambda,a\in M|X=0]-\Pr[K_a=\Lambda,a\in M|X=1]\right)^2}{\sum\limits_{M\in \{S_1,S_2\}}\Pr[K_a=\Lambda,a\in M|X=0]+\Pr[K_a=\Lambda,a\in M|X=1]}
    \\
    &=\sum_{\Lambda\in U}\frac{(\Pr[K_a=\Lambda,a\in S_2|X=0]-\Pr[K_a=\Lambda,a\in S_2|X=1])^2}{\sum\limits_{M\in \{S_1,S_2\}}\Pr[K_a=\Lambda,a\in M|X=0]+\Pr[K_a=\Lambda,a\in M|X=1]}
    \\
    &\leq \beta\sum_{\Lambda\in U}\frac{(\Pr[K_a=\Lambda,a\in S_2|X=0]-\Pr[K_a=\Lambda,a\in S_2|X=1])^2}{\Pr[K_a=\Lambda,a\in S_2|X=0]+\Pr[K_a=\Lambda,a\in S_2| X=1]}\label{eq:cmi_to_conditional_bound}
    \\
    &\leq \beta\Pr[a\in S_2]\sum_{\Lambda}\frac{(\Pr[K_a=\Lambda \mid a\in S_2,X=0]-\Pr[K_a=\Lambda| a\in S_2,X=1])^2}{\Pr[K_a=\Lambda| a\in S_2,X=0]+\Pr[K_a=\Lambda| a\in S_2,X=1]}
    \\
    &\leq 12\beta\Pr[a\in S_2]I(X\!:\!K_a| a\in S_2)
    \\
    &\leq 12\beta\sum_{c}\Pr[a\in S_2]I(X\!:\!K_{a}[c]| a\in S_2) \label{eq:mi_lower:many_col}.
\end{align}

We bound \eqref{eq:mi_lower:many_col} further using \Cref{cl:milbis}:
\begin{equation}
        \beta\sum_{a,c}\Pr[a\in S_2]I(X\!:\!K_{a}[c]| a\in S_2)\leq \beta\sum_{a,c}(1-\alpha)\frac{\eps^2n^2}{d_A^2d_C^2}\leq d_Ad_C\frac{\eps^2n^2}{\alpha d_A^2}\frac{\eps^2n^2}{d_A^2d_C^2}.
\end{equation}
By \Cref{lemma:lower_bound_d_h:mi_bound} and \eqref{mi:lb:mibound}, we require $\sum_{\Lambda \in U}s_a(\Lambda)$ to be upper bounded by a small constant (for our asymptotic purposes, let this constant be one). Recall from the construction of the hard distribution $\alpha:=\min\{n/d_A,1/2\}$, so now we consider two cases for $\alpha$, and solve for $n$ under the constraint that $\sum_{\Lambda \in U}s_a(\Lambda)\leq O(1)$:
\begin{align}
    \alpha&=\Theta(1):\quad \frac{n^4\eps^4}{d_A^3d_C}\leq O(1)\implies n\leq O\left(\frac{d_A^{3/4}d_C^{1/4}}{\eps}\right), \label{eqn:milbcaseoneone}
    \\
    \alpha&=\frac{n}{d_A}:\quad\quad \frac{n^3\eps^4}{d_A^2d_C}\leq O(1)\implies n\leq O\left(\frac{d_A^{2/3}d_C^{1/3}}{\eps^{4/3}}\right).\label{eqn:milbcaseonetwo}
\end{align}

\end{proof}

Bounding the remaining term, $\sum_a\sum_{\Lambda \in V}s_a(\Lambda)$, is done by a direct calculation. The following lemma contains the results of the calculations, the proof can be found in \Cref{sec:milbfar_app}.

\begin{restatable}{claim}{milbzeroone}\label{lemma:mi_lower_K_0_1_hits}

For the hard distribution as defined in \Cref{sec:mi_lower_def_distr}, for a fixed $a \in A$, the following hold:
\begin{enumerate}[label=(\roman*)]
    \item If there is no occurrence of an index $a$, i.e., the occurrence vector $K_a$ is equal to $\ell_0:= \vec{0}$, then we have:
    \begin{align}
        \frac{\left(\Pr[K_a=\ell_0|X=0]-\Pr[K_a=\ell_0|X=1]\right)^2}{\Pr[K_a=\ell_0|X=0]+\Pr[K_a=\ell_0|X=1]}&\leq O\left(d_C^2\left[\frac{\eps n}{d_Ad_C}\right]^4\right).    
    \end{align}
    \item If a specific index $a$ appears only once, i.e., the occurrence vector $K_a$ satisfies $\|K_a\|_1=1$, we have:
    \begin{align}
        \sum_{\ell,\|\ell\|_1=1}\frac{\left(\Pr[K_a=\ell|X=0]-\Pr[K_a=\ell|X=1]\right)^2}{\Pr[K_a=\ell|X=0]+\Pr[K_a=\ell|X=1]}&\leq O\left(d_C^3\left[\frac{\eps n}{d_Ad_C}\right]^6\frac{d_C}{\alpha}\right).    
    \end{align}
\end{enumerate}

\end{restatable}

Using the above, finding the required bounds on $n$ is direct.

\begin{lemma}
\label{lemma:mi_lower_n_bound_v}
    For 
\begin{equation}
    n\leq \min\left\{\frac{d_A^{3/4}d_C^{2/4}}{\eps},\frac{d_A^{5/6}d_C^{1/3}}{\alpha\eps}\right\}
\end{equation}
we have with high probability that $\sum_{\Lambda \in V}s_a(\Lambda)\leq O(1)$. 
\end{lemma}
\begin{proof}
First, note that by \Cref{lemma:mi_lower_K_0_1_hits},
\begin{align}
    \sum_a\sum_{\Lambda \in V}s_a(\Lambda)&=\sum_a\frac{\left(\Pr[K_a=\ell_0|X=0]-\Pr[K_a=\ell_0|X=1]\right)^2}{\Pr[K_a=\ell_0|X=0]+\Pr[K_a=\ell_0|X=1]}
    \\
    &\quad +\sum_a \sum_{\ell,\|\ell\|_1=1}\frac{\left(\Pr[K_a=\ell|X=0]-\Pr[K_a=\ell|X=1]\right)^2}{\Pr[K_a=\ell|X=0]+\Pr[K_a=\ell|X=1]}
    \\
    &\leq O\left(d_A\left[d_C^2\left[\frac{\eps n}{d_Ad_C}\right]^4+d_C^3\left[\frac{\eps n}{d_Ad_C}\right]^6\frac{d_C}{\alpha}\right]\right).
\end{align}
It remains to solve for $n$:
    \begin{align}
        d_A\left[\frac{\eps n}{d_Ad_C}\right]^4d_C^2&\leq O(1)\implies n\leq O\left(\frac{d_A^{3/4}d_C^{2/4}}{\eps}\right),\label{eqn:milbcasetwone}
        \\
        d_Ad_C^3\left[\frac{\eps n}{d_Ad_C}\right]^6\frac{d_C}{\alpha}&\leq O(1)\implies n\leq O\left(\frac{d_A^{5/6}d_C^{1/3}}{\alpha\eps}\right).\label{eqn:milbcasetwotwo}
    \end{align}
\end{proof}

Using \Cref{lemma:mi_lower_n_bound_u} and \Cref{lemma:mi_lower_n_bound_v}, our lower bound follows directly: 
\begin{equation}
n=\Omega\left(\min\left\{\frac{d_A^{3/4}d_C^{1/4}}{\eps},\frac{d_A^{2/3}d_C^{1/3}}{\eps^{4/3}}\right\}\right).
\end{equation}

This completes the proof of \Cref{theo:theo:indtesthellingerlb}.

\section{Conditional Independence Testing}\label{sec:cmi_ub}

In this section, we will prove our result on conditional independence testing in $D_H^2$ distance, which implies a similar result for conditional mutual information testing.  We will prove the following theorem.

\begin{theorem}\label{theo:cmi:ub:main}
The sample complexity of conditional independence testing in $D_H^2$ distance, \Cref{prob:CI_DH2}, is
\begin{align}
    \textnormal{SC}_{\textnormal{CI},H}(\eps, d_A, d_B, d_C) = \widetilde{O}\left( \max \big\{ f_{\textnormal{sym}}(\eps, d_A, d_B, d_C) , f_{\textnormal{asym}}(\eps, d_A, d_B, d_C) \big\} \right),
\end{align}
where we distinguish between terms that are symmetrical and asymmetrical in $A$ and $C$, and assume w.l.o.g.\ $d_A\geq d_C$, 
\begin{align}
    f_{\textnormal{sym}}(\eps, d_A, d_B, d_C) &:= \max\left\{\frac{d_A^{1/2}d_B^{3/4}d_C^{1/2}}{\eps},\min\left\{\frac{d_A^{1/4}d_B^{7/8}d_C^{1/4}}{\eps},\frac{d_A^{2/7}d_B^{6/7}d_C^{2/7}}  {\eps^{8/7}}\right\}\right\} , \quad \text{and} \\
    f_{\textnormal{asym}}(\eps, d_A, d_B, d_C) &:=
    \min\left\{\frac{d_A^{3/4}d_B^{3/4}d_C^{1/4}}{\eps},\frac{d_A^{2/3}d_B^{2/3}d_C^{1/3}}{\eps^{4/3}}\right\}.
\end{align}   
\end{theorem}

    The proof of \Cref{theo:cmi:ub:main} follows from our discussion of \Cref{alg:cond_indep_top_level} which decomposes the testing problem into multiple subproblems solved by two subroutines $\mathsf{CMITestingSmallRegime}$ (\Cref{alg:cmitestsmall}) and $\mathsf{CMITestingLargeRegime}$ (\Cref{alg:cmitestlarge}), discussed in \Cref{sec:cmi_ub:small_regime} and \Cref{sec:cmi_ub:large}. The sample complexity is obtained by combining the sample complexities of these subroutines.

As a corollary to the above theorem, we have the following result.

\begin{corollary}[Formalized bound of \Cref{res:cmi-upper}] The sample complexity of conditional independence testing, \Cref{prob:CMI}, is
\begin{align}
\textnormal{SC}_{\textnormal{CMI}}(\eps, d_A, d_B, d_C) = \widetilde{O}\left( \textnormal{SC}_{\textnormal{CI},H}(\eps, d_A, d_B, d_C)\right),
\end{align}
\end{corollary}
\begin{proof}
    This follows directly from \Cref{theo:cmi:ub:main} and \Cref{theo:cmiheldistreduction}.
\end{proof}

The underlying idea is to reduce the problem of conditional independence testing with respect to $D_H^2$ distance to a polylogarithmic number of instances of equivalence testing in the $\ell_2$ distance. Compared to independence testing described in \Cref{sec:ind_test_hellinger}, we now face an additional obstacle: sampling from the reference distribution $Q=P_{AB}P_{BC}/P_B$ is non-trivial. The basic idea to overcome this issue is as follows. Assume that besides $P_{ABC}$, we could also sample from $P_{AC|B=b}$, for some $b \in B$ of our choice. Then we could first select a random sample from $P_{ABC}$, $(a,b,c)$, followed by a second sample $(a',b,c')$ from $P_{AC|B=b}$. The sample $(a,b,c')$ can then be treated as a sample coming from $P_{AB}P_{BC}/P_B$, as is easily seen: the first sample we draw needs to be of the form $(a,b,\cdot)$, which happens with probability $p_{ab}$. The probability of obtaining $c'$ from the second sample is given by $p_{bc'}/p_b$, since we conditioned on $b$. Together, this results in 
\begin{equation}
\label{eq:cmi_ub:ideal_sampling}
    \Pr[(a,b,c') \ \text{is obtained}]=\frac{p_{ab}p_{bc'}}{p_b}.
\end{equation}
Since we only have sample access to $P_{ABC}$, simulating samples from $Q_{ABC}$ becomes more challenging. In our approach, we split the distribution into two regimes, depending on a threshold $\nu$.
\begin{enumerate}[label=(\roman*)]
    \item A large regime $L$, consisting of all $(a,b,c)$ for which $\hat p_b$ is above a certain threshold $\nu$, that is, $B_L:=\{b| \hat p_b\geq\nu\}$, such that $L=\{(a,b,c) | b\in B_L\}$.
    \item A small regime $S$, for the remaining $b$, defined as $B_S:=\{b | \hat p_b<\nu\}$, $S=\{(a,b,c)|  b\in B_S\}$. 
\end{enumerate}
The reason for the split is that we apply different techniques to simulate (approximate) samples from $Q_{ABC}$ in the respective regimes. To test $P_{ABC}$ for conditional independence, it is then sufficient to test individually whether $D_H^2(\subD{P}{S}{ABC},\subD{Q}{S}{ABC})\geq \eps/2$ or $D_H^2(\subD{P}{L}{ABC},\subD{Q}{L}{ABC})\geq\eps/2$, see the respective subsections, \Cref{sec:cmi_ub:small_regime} and \Cref{sec:cmi_ub:large}. However, both techniques are based on the basic idea of combining two samples with the same coordinate in $B$ to simulate conditional independence as in \eqref{eq:cmi_ub:ideal_sampling}. 
We now introduce our algorithm for conditional independence testing in $D_H^2$ distance, \Cref{alg:cond_indep_top_level}. For this, let
\begin{align}
\label{eq:cmi_ub:complexity_full}
    N:= 10^2\log(d_B) \max\{N_S, N_L\},
\end{align}
where
\begin{align}
    N_S&:= 10^{10}\min\left\{\frac{d_A^{1/4}d_B^{7/8}d_C^{1/4}}{\eps_S},\frac{d_A^{2/7}d_B^{6/7}d_C^{2/7}}{\eps_S^{8/7}}\right\},\quad (\text{\Cref{lem:cmi:small}})  \\   
    N_L&:= 10^6\ceq \log\left(\frac{d_Ad_Bd_C}{\eps_L}\right)^7\max\Bigg\{\min\left\{\frac{d_A^{3/4}d_B^{3/4}d_C^{1/4}}{\eps_L},\frac{d_A^{2/3}d_B^{2/3}d_C^{1/3}}{\eps_L^{4/3}}\right\},
    \\
&\quad\quad\quad\quad\quad\quad\quad\quad\quad\quad\quad\quad\quad\quad\quad\frac{d_A^{1/2}d_B^{3/4}d_C^{1/2}}{\eps_L},N_S\Bigg\},\quad(\text{\Cref{lem:cmi:large}}).
\end{align}
Here $\ceq$ denotes the implicit constant used in equivalence testing of distributions in \Cref{cor:equalence_l2_algo}. This gives the sample complexity reported in \Cref{theo:cmi:ub:main}.

\begin{algorithm}[H]
\LinesNumbered
\DontPrintSemicolon
\setcounter{AlgoLine}{0}
\caption{Conditional Independence Testing in $D_H^2$ distance}
\label{alg:cond_indep_top_level}
\KwIn{A multiset $\mathcal{S}$ of $N$ triplets in $A\times B\times C$, $\eps, \nu \in (0,1)$\Comment*[r]{$N$:~see \eqref{eq:cmi_ub:complexity_full}}} 
\KwOut{`{\bf Yes}' or `{\bf No}' \Comment*[r]{`{\bf Yes}' indicates conditional independence }}

$\mathcal{S}_L$, $\mathcal{S}_S$, and $\mathcal{S}_{\nu}$ $\leftarrow$ remove random disjoint multisets from $\mathcal{S}$, such that  $|\mathcal{S}_L|=N_L$, $|\mathcal{S}_S|=N_S$ and $|\mathcal{S}_{\nu}|=32\log(10^3d_B)N_S$ \Comment*[r]{$N_L$:~\Cref{lem:cmi:large} \& $N_S$:~\Cref{lem:cmi:small}}

$\hat P_b\leftarrow$ empirical frequency of $b\in B$ in $\mathcal{S}_{\nu}$ for every $b \in B$ 

$\nu\leftarrow 1/(2N_S)$, $B_S\leftarrow\{b\in B|\hat p_b < \nu\}$, $B_L\leftarrow\{b\in B|\hat p_b \geq \nu\}$.

resSmall $\leftarrow$ $\mathsf{CMITestingSmallRegime} (\mathcal{S}_S$, $B_S$, $\eps/2)$ \Comment*[r]{\Cref{alg:cmitestsmall} \& \Cref{lem:cmi:small}} 

resLarge $\leftarrow$ $\mathsf{CMITestingLargeRegime} (\mathcal{S}_L$, $B_L$, $\eps/2$) \Comment*[r]{\Cref{alg:cmitestlarge} \&  \Cref{lem:cmi:large}} 

\Return `{\bf Yes}' if $\text{resSmall}=\text{resLarge}=\text{`{\bf Yes}'}$, otherwise `{\bf No}'.
\end{algorithm}

We have the following lemma about the correctness of \Cref{alg:cond_indep_top_level}.

\begin{lemma}
Given $N$ samples (as defined in \eqref{eq:cmi_ub:complexity_full}), \Cref{alg:cond_indep_top_level} returns `{\bf Yes}' if $P_{ABC}$ is conditionally independent, `{\bf No}' if $D_H^2(P_{ABC}, P_{A|B}P_{C|B}P_{B})\geq \eps$, with probability at least $2/3$.
\end{lemma}

\begin{proof}
    In order to prove that \Cref{alg:cond_indep_top_level} succeeds probability at least $2/3$, we first show that the following subtasks succeed with high probability each: 
    \begin{itemize}
        \item $B_S$ and $B_L$ satisfy the conditions mentioned in \Cref{alg:cmitestsmall} ($\forall b\in B_L: p_b\geq \nu/2$) and \Cref{alg:cmitestlarge} ($\forall b\in B_S:p_b<1/N_S$) with probability at least $99/100$. This follows directly from \Cref{lemma:learn_approx}.
        \item  \Cref{alg:cmitestsmall} with probability at least $99/100$, returns `No' (i.e. `Far') in case $D_H^2(\subD{P}{S}{ABC},\subD{Q}{S}{ABC})\geq \eps/2$, and `Yes' if $\subD{P}{S}{ABC}=\subD{Q}{S}{ABC}$. We present the proof in \Cref{sec:cmi_ub:small_regime} (\Cref{lem:cmi:small}).
        
        \item \Cref{alg:cmitestlarge}, with probability at least $99/100$, returns `No' (i.e. `Far') in case $D_H^2(\subD{P}{L}{ABC},\subD{Q}{L}{ABC})\geq \eps/2$, and `Yes' if $\subD{P}{L}{ABC}=\subD{Q}{L}{ABC}$. We present the proof in \Cref{sec:cmi_ub:large} (\Cref{lem:cmi:large}). 
        
    \end{itemize}
    A union bound then guarantees that the entire algorithm succeeds with probability at least $97/100>2/3$. If $D_H^2(P_{ABC},Q_{ABC})\geq \eps$, then by the pigeonhole principle, either $D_H^2(\subD{P}{S}{ABC},\subD{P}{S}{ABC})\geq \eps/2$ or $D_H^2(\subD{P}{L}{ABC},\subD{P}{L}{ABC})\geq \eps/2$, which we will detect with high probability as argued above, and if $P_{ABC}=Q_{ABC}$, then we also have equivalence on $S$ and $L$.
    
    The sample complexity is then the direct result of \Cref{lem:cmi:small} and \Cref{lem:cmi:large}, since it is clear from the description of \Cref{alg:cond_indep_top_level} that $N$ needs to be chosen such that $N = |\mathcal{S}_{\nu}|+|\mathcal{S}_{S}|+|\mathcal{S}_{L}|$.
\end{proof}
We now present the subroutines testing the large (\Cref{alg:cmitestlarge}) and the small regime (\Cref{alg:cmitestsmall}).
We will begin our discussion with the small regime, as this will define the threshold separating the two regimes. For the large regime, we will be able to simulate a distribution $\tilde Q_{ABC}$ such that $\subD{\tilde Q}{L}{ABC}=\subD{Q}{L}{ABC}$, allowing us to reuse ideas from \Cref{sec:ind_test_hellinger}, an approach which will not be possible in the small regime. In the following, we do not optimize for logarithmic or constant factors. However, we avoid to abstract these factors using $O$-notation to provide explicit thresholds in the following algorithms.

\subsection{Testing The Small Regime}
\label{sec:cmi_ub:small_regime}

Our result for the small regime is as follows:

\begin{restatable}{lemma}{lemcmismall}\label{lem:cmi:small}

$\mathsf{CMITestingSmallRegime}$ (\Cref{alg:cmitestsmall}) is correct with probability at least $99/100$ if called using
\begin{equation}        
    N_S:=  10^{10}\min\left\{\frac{d_A^{1/4}d_B^{7/8}d_C^{1/4}}{\eps_S},\frac{d_A^{2/7}d_B^{6/7}d_C^{2/7}}{\eps_S^{8/7}}\right\}
\end{equation}
samples, and a set $B_S$ such that for all $b\in B_S: p_b\leq 1/N_S$.
\end{restatable}

\subsubsection{Description of the Algorithm}
From the definition of $B_S$ in \Cref{alg:cond_indep_top_level} it is clear that the small regime will consider $b$'s for which $p_bN_S<1$, which means that in expectation, we see less than one sample with $b$.
While we can say little about the statistics of a \emph{specific} $b\in B_S$, concentration bounds can tell us how many rare collisions (that is, a $b\in B_S$ appearing at least twice among the samples), we will witness \emph{in total}. 
We now define two algorithms \algSmallP{} and \algSmallQ{}, which simulate samples based on collisions similar to the idea described in \eqref{eq:cmi_ub:ideal_sampling}. The statistics between the outputs of \algSmallP{} and \algSmallQ{} coincide if $P_{ABC}=P_{AB}P_{BC}/P_B$, and can be distinguished if $D_H^2(P_{ABC},P_{AB}P_{BC}/P_B)\geq \eps$.

The algorithms each take a set of $\mathsf{Poi}(\tilde N_S)$ samples from $P_{ABC}$, and process them in a randomly chosen order. They discard samples $(a,b,c)\in L$ and only store samples $(a',b',c')\in S$. As soon as there are two stored samples with the same letter $b$, the process produces an output: 
\begin{itemize}
    \item The first algorithm, \algSmallP{} (\Cref{alg:simsmallfirst}), mimics the original distribution. If \algSmallP{} witnesses first $(a_1,b,c_1)$ and then $(a_2,b,c_2)$, it returns $(a_1,b,c_1)$. The second sample, $(a_2,b,c_2)$, is discarded. 
    
    \item The second algorithm, \algSmallQ{} (\Cref{alg:simsmallsecond}), simulates samples which correspond to a conditionally independent version of \algSmallP{}. Given  sample $(a_1,b,c_1)$ at first and then $(a_2,b,c_2)$, it returns $(a_1,b,c_2)$. This is described in \Cref{alg:simsmallsecond}.
\end{itemize}

An equivalent way, which we adopt in the description of  \Cref{alg:simsmallfirst} and \Cref{alg:simsmallsecond} is to say that the respective algorithm takes randomly chosen pairs for which $B$-coordinates and processes them as described above.
In the following, the notation $\mathcal{S}|_{B=b}$ denotes the multiset obtained from $\mathcal{S}$ by only keeping samples whose $B$-coordinate is $b$.

\begin{algorithm}[H]
\LinesNumbered
\DontPrintSemicolon
\setcounter{AlgoLine}{0}
\caption{\algSmallP}
\label{alg:simsmallfirst}
\KwIn{A multiset $\mathcal{S}$ of $N$ triplets in $A \times B \times C,B_S \subseteq [d_B]$} 
\KwOut{$M \in \mathbb{N}^{d_A \times d_B \times d_C}$}

$M \gets 0^{d_A \times d_B \times d_C}$

\For{$b\in B_S$}{
    \While{$\big|\mathcal{S}|_{B=b}\big|\geq 2$}{
        $(a_1,b,c_1)$, $(a_2,b,c_2)$ $\leftarrow$ pick two random triplets from $\mathcal{S}|_{B=b}$, without replacement

        $M_{a_1,b,c_1} \gets M_{a_1,b,c_1} +1$
    }
}

\Return $M$

\end{algorithm}

\begin{algorithm}[H]
\LinesNumbered
\DontPrintSemicolon
\setcounter{AlgoLine}{0}
\caption{\algSmallQ}
\label{alg:simsmallsecond}
\KwIn{A multiset $\mathcal{S}$ of $N$ triplets in $A \times B \times C,B_S \subseteq [d_B]$} 
\KwOut{$M \in \mathbb{N}^{d_A \times d_B \times d_C}$}

$M \gets 0^{d_A \times d_B \times d_C}$

\For{$b\in B_S$}{
    \While{$\big|\mathcal{S}|_{B=b}\big|\geq 2$}{
        $(a_1,b,c_1)$, $(a_2,b,c_2)$ $\leftarrow$ pick two random triplets from $\mathcal{S}|_{B=b}$, without replacement

        $M_{a_1,b,c_2} \gets M_{a_1,b,c_2} +1$
    }
}

\Return $M$

\end{algorithm}

While the analysis of these statistics enables us to test for conditional independence, it comes with two challenges:
\begin{itemize}
    \item First, the samples simulated are not i.i.d.\ samples, which prevents us from using existing equivalence testing algorithms as a black box.
    
    \item Second, the information our samples carry about the conditional independence of the original distribution can be `skewed': intuitively, this comes from the fact that in order to produce a sample with a specific $b$, we need to see two samples from $P_{ABC}$ which have the same $b$ coordinate, which happens with a probability proportional to $p_b^2$, as opposed to $p_b$, putting smaller $p_b$ at a disadvantage. This is also the reason why we need \algSmallP{}, and cannot directly compare the statistics of \algSmallQ{} against those of $P_{ABC}$. See \Cref{sec:cmi_ub:intuition} for a more detailed discussion.
\end{itemize}
Note that for independence testing, we solved the more general problem of testing whether $P_{AC}=Q_AQ_C$ for arbitrary distributions $P_{AC}$ and $Q_{AC}$. This was possible since simulating samples from $Q_AQ_C$ using samples from $Q_{AC}$ is direct. The analogue no longer holds for conditional independence testing: simulating samples from $Q_{ABC}=P_{AB}P_{BC}/P_B$ is a bottleneck which can dominate the sample complexity. In general, the problem of deciding whether a distribution $P_{ABC}$ is equal to a given conditionally independent distribution $R_{ABC}$ could be a different (and, in terms of sample complexity, likely simpler) problem. 

We can now compare the statistics between samples produced by \Cref{alg:simsmallfirst} and \Cref{alg:simsmallsecond} to test for conditional independence:

\begin{algorithm}[H]
\LinesNumbered
\DontPrintSemicolon
\setcounter{AlgoLine}{0}
\caption{$\mathsf{CMITestingSmallRegime}$}
\label{alg:cmitestsmall}
\KwIn{A multiset $\mathcal{S}_S$ of $N_S$ triplets in $A \times B \times C$, $B_S\subseteq [d_B]$, $\eps_S\in(0,1)$}

\KwOut{`{\bf Yes}', `{\bf No}' or `{\bf Abort}'}
$\tilde N_S:=N_S/8$

$\forall i \in [4]$: $M_i\leftarrow\mathsf{Poi}(\tilde N_S)$ 

\If{$\sum_{i=1}^4M_i> N_S$}{

\Return{`{\bf Abort}'}

}

$\forall i \in [4]$: Pick disjoint multisets $\mathcal{S}_i$ of size $M_i$ from $\mathcal{S}_S$

$X \gets \mathsf{SimABC}(\mathcal{S}_1, B_S)$, $X' \gets \mathsf{SimABC}(\mathcal{S}_2, B_S)$

$Y \gets \mathsf{SimABC_{CI}}(\mathcal{S}_3, B_S)$, $Y' \gets \mathsf{SimABC_{CI}}(\mathcal{S}_4, B_S)$

$Z\leftarrow\sum_{abc}X_{abc}X'_{abc}-2X_{abc}Y_{abc}+Y_{abc}Y'_{abc}$

\Return{`{\bf Yes}' if $Z< \frac{\tilde N_S^4\eps_S^4}{2\cdot 4^7d_Ad_B^3d_C}$, \textnormal{otherwise} `{\bf No}'}
                          
\end{algorithm}

To prove that the algorithm succeeds with high probability, we need to determine $\E[Z]$ for the two cases where $\subD{P}{S}{ABC}$ is either conditionally independent, or gapped away from conditional independence, and we also require a bound on the variance of $Z$.

Before proceeding further, we remark on the fact that there might be instances where we do not need to test the light regime.

\begin{remark}
    In some instances, it might be that we can infer that $P[B_S]\leq \eps/(2\log(d_Ad_C))$, which imples $I(A:C|B_S)\leq \eps/2$, since $\forall b\in B: I(A:C|B=b)\leq \log(d_Ad_C)$. It might also be that $B_S=\varnothing$. While potentially interesting for instance-optimal approaches, we do not treat these cases explicitly: in such instances, $Z$ will simply never cross the threshold in \Cref{alg:cmitestsmall}.
\end{remark}

\subsubsection{Intuition for the Sample Complexity}
\label{sec:cmi_ub:intuition}
Before we start with the formal proof, let us start with a brief discussion on the intuition for the sample complexity presented in \Cref{lem:cmi:small}. We will make a few simplifying assumptions, which nevertheless results in the asymptotically same sample complexity as our formal analysis. 

For simplicity, let us assume that we have taken $N_I$ samples.
Intuitively, note that for a `rare' $b$, the probability of collision among $N_I$ samples is proportional to $O(p_b^2N_I^2)$. Thus, in total there will be roughly $O(\sum_{i \in S}p_i^2N_I^2)=O(\|\subD{P}{S}{ABC}\|_2^2N_I^2)$ rare collisions. Assume that we have a multiset containing pairs of elements that collided. Ignoring the fact that we sample without replacement, the probability of drawing a specific $b$ from this multiset is then roughly as follows
\begin{align}
    \Pr\left[\text{we draw~}b\right]&=\Pr\left[b\text{~appears in a collision}\right]\Pr\left[\text{we draw~}b|b\text{~appears in a collision}\right]
    \\
    &=O\left((N_Ip_b)^2\frac{1}{(N_I\|\subD{P}{S}{ABC}\|_2)^2}\right)
    \\
    &\geq p_b \cdot O\left(\frac{\eps_S}{4d_B\|\subD{P}{S}{ABC}\|_2^2}\right), \label{eq:cmi_ub_small_regime_intuition_factor}
\end{align}
where $\eps_S$ denotes the threshold we test for in \Cref{alg:cmitestsmall}. 
In the last inequality above, we used that we can neglect $p_b\leq \eps_S/(2d_B)$, as such $b$ combined carries less than $\eps_S/2$ weight in total. 
Compared to the original distribution, where a specific $b$ appears with probability $p_b$, the probability to draw a specific $b$ from the multiset of `rare' collisions is skewed by a factor $\gamma=O({\eps_S}/{(d_B\|\subD{P}{S}{ABC}\|_2^2)})$.
This effect needs to be taken into account when we want to test for properties of the original distributions: if we have $D_H^2(P_{ABC},Q_{ABC})\geq \eps_S$, the bias our sampling introduces would require us to test for an increased precision $\eta:=\eps_S\cdot \widetilde{O}(\eps_S/{(2d_B\|\subD{P}{S}{ABC}\|_2^2)})$ to surely detect if $P$ is conditionally independent or not.
We can bound this further using $\|\subD{P}{S}{ABC}\|_2^2\leq O(\min\{d_B/N_I^2,1/N_I\})$, see \Cref{fact:prelim:l2}. 

Performing the equivalence testing to precision $\eta$ in the squared Hellinger distance requires $\Theta(\min\{d^{2/3}/\eta^{4/3},d^{3/4}/\eta\})$ i.i.d.\ samples \cite[Thm.\ 4.2]{sublinearly_Testable}. For the moment, we ignore the fact that our samples are not independent. It is easy to check that here, we always have $d^{2/3}/\eta^{4/3}\leq d^{3/4}/\eta$, which is why we only consider the first option, $O(d^{2/3}/\eta^{4/3})$. The number of rare collisions we witness, $O(N_I^2\|\subD{P}{S}{ABC}\|_2^2)$, needs to be large enough to perform the testing, such that we require
\begin{equation}
    N_I^2\|\subD{P}{S}{ABC}\|_2^2\geq\frac{(d_Ad_Bd_C)^{2/3}(\|\subD{P}{S}{ABC}\|_2^2d_B)^{4/3}}{\eps_S^{8/3}}\implies N_I\geq O\left((\|\subD{P}{S}{ABC}\|_2^2)^{1/6}\frac{d_B(d_Ad_C)^{1/3}}{\eps_S^{4/3}}\right).
\end{equation}
We then insert the two bounds on $\|\subD{P}{S}{ABC}\|_2^2$ (see \Cref{fact:prelim:l2}), and find for the sufficient $N_I$
\begin{equation}
    \|\subD{P}{S}{ABC}\|_2^2\leq\frac{d_B}{N_I^2}:\quad N_I\geq O\left(\frac{d_B^{1/6}}{N_I^{1/3}}\frac{d_B(d_Ad_C)^{1/3}}{\eps_S^{4/3}}\right)\implies N_I\geq O\left(\frac{d_B^{7/8}(d_Ad_C)^{1/4}}{\eps_S}\right),
\end{equation}
and 
\begin{equation}
    \|\subD{P}{S}{ABC}\|_2^2\leq\frac{1}{N_I}:\quad N_I\geq O\left(\frac{1}{N_I^{1/6}}\frac{d_B(d_Ad_C)^{1/3}}{\eps_S^{4/3}}\right)\implies N_I\geq O\left(\frac{d_B^{6/7}(d_Ad_C)^{2/7}}{\eps_S^{8/7}}\right),
\end{equation}
resulting in the sample complexity reported in \Cref{lem:cmi:small}.
We already pointed out that a few assumptions we made along the way do not hold. In the following, we will present an approach which addresses these issues. We will find that the actual sample complexity still coincides with our informal derivation here, which is why we included the above discussion to provide some intuition on the form of the sample complexity.

\subsubsection{Formal Analysis}
We will first argue in \Cref{rem:cmi_ub_small_cond_cutoff} that the condition we imposed on the size of $\mathsf{Poi}(\tilde N_S)$ (see \Cref{alg:cmitestsmall}) can be neglected in our analysis.

\begin{lemma}\label{rem:cmi_ub_small_cond_cutoff} 
The cut-off imposed on $\mathsf{Poi}(\tilde N_S)$ in \Cref{alg:cmitestsmall} decreases the probability of success by at most $\Pr[\sum_iM_i> N_S]$.
\end{lemma}

\begin{proof}
 
Let $A$ be the event that the algorithm succeeds. Then, if $\Pr[A]\geq 1-\delta_1$, $\Pr[\sum_iM_i\leq N_S]\geq 1-\delta_2$ and $\delta_1+\delta_2<1$, Bayes rule implies that 
\begin{align}
    \Pr\Big[A\Big|\sum_iM_i\leq N_S\Big]&=1-\frac{\Pr[\bar A,\sum_iM_i\leq N_S]}{\Pr[\sum_iM_i\leq N_S]}
    \\
    &\geq 1-\frac{\Pr[\bar A]}{\Pr[\sum_iM_i\leq N_S]}\geq 1-\frac{\delta_1}{1-\delta_2}\geq 1-\delta_1-\delta_2.
\end{align}
\end{proof}

We denote by $X_{abc}$ the number of times \algSmallP{} returns the sample $(a,b,c)$, and define $Y_{abc}$ analogously for \algSmallQ. Further, let $X_b=\sum\limits_{a,c}X_{abc}$.

For simplicity, we write $x_{b}:=\tilde N_Sp_b$ in the following. Importantly for our analysis, our choice for $B_S$ guarantees that $x_b\leq 1$. 
In the following, for ease of readability, we will often omit the index $abc$ when it is clear from the context. When providing \Cref{alg:simsmallfirst} and \Cref{alg:simsmallsecond} with $\mathsf{Poi}(\tilde N_S)$ samples each, we have the following.
\begin{align}
    \Pr[X_{abc}=k]&=\sum_{\ell=k}^{\infty}\Pr[X_{abc}=k|X_b=\ell]\Pr[X_b=\ell]
    \\
    &=\sum_{\ell=k}^{\infty}\binom{\ell}{k}p_{ac|b}^{\ell}(1-p_{ac|b})^{k-\ell}\left(\frac{x_{b}^{2\ell}e^{-x_{b}}}{(2\ell)!}+\frac{x_{b}^{2\ell+1}e^{-x_{b}}}{(2\ell+1)!}\right).
\label{eq:cmi_smallpb_distributionx}
\end{align}
For \algSmallP{}  (\Cref{alg:simsmallfirst}), two samples from $b$ result in an output for $(a,b,c)$ if the first sample reads $(a,b,c)$, and the second one is arbitrary, resulting in $p_{ac|b}$. For \algSmallQ{} (\Cref{alg:simsmallsecond}), two samples from $b$ result in an output for $(a,b,c)$ if the first sample has the $A$-coordinate equal to $a$, and the second one has $C$-coordinate equal to $c$, resulting in $q_{ac|b}:=p_{a|b}p_{c|b}$. 

Note that our process outputs $\ell$ samples of a specific $b$-value exactly if we witnessed either $2\ell$ or $2\ell+1$ samples with this specific $b$ from the original distribution. Analogously,
\begin{align}
    \Pr[Y_{abc}=k]&=\sum_{\ell=k}^{\infty}\Pr[Y_{abc}=k|Y_b=\ell]\Pr[Y_b=\ell]
    \\
    &=\sum_{\ell=k}^{\infty}\binom{\ell}{k}q_{ac|b}^{\ell}(1-q_{ac|b})^{k-\ell}\left(\frac{x_{b}^{2\ell}e^{-x_{b}}}{(2\ell)!}+\frac{x_{b}^{2\ell+1}e^{-x_{b}}}{(2\ell+1)!}\right).
\label{eq:cmi_smallpb_distributiony}
\end{align}
It is clear that $\forall b\in [d_B],\forall k\in \mathbb{N}:\E[X_b^k]=\E[Y_b^k]$.

Now we will bound the expectations of the random variables $X_{abc}, Y_{abc}, X_b, Y_b$ and other related quantities, whose proof will be presented in \Cref{sec:cmi_ub_app}.

\begin{restatable}{lemma}{lemprocessproperties}
\label{lem:cmi:small:xy}
    For $X_{abc}$ and $Y_{abc}$ distributed according to \eqref{eq:cmi_smallpb_distributionx}, and \eqref{eq:cmi_smallpb_distributiony}, respectively, and $q_{ac|b}=p_{a|b}p_{c|b}$ as defined before, the following hold: 
    \begin{enumerate}[label=(\roman*)]
        \item
        $\E[X_{b}]= \E[Y_b]=\frac{\tilde N_Sp_{b}}{2}-\frac{e^{-\tilde N_Sp_{b}}}{2}\sinh(\tilde N_Sp_{b})$, and $\E[X_{b}^2] = \E[Y_b^2] =\frac{\tilde N_S^2p_{b}^2}{4}-\frac{\tilde N_Sp_{b}e^{-2\tilde N_Sp_{b}}}{4}+\frac{e^{-\tilde N_Sp_{b}}}{4}\sinh(\tilde N_Sp_{b})$.
        
        \item 
        $\E[X_{abc}]=p_{ac|b}\E[X_{b}]$, and
        $\E[X_{abc}^2] = p_{ac|b}(1-p_{ac|b})\E[X_{b}]+p_{ac|b}^2\E[X_{b}^2]$.
        
        Similarly,
        $\E[Y_{abc}]=q_{ac|b}\E[X_{b}]$, and
        $\E[Y_{abc}^2] = q_{ac|b}(1-q_{ac|b})\E[X_{b}]+q_{ac|b}^2\E[X_{b}^2]$.
        
        \item
        For $(a,c)\neq (a',c')$, $\E[X_{abc}X_{a'bc'}] \leq 12p_{ac|b}p_{a'c'|b}\tilde N_S^3p_b^3$, and $\E[Y_{abc}Y_{a'bc'}] \leq 12q_{ac|b}q_{a'c'|b}\tilde N_S^3q_b^3$.
    \end{enumerate}
\end{restatable}
Our tester will use the following estimator, which uses two independent multisets of samples $X, X'$, generated by running two instances of \algSmallP{} (\Cref{alg:simsmallfirst}) on different sets of samples from $P_{ABC}$. Analogously, our estimator also requires two sets $Y, Y'$ generated by different instances of \algSmallQ{} (\Cref{alg:simsmallsecond}). Then, our tester determines the variable $Z$ according to
\begin{align}
    \label{eq:cmi_small_z_zabc}
    Z=\sum_{abc}Z_{abc},\quad Z_{abc}:=X_{abc}X'_{abc}-2X_{abc}Y_{abc}+Y_{abc}Y'_{abc}.
\end{align}
We will calculate $\E[Z]$ and $\Var[Z]$, which we will express in terms of $X$ and $Y$. For better readability, we write $i$ instead of $abc$. 
Note that a very interesting property of this estimator is that, no matter how $X$ and $Y$ are distributed, we have
\begin{equation}\label{eqn:estzdefi}
    \mathbb{E}[Z]=\sum_i\mathbb{E}[X_i]\mathbb{E}[X_i']-2\mathbb{E}[X_i]\mathbb{E}[Y_i]+\mathbb{E}[Y_i]\mathbb{E}[Y_i']=\sum_i\left(\mathbb{E}[X_i]-\mathbb{E}[Y_i]\right)^2,
\end{equation}
which is zero if and only if $\forall i: \E[X_i]=\E[Y_i]$. For our specific algorithms \algSmallP \ and \algSmallQ, we have the following results about the expectation of $Z$.

\begin{lemma}\label{lem:cmi:small:exp}
    For $i\in[4]$, let $\mathcal{S}_i$ be multisets of i.i.d.\ samples where $|\mathcal{S}_i|$ is determined by drawing from $\mathsf{Poi}(\tilde N_S)$. Let $X=\algSmallP(\mathcal{S}_1,B_S)$, $X'=\algSmallP(\mathcal{S}_2,B_S)$ (\Cref{alg:simsmallfirst}), and $Y=\algSmallQ(\mathcal{S}_3,B_S)$, $Y'=\algSmallQ(\mathcal{S}_4,B_S)$ (\Cref{alg:simsmallsecond}), and $Z$ constructed from $X$, $X'$, $Y$, and $Y'$ according to \eqref{eq:cmi_small_z_zabc}. Then
    \begin{enumerate}[label=(\roman*)]
    \item $D_H^2(\subD{P}{S}{ABC},\subD{Q}{S}{ABC})\geq \eps_S$ implies $\E[Z] \geq {\tilde N_S^4\eps_S^4}/{(4^7d_Ad_B^3d_C)}$, 
    \item $\subD{P}{S}{ABC}=\subD{Q}{S}{ABC}$ implies $\E[Z]=0$. 
    \end{enumerate}
\end{lemma}

\begin{proof}
Let us start with the proof of $(i)$.

\begin{enumerate}[label=(\roman*)]
    \item  Recall the definitions of $Z$ and $Z_{abc}$ from \eqref{eq:cmi_small_z_zabc}. Then, for a specific index $abc$, we have:
    \begin{align}
    \E[Z_{abc}]&=\left(\E[X_{abc}]-\E[Y_{abc}]\right)^2.\label{eqn:cmi:small:expxy}
    \end{align}
    
    Assume without loss of generality, $p_{ac|b}>q_{ac|b}$ (otherwise flip the roles of $p$ and $q$). Then, using \Cref{lem:cmi:small:xy}, we find that
    \begin{align}
        \E[X_{abc}]-\E[Y_{abc}]&=(p_{ac|b}-q_{ac|b})\left[\frac{\tilde N_Sp_b}{2}-\frac{1}{2}\sinh(\tilde N_Sp_b)e^{-\tilde N_Sp_b}\right]
        \\
        &=(p_{ac|b}-q_{ac|b})\left[\frac{\tilde N_Sp_b}{2}-\frac{1}{4}+\frac{1}{4}e^{-2\tilde N_Sp_b}\right]
        \\
        &\geq\frac{p_{ac|b}-q_{ac|b}}{4}\tilde N_S^2p_b^2,
    \end{align}
    where in the last step, we used that for $x<1$, $e^{-x}\geq 1-x+x^2/4$, according to \Cref{lemma:bounds_exp}.
    
    Next, we introduce $S_+:=\{b\in S| p_b\geq \eps_S/(4d_B)\}$ and $S_-:=\{b\in S| p_b< \eps_S/(4d_B)\}$. So combining the above with \eqref{eqn:cmi:small:expxy}, and using $q_{abc}=q_{ac|b}p_b$, we get 
    \begin{align}
        \E[Z]&=\sum_{b\in S,a,c}\frac{(p_{ac|b}-q_{ac|b})^2}{16}\tilde N_S^4p_b^4
        \\
        &=\sum_{\substack{a,c\\ b\in S_+}}\frac{(p_{ac|b}-q_{ac|b})^2}{16}\tilde N_S^4p_b^4+\sum_{\substack{a,c\\ b\in S_-}}\frac{(p_{ac|b}-q_{ac|b})^2}{16}\tilde N_S^4p_b^4
        \\
        &\geq \frac{\tilde N_S^4p_b^2}{16}\|\subD{P}{S_+}{ABC}-\subD{Q}{S_+}{ABC}\|_2^2
        \\
        &\geq  \frac{\tilde N_S^4\eps_S^2}{16\cdot 4^2d_B^2}\|\subD{P}{S_+}{ABC}-\subD{Q}{S_+}{ABC}\|_2^2.
    \end{align}
    In the next step, we use \Cref{fact:relations_distances}, which states that for any distributions $P$, $Q$, $D_H^2(P,Q)\leq 2\sqrt{d_A d_B d_C \|P-Q\|_2^2}$ holds. Thus we can write the above as the following 
    \begin{align}
        \E[Z]&\geq \frac{\tilde N_S^4\eps_S^2}{4^6d_Ad_B^3d_C}\left(D_H^2(\subD{P}{S_+}{ABC},\subD{Q}{S_+}{ABC})\right)^2.
    \end{align}
    Finally, using the fact that for any $b\in S_-$, $p_b\leq\eps_S/(4d_B)$, we can say that 
    \begin{equation}
        D_H^2(\subD{P}{S_-}{ABC},\subD{Q}{S_-}{ABC})\leq \frac{\eps_S}{4d_B}\sum_{b\in S_-,a,c}D_H^2(P_{AC|B=b},Q_{AC|B=b})\leq \frac{\eps_S}{4d_B}\cdot 2d_B=\frac{\eps_S}{2}.
    \end{equation}
    This implies that if $D_H^2(\subD{P}{S}{ABC},\subD{Q}{S}{ABC})\geq \eps_S$, then $D_H^2(\subD{P}{S_+}{ABC},\subD{Q}{S_+}{ABC})\geq \eps_S/2$, and we find 
    \begin{align}
        \E[Z]&\geq \frac{\tilde N_S^4\eps_S^4}{4^7d_Ad_B^3d_C}.
    \end{align}
    
    \item Now let us consider the case when $\subD{P}{S}{ABC}=\subD{Q}{S}{ABC}$. In this case, $p_{ac|b}=p_{a|b}p_{c|b}=q_{ac|b}$, such that \eqref{eq:cmi_smallpb_distributionx} and \eqref{eq:cmi_smallpb_distributiony} are identical. Following the definition of the estimator $Z$ as stated in \eqref{eqn:estzdefi}, we see that $\E[Z]=0$. 
\end{enumerate}
\end{proof}

The variance of $Z$ (defined in \Cref{eq:cmi_small_z_zabc}) can be bounded as follows.

\begin{restatable}{lemma}{lemcmismallvar}
    \label{lem:cmi:small:var}
    For $i\in[4]$, let $\mathcal{S}_i$ be multisets of i.i.d.\ samples where $|\mathcal{S}_i|$ is determined by drawing from $\mathsf{Poi}(\tilde N_S)$. Let $X=\algSmallP(\mathcal{S}_1,B_S)$, $X'=\algSmallP(\mathcal{S}_2,B_S)$ (\Cref{alg:simsmallfirst}), and $Y=\algSmallQ(\mathcal{S}_3,B_S)$, $Y'=\algSmallQ(\mathcal{S}_4,B_S)$ (\Cref{alg:simsmallsecond}), and $Z$ constructed from $X$, $X'$, $Y$, and $Y'$ according to \eqref{eq:cmi_small_z_zabc}. Then 
    \begin{equation}
        \Var[Z]\leq 2\cdot 10^3\left(\|\subD{P}{S}{ABC}\|_2^2+\|\subD{Q}{S}{ABC}\|_2^2\right)\tilde N_S^2.
    \end{equation}   
\end{restatable}

We will prove \Cref{lem:cmi:small:var} in \Cref{sec:cmi_ub_app}. Assuming the above lemma holds, let us proceed to prove \Cref{lem:cmi:small}, which we restate for better readability.

\lemcmismall*

\begin{proof}
 
From \Cref{lem:cmi:small:exp}, we know that when $\subD{P}{S}{ABC}$ and $\subD{Q}{S}{ABC}$ are far, we have 
\begin{equation}
    \E[Z]\geq \frac{\tilde N_S^4\eps_S^4}{4^7 d_Ad_B^3d_C}.
\end{equation}
Next, we bound the variance of $Z$. For both the distributions $P_{ABC}$ and $Q_{ABC}$, we can bound analogously
\begin{equation}
    \|\subD{P}{S}{ABC}\|_2^2\leq \sum_{abc\in S}p_{ac|b}^2p_b^2\leq \sum_{b\in S}p_b^2=\|\subD{P}{S}{B}\|_2^2 \leq \min\left\{\frac{1}{\tilde N_S},\frac{d_B}{\tilde N_S^2}\right\},
\end{equation}
and
\begin{equation}
    \|\subD{Q}{S}{ABC}\|_2^2\leq \sum_{abc\in S}q_{ac|b}^2q_b^2\leq \sum_{b\in S}q_b^2=\|\subD{Q}{S}{B}\|_2^2 \leq \min\left\{\frac{1}{\tilde N_S},\frac{d_B}{\tilde N_S^2}\right\},
\end{equation}
such that from \Cref{lem:cmi:small:var}, it follows that
\begin{equation}\label{eqn:varlb}
    \Var[Z]\leq 2\cdot 10^3\min\{\tilde N_S, d_B\}.
\end{equation}
We will apply Chebyshev's inequality (see \Cref{lem:chebyshev}). For this purpose, we simply set $t=\E[Z]/2$. Thus we need the following in order to ensure that the output of our estimator allows us to reliably distinguish the two cases:
\begin{equation}\label{eqn:varub}
    \frac{\Var[Z]}{\E[Z]^2/4}\leq \frac{1}{100}\implies \sqrt{\min\{\tilde N_S, d_B\}}\leq \frac{\tilde N_S^4\eps^4}{10\cdot (2\cdot 10)^{3/2}\cdot2\cdot 4^7d_Ad_B^3d_C}.
\end{equation}

For readability, let $\alpha:=10\cdot (2\cdot 10)^{3/2}\cdot2\cdot 4^7<10^{9}$.
Combining \eqref{eqn:varlb} with \eqref{eqn:varub}, we have two different ways to choose $\tilde N_S$:
\begin{itemize}
    \item \textbf{Case $1$:}
\begin{align}
    &\frac{\tilde N_S^4\eps^4}{\alpha d_Ad_B^3d_C}\geq d_B^{1/2}\quad \implies \tilde N_S\geq \alpha\frac{d_A^{1/4}d_B^{7/8}d_C^{1/4}}{\eps_S},   
\end{align}

\item \textbf{Case $2$:}
\begin{align}
    &\frac{\tilde N_S^4\eps^4}{\alpha d_Ad_B^3d_C}\geq \tilde N_S^{1/2}\quad \implies \tilde N_S\geq \alpha\frac{d_A^{2/7}d_B^{6/7}d_C^{2/7}}{\eps_S^{8/7}}.   
\end{align}
    
\end{itemize}
The total number of samples that is sufficient to correctly test the small regime for conditional independence with probability of success at least $99/100$ is thus
\begin{equation}
    N_S\geq 10^{10}\min\left\{\frac{d_A^{1/4}d_B^{7/8}d_C^{1/4}}{\eps_S},\frac{d_A^{2/7}d_B^{6/7}d_C^{2/7}}{\eps_S^{8/7}}\right\},
\end{equation}
where an additional factor $8$ comes from the fact that $N_S=\tilde N_S/8$.
\end{proof}

\subsection{Testing the Large Regime}
\label{sec:cmi_ub:large}
In this subsection, we will prove the sample complexity of the large regime. Recall that $\ceq$ is introduced in \Cref{cor:equalence_l2_algo} to abstract the implicit constant in a testing subroutine.

\begin{restatable}{lemma}{lemcmilarge}\label{lem:cmi:large}
$\mathsf{CMITestingLargeRegime}$ (\Cref{alg:cmitestlarge}) correctly tests the large regime of $P_{ABC}$ for conditional independence with probability at least $99/100$ if called using
\begin{equation}    
    N_L:= 10^6\ceq \log\left(\frac{d_Ad_Bd_C}{\eps_L}\right)^7\max\left\{\min\left\{\frac{d_A^{3/4}d_B^{3/4}d_C^{1/4}}{\eps_L},\frac{d_A^{2/3}d_B^{2/3}d_C^{1/3}}{\eps_L^{4/3}}\right\},\frac{d_A^{1/2}d_B^{3/4}d_C^{1/2}}{\eps_L},N_S\right\}
\end{equation} 
samples from $P_{ABC}$, and a set $B_L$ such that for all $b\in B_L: p_b \geq 1/N_S$, for $N_S$ as defined in \Cref{lem:cmi:small} and $d_A\geq d_C$. 
\end{restatable}

\begin{proof}
\Cref{lem:cmi:large} follows directly from the proof of correctness of \Cref{alg:cmitestlarge}, presented in \Cref{lemma:algo_cmi_correctness}. 
\end{proof}

Before we describe the algorithm $\mathsf{CMITestingLargeRegime}$, let us first present a few subroutines that will be used in $\mathsf{CMITestingLargeRegime}$.

\subsubsection{Subroutines}
Unlike in the small regime, where we were not able to simulate the desired distribution exactly, we will simulate samples from a distribution $\subD{\tilde Q}{}{ABC}$ which satisfies $\subD{\tilde Q}{L}{ABC}=(P_{AB}P_{C|B})^L$. However, other than for independence testing, this is no longer straightforward. We now describe an algorithm $\mathsf{SimABC_{CI}Large}$ to sample from $P_{AB}P_{C|B}$ in the large regime and prove its correctness.

\begin{algorithm}[H]
\LinesNumbered
\DontPrintSemicolon
\setcounter{AlgoLine}{0}
\caption{$\mathsf{SimABC_{CI}Large}$}
\label{alg:cond_indep_sampling_large}
\KwIn{A multiset $\mathcal{S}$ of triplets from $A\times B\times C$, $B_L\subseteq [d_B]$}

\KwOut{A multiset $\mathcal{S}_{\text{out}}$ of triplets from $A\times B\times C$ or `\textbf{Abort}'}

Pick disjoint multisets $\mathcal{S}_1$ and $\mathcal{S}_2$ from $\mathcal{S}$ such that $|\mathcal{S}_1|=4|\mathcal{S}|/5$, $|\mathcal{S}_2|=|\mathcal{S}|/5$

Define $d_B$ empty queues $Q[1]$, ... , $Q[d_B]$, $\mathcal{S}_{\text{out}}\leftarrow\varnothing$

\While{$\mathcal{S}_1\neq \varnothing$}{

$(a,b,c)\leftarrow $ remove a random triplet from $\mathcal{S}_1$

insert $a$ into $Q[b]$
}

\While{$\mathcal{S}_2\neq \varnothing$}{

$(a,b,c)\leftarrow$ remove a random triplet from $\mathcal{S}_2$

\If{$b\in B_L$}{

\If{$Q[b]=\varnothing$}{
  \Return `\textbf{Abort}' 
 }

    $a'\leftarrow$ dequeue first element from $Q[b]$ 
 
    $\mathcal{S}_{\text{out}}\leftarrow \mathcal{S}_{\text{out}}\cup \{(a',b,c)\}$
}\Else{
    $\mathcal{S}_{\text{out}}\leftarrow \mathcal{S}_{\text{out}}\cup \{(a,b,c)\}$ \Comment*[r]{No modification required if $b\notin B_L$}
}
}

\Return $\mathcal{S}_{\text{out}}$.
\end{algorithm} 

The correctness of \Cref{alg:cond_indep_sampling_large} is presented in the following lemma.

\begin{lemma}
\label{lemma:cmi_ub:simulate_large}
    Let $\zeta \in (0,1)$, $\xi\in (0,1)$, $P_{ABC}$ be a tripartite distribution, and $K:=\{(a,b,c)|p_b\geq \xi\}$. Using $N\geq 8\log(4 d_B/\zeta)/\xi$ samples from $P_{ABC}$, \Cref{alg:cond_indep_sampling_large} simulates $N/5$ samples from $\tilde Q_{ABC}$, a distribution which satisfies $\subD{\tilde Q}{K}{ABC}=\subD{ (P_{AB}P_{C|B})}{K}{}$ and $\tilde Q_B=Q_B$, with probability at least $1-\zeta$.
\end{lemma}

\begin{proof}
    To simulate samples from the reference distribution $Q=P_{AB}P_{BC}/P_B$, we perform two phases of sampling:
    \begin{enumerate}
        \item In the first phase, we take $4N/5$ samples $(a_1,b_1,c_1), \ldots, (a_{N}, b_{N}, c_{N})$ from $P_{ABC}$, and sort their $A$-coordinate into different queues according to their $B$-coordinate.
    
        \item  In the second phase, we take the remaining set of $N/5$ samples from $P_{ABC}$. If we obtain $(a,b,c)$ as a sample where $b\in K$, we pick the first element from the respective $b$-queue, say, $a'$, and we return $(a',b,c)$. Samples of the form $(a,b,c)$ with $b\notin K$ will be returned directly. 
    \end{enumerate}
    Since $p_b$ for $b\in K$ is lower bounded by $\xi$, \Cref{lemma:learn_approx} guarantees that for a given $b$, the number of samples we see in the two phases is within a factor two of the respective expected value. Since
    \begin{equation}
        \mathbb{E}\left[|Q[b]|\right]\geq 4\mathbb{E}\left[\#\text{queries to }Q[b]\right],
    \end{equation}
    there are, with high probability, always enough samples in each queue. 

    To see that this procedure produces samples from $\subD{\tilde Q}{K}{ABC}$, we ask for the probability to obtain $(a',b,c) \in K$. First, the sample we draw needs to be of the form $(\cdot,b,c)$, which happens with probability $p_{bc}$. The probability to draw $a'$ from the $b$-queue is given by $p_{a'b}/p_b$, since we conditioned on $b$. Together, this results in 
    \begin{equation}
        \Pr[(a',b,c) \ \text{is obtained}]=\frac{p_{a'b}p_{bc}}{p_b}.
    \end{equation}
    Since we only test regions with non-rare $b$'s, it is irrelevant from where we sample them, as long as $p_b$'s for $b\in K$ are not affected.
    The equality on $B$, $\tilde Q_B=Q_B$ follows directly, since if we take a random sample from $\mathcal{S}_2$, the value of $b$ is never modified before adding the element to $\mathcal{S}_{\text{out}}$.
\end{proof}

We will also need a way to decide whether $\subD{Q}{K}{D}=\subD{P}{K}{D}$, in the regime where $\|\subD{Q}{K}{D}\|_2$ is considerably smaller than the error we are willing to tolerate. In this regime, it is more sample efficient to simply test whether $\|\subD{P}{K}{D}\|_2$ is sufficiently far below the threshold as well, instead of performing actual equivalence testing.

\begin{lemma}\label{lemma:cmi_ub_case_dist_small_l2}
Consider two unknown distributions $P_D$ and $Q_D$ defined over a set $D$, and $K\subseteq D$. Let $\eta \in (0,1)$ be such that $\|\subD{Q}{K}{D}\|_2\leq \eta/(10\sqrt{|K|})$. Using $O(\sqrt{|D|}+\sqrt{|K|}/\eta)$ samples, to give one of the following guarantees which holds with probability at least $2/3$:
    \begin{enumerate}[label=(\roman*)]
        \item $D_H^2(\subD{P}{K}{D},\subD{Q}{K}{D})\leq  \eta$,
        \item $\subD{P}{K}{D}\neq \subD{Q}{K}{D}$.
    \end{enumerate}
\end{lemma}

\begin{proof}
    For readability, let $\tau:=\eta/\sqrt{|K|}$. We will use \Cref{lemma:get_weight_l2}, to learn an estimate $c_P$ of $\|\subD{P}{K}{D}\|_2$ with error parameter $\tau/20$. We then show that
    \begin{enumerate}[label=(\roman*)]
        \item If $c_P\leq \frac{\tau}{4}$, then $D_H^2(\subD{P}{K}{D},\subD{Q}{K}{D})\leq  \sqrt{|K|}\tau$,
        \item If $c_P> \frac{\tau}{4}$, then $\|\subD{P}{K}{D}\|_2\neq \|\subD{Q}{K}{D}\|_2$ (and hence $\subD{P}{K}{D}\neq \subD{Q}{K}{D}$).
    \end{enumerate}
    Recall that \Cref{lemma:get_weight_l2} guarantees that with probability at least $2/3$, $\|\subD{P}{K}{D}\|_2/2\leq c_P\leq 2\|\subD{P}{K}{D}\|_2+\tau/20$ holds. Assume $c_P< \tau/4$, then 
    \begin{align}
        \tau> \frac{\tau}{10}+2\frac{\tau}{4}\geq c_Q+2c_P&
        \geq \frac{1}{\sqrt{|K|}}\left(\|\subD{Q}{K}{D}\|_1+\|\subD{P}{K}{D}\|_1\right)
        \geq \frac{1}{\sqrt{|K|}}D_H^2(\subD{Q}{K}{D},\subD{P}{K}{D}).
    \end{align}
    Now consider the case when $c_P>\tau/4$. Since $\tau\geq 10 c_Q$, we can say that 
    \begin{equation}
        2\|\subD{P}{K}{D}\|_2+\frac{\tau}{20}> \frac{\tau}{4}=\frac{\tau}{20}+\frac{\tau}{5}\geq \frac{\tau}{20}+2\|\subD{Q}{K}{D}\|_2\implies \|\subD{P}{K}{D}\|_2> \|\subD{Q}{K}{D}\|_2.
    \end{equation}
\end{proof}

As a corollary of the above lemma, we have the following.

\begin{corollary}
\label{cor:equiv_small_algo}
Consider two unknown distributions $P$ and $Q$ defined over a set $[d]$, $S \subseteq [d]$ and $\eps, \delta,\zeta \in (0,1)$. There exists an algorithm \textnormal{\algLnorm}$(\mathcal{S}_P,S,\zeta,\eps, \delta)$, which takes a multiset $\mathcal{S}_P$ of i.i.d.\ samples from $P$ with $\|\subD{Q}{S}{}\|_2\leq \zeta$, and distinguishes with probability at least $1-\delta$ whether $D_H^2(\subD{P}{S}{},\subD{Q}{S}{})\leq \eps$ or $\subD{P}{S}{}\neq\subD{Q}{S}{}$, provided $|\mathcal{S}_P|\geq \cltwo\sqrt{d}/\eta$ samples, for some instance independent constant $\cltwo$.
\end{corollary}

\subsubsection{Description of the Algorithm}
We first state the algorithm testing the large regime, \Cref{alg:cmitestlarge}, which we then explain. The proof of correctness also gives our sample complexity, and we will make use of the tools presented in \Cref{sec:mi_ub_prelim}. There will be three stages. First, we partition $\subD{P}{L}{ABC}$ into polylogarithmic many \emph{categories} $L_{ij}^k$, illustrated in \Cref{fig:cmi_partition}. After simulating a pool of samples from $Q_{ABC}=P_{AB}P_{C|B}$ using \Cref{alg:cond_indep_sampling_large}, we perform equivalence testing individually on each category. 
Details of the sampling are described in \Cref{lemma:cmi_ub:simulate_large}, which is built on the fact that $p_b \geq 1/N_S$ for $b\in B_L$. We now define a few quantities we will use in \Cref{alg:cmitestlarge}, which depend on the size of $\pL{i}{j}{k}$, which are sets we define inside the algorithm, as well as parameters $k_A:=\lceil\log(1/M)\rceil$, $k_B:=\lceil\log(1/\nu)\rceil$, $k_C:=\lceil\log(1/M)\rceil$, for $M$ to be defined, and $k_{ABC}:=\log^3(d_Ad_Bd_C/\eps)$. Note that there are $(k_A+1)(k_B+1)(k_C+1)<k_{ABC}$ categories in total.
\begin{align}
    \gamma(i,j,k)&:=
    \begin{cases}
         \sqrt{\frac{\eps_L e^{k-(i+j+3)}}{e^9k_{ABC}}} & \text{if } i<k_A, j<k_C,
        \\
        \frac{\eps_L}{4k_{ABC}|\pL{i}{j}{k}|^{1/2}} & \text{if either } i=k_A \text{ or } j=k_C,
        \\
         \frac{e^{k+1}\eps_L}{2|B_k|\sqrt{|\pL{i}{j}{k}|}} & \text{if }i=k_A, j=k_C,
    \end{cases}
    \\
    b(i,j,k)&:=e^9\min\left\{\sqrt{|\pL{i}{j}{k}|}e^{k+4-(i+j)},\sqrt{e^{k+4-(i+j)}}\right\},
    \\
    \eps_L(i,j,k)&:=\eps_L/\sqrt{|\pL{i}{j}{k}|}.
\end{align}
These will follow directly from our derivations. Further, let us define 

\begin{equation}
    \label{eq:cmi_ub:full_sc}
    N_L:=10^3  k_{ABC}\log( k_{ABC})\max \{N_{L,\text{heavy}}, N_{L,\text{mixed}}, N_{L,\text{light}}   \},
\end{equation}
where 
\begin{align}
    N_{L,\text{heavy}} &:= 10^{7} \ceq k_{ABC}   \frac{\sqrt{d_Ad_Bd_C}}{\eps_L}, \quad (\text{\Cref{cmi:large:heavy}}),\label{eq:cmi_ub:n_heavy_alg}\\
    N_{L,\text{mixed}} &:= 10^3 \ceq k_{ABC}^2 \min\left\{\frac{d_A^{3/4}d_B^{3/4}d_C^{1/4}}{\eps_L},\frac{d_A^{2/3}d_B^{2/3}d_C^{1/3}}{\eps_L^{4/3}}\right\}, \quad (\text{\Cref{cmi:large:mixed}}),\\    
    N_{L,\text{light}} &:= 10\max\left\{\ceq  k_{ABC}\frac{d_A^{1/2}d_B^{3/4}d_C^{1/2}}{\eps_L},\frac{\log(c_1k_b)}{\nu}\right\}, \quad (\text{\Cref{cmi:large:light}})\label{eq:cmi_ub:n_light_alg}.
\end{align}

\begin{algorithm}
\LinesNumbered
\DontPrintSemicolon
\setcounter{AlgoLine}{0}
\caption{$\mathsf{CMITestingLargeRegime}$}
\label{alg:cmitestlarge}
\KwIn{$\mathcal{S}$, multiset of $N_L$  triplets from $A\times B\times C$, $\eps_L,\nu \in (0,1)$ $S_B\subseteq[d_B]$ \Comment*[r]{see~\eqref{eq:cmi_ub:full_sc}}}
\KwOut{`{\bf Yes}' or `{\bf No}'}

\Comment*[l]{{\textbf{Step $1$: Partition $A\times B \times C$}}}
\vspace{5pt}

$M\leftarrow\max\{N_{L,\text{heavy}},N_{L,\text{mixed}},N_{L,\text{light}}\}$, $\delta\leftarrow\frac{1}{10^3k_{ABC}}$ \Comment*[r]{see \eqref{eq:cmi_ub:n_heavy_alg}-\eqref{eq:cmi_ub:n_light_alg}}

$\mathcal{S}_{AB}$, $\mathcal{S}_{BC}$, $\mathcal{S}_{B}$ $\leftarrow$ remove random disjoint multisets from $\mathcal{S}$ of size $8M k_{ABC}$ each

$\forall (a,b)\in A\times B: \hat p_{ab}\leftarrow$ empirical frequencies of $(a,b)$ in $\mathcal{S}_{AB}$

$\forall (b,c)\in B\times C: \hat p_{bc}\leftarrow$ empirical frequencies of $(b,c)$ in $\mathcal{S}_{BC}$

$\forall b\in B: \hat p_{b}\leftarrow$ empirical frequencies of $b$ in $\mathcal{S}_{B}$

$\forall i\in [k_A]_0,j\in [k_C]_0,k \in [k_B]_0$: define $\pL{i}{j}{k}$ according to \Cref{fig:cmi_partition} \Comment*[r]{partition $A\times B\times C$}

\vspace{5pt}
\Comment*[l]{{\textbf{Step $2$: Construct Reference Distribution and Simulate Samples}}}
\vspace{5pt}

$\hat Q_{ABC}(a,b,c)\leftarrow\hat p_{ab}\hat p_{bc}/\hat p_b$, $\mathcal{S}_{\text{sim}}\leftarrow$ remove random multiset of size $N_L/8$ from $\mathcal{S}$

$\mathcal{S}_{CI}\leftarrow$ \texttt{SimABC$_{CI}$Large}($\mathcal{S}_{\text{sim}}$) \Comment*[r]{\Cref{alg:cond_indep_sampling_large}, \Cref{lemma:cmi_ub:simulate_large}, $|\mathcal{S}_{CI}|=N_L/40$}

\vspace{5pt}
\Comment*[l]{{\textbf{Step $3$: Equivalence Testing}}}
\vspace{5pt}

\For{\textnormal{all} $i,j,k$}{

    $\mathcal{S}_P$, $\mathcal{S}_Q$ $\leftarrow$ remove random multisets from $\mathcal{S}$ and $\mathcal{S}_{CI}$, respectively, of size $M$ each

    \If{$i<k_A$, $j<k_C$, $\|\subD{\hat Q}{\pL{i}{j}{k}}{ABC}\|_2< \frac{\eps_L(i,j,k)}{10 e^9}$}
    {
        \If{\textnormal{\algLnorm}$(\mathcal{S}_P, L_{ij}^k, e^9\|\subD{\hat Q}{\pL{i}{j}{k}}{ABC}\|_2, \eps_L(i,j,k), \delta)$=\textnormal{`{\bf No}'}}{
        \Return `{\bf No}'\Comment*[r]{\Cref{cor:equiv_small_algo}}
    }
    }
    \Else
    {
        \If{$i=k_A$ \textnormal{\&} $j=k_C$}{
           $\mathcal{S}_P\leftarrow \{(a,b,c)\in \mathcal{S}_P|b\in B_L\}$, $\mathcal{S}_Q\leftarrow \{(a,b,c)\in \mathcal{S}_Q|b\in B_L\}$ \Comment*[r]{\Cref{sec:cmi_ub:sc_light}}
        }
        \If{\textnormal{\algEquiv}$(\mathcal{S}_P, \mathcal{S}_Q, \pL{i}{j}{k}, b(i,j,k), \gamma(i,j,k),\delta)$=\textnormal{`{\bf Far}'}}{
        
        \Return `{\bf No}' \Comment*[r]{\Cref{cor:equalence_l2_algo}}
        }
        
    }
}

\Return `{\bf Yes}'

\end{algorithm}

Similar to the independence testing, our overall idea is to partition the distribution such that coordinates of similar weight are grouped together (see \Cref{fig:cmi_partition}), and to perform equivalence testing for each of the pieces. To achieve this, we first learn the marginals on $P_{AB}$, $P_{BC}$, and $P_B$ approximately, that is up to a constant factor. 

Now we define \emph{categories} $\pL{i}{j}{k}:=\cup_{b\in B_k}A_i^b\times b \times C_j^b$, splitting $P_{ABC}$ into $(k_A+1)(k_B+1)(k_C+1)=O(\mathsf{polylog}(d_Ad_Bd_C/\eps))$ many categories, which we can test individually for precision $\eps/k_{ABC}$ (pigeonhole principle), at the cost of a polylogarithmic overhead.
The partitioning of $\subD{P}{L}{ABC}$ is illustrated in \Cref{fig:cmi_partition}.

    \begin{figure}[H]
        \begin{center}
    \scalebox{0.75}{
    \begin{tikzpicture}

\def\binX{1}
\def\tcol{blue}


\def\angleP{pi/3}
\def\sf{1.5}
\def\hX{\binX*cos(\angleP)} \def\hY{\binX*sin(\angleP)}

\def\spYtop{4.5}
    
\foreach \j in {0,...,4}{
\foreach \i in {4,...,0}
{
    \def\spx{(\j*\binX+\hX*\i)}
    \def\spy{\spYtop+\hY*\i}
    
    \tkzDefPoint(\spx*\sf,\spy){A} 
    
    \tkzDefPoint((\spx+\binX)*\sf,\spy){B} 
    \tkzDefPoint((\spx+\binX+\hX)*\sf,\spy+\hY){C} 
    \tkzDefPointWith[colinear= at C](B,A) \tkzGetPoint{D}

    \ifthenelse{\i = 4}{
        \ifthenelse{\j = 4}{\def\tcol{red}}{\def\tcol{purple}}
    }{
        \ifthenelse{\j = 4}{\def\tcol{purple}}{\def\tcol{blue}}
    } 
    \tkzDrawPolygon[fill=\tcol!20](A,B,C,D)

}
}

\tkzDefPoint(\hX*\sf*5-0.8,4.5+5*\hY+0.2){X};
\node at (X) {$\log\left(\frac{1}{\hat p_{ab}}\right)$};

\tkzDefPoint(\hX*\sf*5,4.5+5*\hY){X};
\tkzDefPoint((\binX+\hX)*\sf*5,4.5+5*\hY){Y};

\draw[dashed, white, thick] (X) -- (Y);
\tkzDefPoint((\binX+\hX)*\sf*5,5*\hY-1){X};

\tkzDefPoint(\hX*\sf+0.05,4.5+\hY+0.35){U};
\node[anchor=east] at (U) {$1$};
\tkzDefPoint(2*\hX*\sf+0.05,4.5+2*\hY+0.35){U};
\node[anchor=east] at (U) {$i$};
\tkzDefPoint(3*\hX*\sf+0.05,4.5+3*\hY+0.35){U};
\node[anchor=east] at (U) {$i+1$};
\tkzDefPoint(4*\hX*\sf+0.15,4.5+4*\hY+0.35){U};
\node[anchor=east] at (U) {$k_A$};

\tkzDefPoint(2.5*\binX*\sf+\hX*\sf*2.5,4.5+2.5*\hY){U};
\node at (U) {$\pL{i}{j}{0}$};

\tkzDefPoint(\hX*\sf*5.5,4.5+4.5*\hY){U};
\node at (U) {$\pL{k_A}{0}{0}$};

\tkzDefPoint(\hX*\sf*4.5,4.5+3.5*\hY){U};
\node at (U) {$...$};
\tkzDefPoint(2*\binX*\sf+\hX*\sf*4.5,4.5+3.5*\hY){U};
\node at (U) {$...$};

\tkzDefPoint(\hX*\sf*3.5,4.5+2.5*\hY){U};
\node at (U) {$\pL{i}{0}{0}$};
\tkzDefPoint(\binX*\sf+\hX*\sf*3.5,4.5+2.5*\hY){U};
\node at (U) {$...$};

\tkzDefPoint(\hX*\sf*2.5,4.5+1.5*\hY){U};
\node at (U) {$...$};
\tkzDefPoint(2*\binX*\sf+\hX*\sf*2.5,4.5+1.5*\hY){U};
\node at (U) {$...$};

\tkzDefPoint(2.5*\binX*\sf+\hX*\sf*4.5,4.5+4.5*\hY){U};
\node at (U) {$\pL{k_A}{j}{0}$};

\tkzDefPoint(4.5*\binX*\sf+\hX*\sf*4.5,4.5+4.5*\hY){U};
\node at (U) {$\pL{k_A}{k_C}{0}$};


\def\binX{1.5}
\def\angleP{pi/2}
\def\sf{1}
\def\hX{\binX*cos(\angleP)} \def\hY{\binX*sin(\angleP)}

\foreach \j in {0,...,4}{
\foreach \i in {2,...,0}
{
    \def\spx{(\j*\binX+\hX*\i)}
    \def\spy{\hY*\i}
    
    \tkzDefPoint(\spx*\sf,\spy){A} 
    
    \tkzDefPoint((\spx+\binX)*\sf,\spy){B} 
    \tkzDefPoint((\spx+\binX+\hX)*\sf,\spy+\hY){C} 
    \tkzDefPointWith[colinear= at C](B,A) \tkzGetPoint{D}
    
    \ifthenelse{\j = 4}{
        \def\tcol{purple}
    }{
        \def\tcol{blue}
    } 
    \tkzDrawPolygon[fill=\tcol!20](A,B,C,D)

}
}

\node[anchor=east] at (0.05, 4.85) {$0$};
\node[anchor=east] at (\binX*\sf+0.1, 4.85) {$1$};
\node[anchor=east] at (2*\binX*\sf+0.1, 4.85) {$j$};
\node[anchor=east] at (3*\binX*\sf+0.15, 4.85) {$j+1$};
\node[anchor=east] at (4*\binX*\sf+0.1, 4.85) {$k_C$};

\node[anchor=east] at (0.0, 3.35) {$1$};
\node[anchor=east] at (0.05, 1.85) {$k_B$};

\node at (1.5*\sf*0.5, 0.75) {$\pL{0}{0}{k_B}$};
\node at (1.5*\sf*0.5, 1.5+0.75) {$...$};
\node at (1.5*\sf*0.5, 3+0.75) {$\pL{0}{0}{0}$};
\node at (1.5*\sf*1.5, 3+0.75) {$...$};
\node at (1.5*\sf*3.5, 3+0.75) {$...$};

\node at (1.5*\sf*1.5, 0.75) {$...$};
\node at (1.5*\sf*3.5, 0.75) {$...$};

\node at (1.5*\sf*2.5, 0.75) {$\pL{0}{j}{k_B}$};
\node at (1.5*\sf*2.5, 1.5+0.75) {$...$};
\node at (1.5*\sf*2.5, 3+0.75) {$\pL{0}{j}{0}$};

\node at (1.5*\sf*4.5, 0.75) {$\pL{0}{k_C}{k_B}$};
\node at (1.5*\sf*4.5, 1.5+0.75) {$...$};
\node at (1.5*\sf*4.5, 3+0.75) {$\pL{0}{k_C}{0}$};

\node at (-1,-0.6) {$\log\left(\frac{1}{\hat p_{b}}\right)$};

\def\binX{1.5}
\def\spXs{7.5}
\def\sf{2/3}
\def\angleP{pi/3}
\def\hX{\binX*cos(\angleP)} \def\hY{\binX*sin(\angleP)}

\foreach \j in {0,...,4}{
\foreach \i in {2,...,0}
{
    \def\spx{\spXs+\j*\hX}
    \def\spy{\binX*\i+\j*\hY*\sf}
    
    \tkzDefPoint(\spx,\spy){A} 
    
    \tkzDefPoint((\spx+\hX),\spy+\hY*\sf){B} 
    \tkzDefPoint((\spx+\hX),\spy+\hY*\sf+\binX){C} 
    \tkzDefPointWith[colinear= at C](B,A) \tkzGetPoint{D}
    
    \ifthenelse{\j = 4}{
        \def\tcol{red}
    }{
        \def\tcol{purple}
    } 

    \tkzDrawPolygon[fill=\tcol!20](A,B,C,D)

}
}

\tkzDefPoint(5*(3/2)+1,4){L} 
\node[fill=white,opacity=0.9, text opacity=1] at (L) {$\log\left(\frac{1}{\hat p_{bc}}\right)$};

\def\tcol{yellow}
\tkzDefPoint(0,-1){A} 
\tkzDefPoint(1.5*5,-1){B} 
\tkzDefPoint(1.5*5,0){C}
\tkzDefPointWith[colinear= at C](B,A) \tkzGetPoint{D}
\tkzDrawPolygon[fill=\tcol!20](A,B,C,D)

\tkzDefPoint(1.5*5,-1){A} 
\tkzDefPoint(1.5*5+5*\hX,-1+5*\hY*\sf){B} 
\tkzDefPoint(1.5*5+5*\hX,5*\hY*\sf){C}
\tkzDefPointWith[colinear= at C](B,A) \tkzGetPoint{D}
\tkzDrawPolygon[fill=\tcol!20](A,B,C,D)

\node at (1.5*2.5, -0.5) {Small regime $S$};

\draw[dashed, white, thick] (Y) -- (X);
\draw[dashed, white, thick] (\binX*5,-1) -- (X);
\draw[dashed, white, thick] (\binX*5,-1) -- (0,-1);

\def\angleP{pi/3}
\def\sf{1.5}
\def\hX{\binX*cos(\angleP)} \def\hY{\binX*sin(\angleP)}
\draw[->, very thick] (0,4.5) -- (5*1.5+0.35,4.5);
\draw[->, very thick] (0,4.5) -- (0,-1.35);
\tkzDefPoint(3.55*\hX*\sf,3.55*\hY+4.5){E}
\draw[->, very thick] (0,4.5) -- (E);

\node at (15,3.8) {$\begin{aligned}
\pL{i}{j}{k}&:=\left\{(a,b,c)\middle|{\scriptscriptstyle\begin{aligned}e^{-i-1}&\leq \hat p_{ab}< e^{-i}\\ e^{-j-1}&\leq \hat p_{bc}< e^{-j}\\e^{-k-1}&\leq \hat p_b <e^{-k}\end{aligned}}\right\}
\\
\pL{i}{k_C}{k}&:=\left\{(a,b,c)\middle|{\scriptscriptstyle\begin{aligned}e^{-i-1}&\leq \hat p_{ab}< e^{-i}
\\ 0&\leq \hat p_{bc}< e^{-k_C}\\e^{-k-1}&\leq \hat p_b <e^{-k}\end{aligned}}\right\}
\\
\pL{k_A}{j}{k}&:=\left\{(a,b,c)\middle|{\displaystyle\begin{aligned}0&\leq \hat p_{ab}< e^{-k_A} \\ e^{-j-1}&\leq \hat p_{bc}< e^{-j}\\e^{-k-1}&\leq \hat p_b <e^{-k}\end{aligned}}\right\}
\\
\pL{k_A}{k_C}{k}&:=\left\{(a,b,c)\middle|{\textstyle\begin{aligned} ~0&\leq \hat p_{ab}< e^{-k_A} \\  ~0&\leq \hat p_{bc}< e^{-k_C}\\e^{-k-1}&\leq \hat p_b <e^{-k}\end{aligned}}\right\}
\\
S&:=\left\{(a,b,c)\middle| 0\leq \hat p_b\leq e^{-k_B}\right\}
\end{aligned}$};

\end{tikzpicture}
    }
    \caption{\label{fig:cmi_partition} Partition of $d_A\times d_B\times d_C$ based on $\hat P_{AB}$, $\hat P_{BC}$, and $\hat P_{B}$. Indices $(a,b,c)$ of similar weight $\hat p_{ab}\hat p_{bc}/\hat p_b$ are grouped together in categories $L_{ij}^k$, which are used to perform piecewise equivalence testing with $P_{ABC}$. The axes are labeled according to the corresponding category, which is inverse logarithmic to the weight of the probabilities. The color of the categories indicate a different analysis of the sample complexity of testing the categories. The red and orange regimes dominate the sample complexity. The small regime is treated separately.}
        \end{center}
    \end{figure}
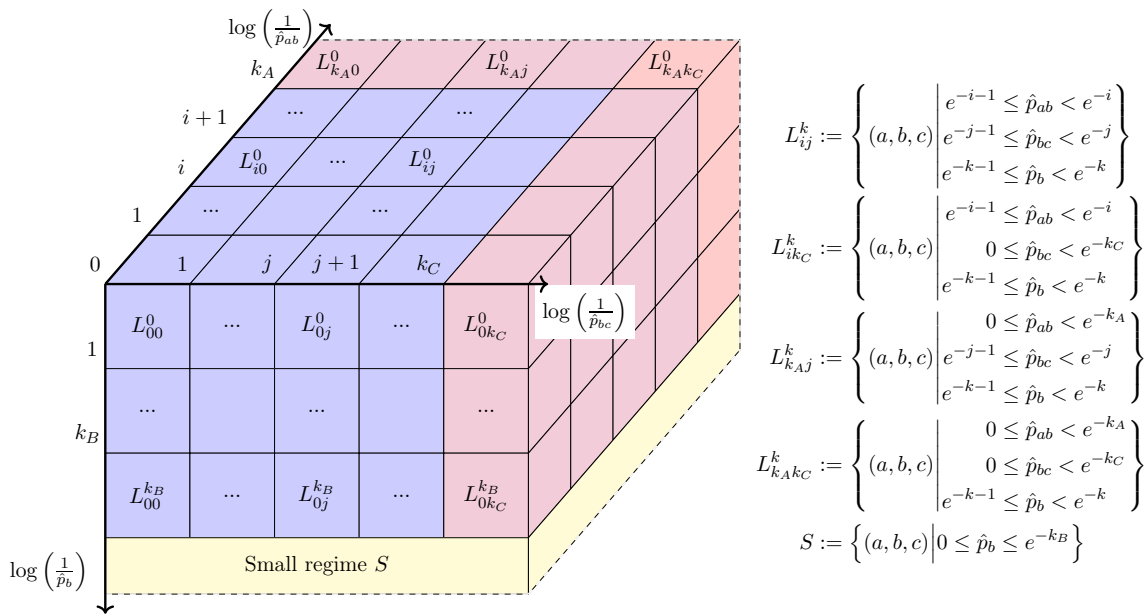

To perform these equivalence tests, we will need samples from our reference distribution $Q_{ABC}=P_{AB}P_{C|B}$. We will sample from a slightly different distribution, $\tilde Q_{ABC}$, which satisfies $\subD{\tilde Q}{L}{ABC}=\subD{ (P_{AB}P_{C|B})}{L}{}$ and $\tilde Q_B=Q_B$. Since we are only interested in ensuring equivalence on $L$ (or subsets of $L$), we can treat $Q$ and $\tilde Q$ as equivalent. Our method of sampling is described in \Cref{lemma:cmi_ub:simulate_large}.

Next, we perform equivalence testing on the individual categories.
In total, we have three different types of categories. We will bound the sample complexity of testing these three categories individually. Then, we will combine the results for each category to obtain our final result.
For heavy categories where $\|\subD{Q}{\pL{i}{j}{k}}{ABC}\|_2$ is very small, equivalence testing could become very costly, as will become evident from the proof of \Cref{cmi:large:heavy}. We get around this by noting that if $\|\subD{P}{\pL{i}{j}{k}}{ABC}\|_2$ and $\|\subD{Q}{\pL{i}{j}{k}}{ABC}\|_2$ are small enough, then certainly $\|\subD{P}{\pL{i}{j}{k}}{ABC}-\subD{Q}{\pL{i}{j}{k}}{ABC}\|_2<\eps_L/k_{ABC}$, without explicitly testing for equivalence. The formal argument is presented in \Cref{lemma:cmi_ub_case_dist_small_l2}, and explains the case distinction in \Cref{alg:cmitestlarge}.

\begin{enumerate}
    \item \textbf{Heavy categories:} This includes all categories $\pL{i}{j}{k}$ where $i<k_A$ and $j<k_C$, such that for all $a,b,c$, both $p_{ab}$ and $p_{bc}$ are bounded away from zero. Note that some categories will be excluded from the testing, as mentioned above. This also provides us with a lower bound on $\|\subD{Q}{\pL{i}{j}{k}}{ABC}\|_2$, $\|\subD{Q}{\pL{i}{j}{k}}{ABC}\|_2\geq \eps_L(i,j,k)/(10e^{18})$.
    \begin{restatable}{claim}{cmilargeheavy}\label{cmi:large:heavy}

     For any heavy category $\pL{i}{j}{k}$ as defined in \Cref{fig:cmi_partition} with $\|\subD{Q}{T}{ABC}\|_2\geq \Omega(\eps_L(i,j,k))$, distinguishing whether $\subD{P}{\pL{i}{j}{k}}{ABC}=\subD{Q}{\pL{i}{j}{k}}{ABC}$ or $D_H^2(\subD{P}{\pL{i}{j}{k}}{ABC},$ $\subD{Q}{\pL{i}{j}{k}}{ABC}) \geq \eps_L/k_{ABC}$ can be done using \begin{equation}\label{eq:cmi_ub_heavy_sc}
        N_{L,\textnormal{heavy}}:= 10^{7} \ceq k_{ABC}   \frac{\sqrt{d_Ad_Bd_C}}{\eps_L},
    \end{equation}
    samples from $P_{ABC}$ and $Q_{ABC}$ each with probability of success at least $99/100$.    
    \end{restatable}

    \item \textbf{Mixed categories:} This includes all categories $\pL{i}{j}{k}$ where either $i=k_A$ or $j=k_C$. In the discussion, without loss of generality, we will assume that $a$ is coming from the last bin.

\begin{restatable}{claim}{cmilargemixed}\label{cmi:large:mixed}
For any mixed category $\pL{i}{j}{k}$ as defined in \Cref{fig:cmi_partition}, distinguishing whether $\subD{P}{\pL{i}{j}{k}}{ABC}=\subD{Q}{\pL{i}{j}{k}}{ABC}$ or $D_H^2(\subD{P}{\pL{i}{j}{k}}{ABC},$ $\subD{Q}{\pL{i}{j}{k}}{ABC}) \geq \eps_L/k_{ABC}$ can be done using
    \begin{equation}
        N_{L,\textnormal{mixed}}:= 10^3 \ceq k_{ABC}^2 \min\left\{\frac{d_A^{3/4}d_B^{3/4}d_C^{1/4}}{\eps_L},\frac{d_A^{2/3}d_B^{2/3}d_C^{1/3}}{\eps_L^{4/3}}\right\}
    \end{equation}
samples from $P_{ABC}$ and $Q_{ABC}$ each (where $d_A\geq d_C$), with probability of success at least $99/100$.    
\end{restatable}

    \item \textbf{Light categories:} This includes all categories $\pL{i}{j}{k}$ where both $i=k_A$ and $j=k_C$.

    \begin{restatable}{claim}{cmilargelight}\label{cmi:large:light}
    For any light category $\pL{i}{j}{k}$ as defined in \Cref{fig:cmi_partition}, distinguishing whether $\subD{P}{\pL{i}{j}{k}}{ABC}=\subD{Q}{\pL{i}{j}{k}}{ABC}$ or $D_H^2(\subD{P}{\pL{i}{j}{k}}{ABC},$ $\subD{Q}{\pL{i}{j}{k}}{ABC}) \geq \eps_L/k_{ABC}$ can be done using
    
    \begin{equation}
    \label{eq:cmi_ub_light_sc}
        N_{L,\textnormal{light}}:= 10\max\left\{\ceq k_{ABC}\frac{d_A^{1/2}d_B^{3/4}d_C^{1/2}}{\eps_L},\frac{\log(10^3k_{ABC})}{\nu}\right\}
    \end{equation} samples from $P_{ABC}$ and $Q_{ABC}$ with probability of success at least $99/100$.   
    \end{restatable}
\end{enumerate}

We will prove the above three claims in the following. Together, they directly imply \Cref{lem:cmi:large}.
\begin{lemma}
\label{lemma:algo_cmi_correctness}
    With probability at least $2/3$, \Cref{alg:cmitestlarge} correctly performs conditional independence testing in $D_H^2$ distance in the large regime with the sample complexity specified in \Cref{lem:cmi:large}.
\end{lemma}
\begin{proof}
    Analogous to the proof of \Cref{lemma:indep_hellinger_correctness}, we first argue that our partitioning satisfies certain properties. With probability at least $9/10$, it holds that
        \begin{align}
            \forall \pL{i}{j}{k},\forall (a,b,c)\in \pL{i}{j}{k}:
            \begin{cases}
                e^{-(i+2)} \mathbbm{1}[i<k_A]\leq q_{ab}\leq e^{-i+1}
                \\
                e^{-(j+2)}\mathbbm{1}[j<k_C]\leq q_{bc}\leq e^{-j+1}
                \\
                e^{-(k+2)}\leq q_{b}\leq e^{-k+1}.
            \end{cases}
        \end{align}
    The guarantee follows directly from \Cref{lemma:learn_approx} (note that $8\log(4d_Ad_C/\zeta)/(8Mk_{ABC})\geq 1/(2M)$, so any element whose probability mass is less than $1/(2M)$ will not have empirical frequency above $1/M$) and a union bound, analogously for each of $q_{ab}$, $q_{bc}$ and $q_b$. 
    In particular, it holds with high probability that
    \begin{equation}
    \label{eq:cmi_bucketing_guarantee}
        \forall \pL{i}{j}{k},\forall (a,b,c),(a',b',c')\in \pL{i}{j}{k}:
        \begin{cases}
            \frac{p_{ab}}{p_{a'b'}}\leq e^3 \text{ if }i<k_A, & \text{and } p_{ab}\leq e^{-k_A+1} \text{ otherwise,}
            \\
            \frac{p_{bc}}{p_{b'c'}}\leq e^3 \text{ if }j<k_C, &\text{and } p_{bc}\leq e^{-k_C+1} \text{ otherwise,}
            \\
            \frac{p_{b}}{p_{b'}}\leq e^3.
        \end{cases}
    \end{equation}
    We will use these inequalities in the following subsections. 
    
    To summarize, this implies that no `small' $p_{ab},p_{bc}\leq 1/(2M)$ is placed in one of the first $k_A$ buckets, while `large' $p_{ab}$ and $p_{bc}$ are estimated up to a factor of two, and hence are assigned to either the correct or a neighboring bin.
     
    Within a category $\pL{i}{j}{k}$ for which $i<k_A$ and $j<k_C$, we then have  
    \begin{equation}
        \forall (a,b,c), (a',b',c')\in \pL{i}{j}{k}: \frac{p_{ab}p_{bc}/p_b}{p_{a'b'}p_{b'c'}/p_{b'}}\leq e^9.
    \end{equation} 

    For the following, let 
    \begin{align}
        s_{ij}^k:=\sum_{(a,b,c)\in \pL{i}{j}{k}}\left(\sqrt{p_{abc}}-\sqrt{\frac{p_{ab}p_{bc}}{p_b}}\right)^2, \quad\text{such that}\quad D_H^2(P_{ABC},Q_{ABC})= \sum_{k,i,j}s_{ij}^k.
        \label{cmi:heavy:hellinger}
    \end{align}
        
    \noindent This implies that if $D_H^2(P_{ABC},Q_{ABC})\geq \eps_L$, then by the pigeonhole principle, at least one $s_{ij}^k$ will satisfy
    $s_{ij}^k\geq \eps_L/k_{ABC}$, which we will test this individually for each $\pL{i}{j}{k}$.

    The subroutines testing the individual categories will each succeed with probability at least $99/100$. To ensure an \emph{overall} high probability of success, we repeat each test $6\log(10^3k_{ABC})$ times and take a majority vote.
    The correctness of majority voting follows directly from the Chernoff bound, by choosing $\delta$ such that $(1-\delta)p=1/2$, we find the probability that less than $50\%$ of the outcomes indicate the right result is bounded by $\exp(-Np(1-1/2p)^2/2)$. We have $k_{ABC}$ instances in total, so performing each one (assuming $p\geq 99/100$)
    \begin{equation}
        \frac{\log(10^3k_{ABC})}{p(1-p/2)^2/2}<6\log(10^3k_{ABC})    
    \end{equation}
    times ensure an overall probability of success of at least $99/100$. We showed that the approximate learning is precise enough to make the distinction into bins correctly with probability at least $99/100$, sampling from $\tilde Q_{ABC}$ succeeds with probability at least  $99/100$ as well, and testing on the individual categories is successfull with probability $99/100$ too. This implies a probability of success of at least $2/3$. The contributions to the sample complexity are

    \begin{enumerate}[label=(\roman*)]
        \item The sample complexity of Step 1 in \Cref{alg:cmitestlarge} can be bounded by 
        \begin{equation}
            N_{1} \geq 24\log(c_1d_Ad_Bd_C)\max\left \{N_{L,\text{heavy}},N_{L,\text{mixed}},N_{L,\text{light}}\right\}
        \end{equation}
        \item The sample complexity of Step 2 in \Cref{alg:cmitestlarge} can be upper bounded by 
        \begin{equation}
            N_{2}\geq 5\cdot 6\log(c_1k_{ABC})k_{ABC}\max\left \{\frac{4}{\eps_L(i,j,k)},N_{L,\text{heavy}},N_{L,\text{mixed}},N_{L,\text{light}}\right\}.
        \end{equation}
        The additional factor in the sample complexity is due to simulating samples, as mentioned in \Cref{lemma:cmi_ub:simulate_large}. 
    \end{enumerate}
    Inserting the sample complexities, we can bound
    \begin{align}
        &N_{1}+N_{2}
        \\
        &\geq \ceq 10^5k_{ABC}^3\log(c_1k_{ABC})\max\left\{\min\left\{\frac{d_A^{3/4}d_B^{3/4}d_C^{1/4}}{\eps},\frac{d_A^{2/3}d_B^{2/3}d_C^{1/3}}{\eps^{4/3}}\right\},\frac{d_A^{1/2}d_B^{3/4}d_C^{1/2}}{\eps}\right\}.
    \end{align}
    To bound the probability of success, we apply a union bound, over all tasks we perform: we learn $P_{AB}$, $P_{BC}$ and $P_B$, create samples and test the different regimes. The testing is already boosted using majority voting, such that each of the five tasks individually succeeds with probability at least $99/100$, giving an overall probability of success at least $9/10$.
\end{proof}

\subsubsection{Sample Complexity of Heavy Categories}
\label{sec:cmi_ub:sc_heavy}
\cmilargeheavy*

\begin{proof}
Recall that in this case, for all $a,b,c$, both $p_{ab}$ and $p_{bc}$ are bounded away from zero, and there are $k_Ak_Bk_C<k_{ABC}$ many such $s_{i,j}^k$ for which this is the case. For ease of presentation, we will denote $\pL{i}{j}{k}$ as $T$ throughout this proof.

We can further assume that all $p_{ab}$ differ pairwise by at most a factor $e^3$, and by the case distinction, that $\|\subD{Q}{T}{ABC}\|_2\geq \eps_L/(8k_{ABC})$.  
Consider an arbitrary $s_{i,j}^k$. From \eqref{cmi:heavy:hellinger}, we have the following:
\begin{align}
            s_{i,j}^k&
            =\sum_{(a,b,c)\in T}\left(\sqrt{p_{abc}}-\sqrt{\frac{p_{ab}p_{bc}}{p_b}}\right)^2
            \\
            &=\sum_{(a,b,c)\in T}\frac{\left(p_{abc}-\frac{p_{ab}p_{bc}}{p_b}\right)^2}{\left(\sqrt{p_{abc}}+\sqrt{\frac{p_{ab}p_{bc}}{p_b}}\right)^2}
            \\
            &\leq \frac{e^9}{\max\limits_{\substack{b\in B_k \\(a,c)\in S_{ij}^{b}}}\left\{\frac{p_{ab}p_{bc}}{p_b}\right\}}\|\subD{P}{T}{ABC}-\subD{Q}{T}{ABC}\|_2^2\label{eq:dh_indep:heavy_bound}.
\end{align}
Thus, following \eqref{eq:dh_indep:heavy_bound}, distinguishing whether $s_{i,j}^k=0$ or $s_{i,j}^k\geq \eps_L/k_{ABC}$ is reduced to testing whether:
\begin{equation}\label{eqn:cmi:large:heavy:decision}
\subD{P}{T}{ABC}=\subD{Q}{T}{ABC}\quad\text{or}\quad\|\subD{P}{T}{ABC}-\subD{Q}{T}{ABC}\|_2\geq \eta,
\end{equation}
where 
\begin{equation}
    \eta:= \sqrt{\frac{\eps_L}{e^9k_{ABC} }e^{k-(i+j+3)}}\leq \sqrt{\frac{\eps_L}{e^9k_{ABC} }\max\limits_{(a,b,c)\in T}\left\{\frac{p_{ab}p_{bc}}{p_b}\right\}}.
\end{equation}

From \Cref{lemma:equivalence_l2}, we know that equivalence testing with respect to $\ell_2$ distance with proximity parameter $\eta$ on a subset $T$ requires $O(\|\subD{Q}{T}{ABC}\|_2/\eta^2+\sqrt{|T|}+1/\eta)$ samples. We now analyze the first term, $\|\subD{Q}{T}{ABC}\|_2/\eta^2$. Note that 
        \begin{equation}
        \label{eq:dh_cmi:l_2_norm}
            \|\subD{Q}{T}{ABC}\|_2= \sqrt{\sum_{(a,b,c)\in T}\left(\frac{p_{ab}p_{bc}}{p_b}\right)^2}\leq \sqrt{|T|}\max_{(a,b,c)\in T}\left\{\frac{p_{ab}p_{bc}}{p_b}\right\}.
        \end{equation}
        Hence, we have that 
        \begin{equation}
        \label{eq:cmi_ub:heavy_sc_raw}
            \frac{\|\subD{Q}{T}{ABC}\|_2}{\eta^2}\leq \frac{e^9\sqrt{|T|}k_{ABC}}{\eps_L}.
        \end{equation}
        This will clearly dominate the term $\sqrt{|T|}$, and it remains to consider the term $1/\eta$, which is dominant whenever 
        \begin{equation}
            \frac{\|\subD{Q}{T}{ABC}\|_2}{\eta^2}\leq \frac{1}{\eta}\quad\Leftrightarrow\quad  \|\subD{Q}{T}{ABC}\|_2\leq \eta.
        \end{equation}
        However, since we have a lower bound on $\|\subD{Q}{T}{ABC}\|_2$, $\|\subD{Q}{T}{ABC}\|_2\geq \eps_L(i,j,k)/(10e^{18})$, we can easily upper bound $1/\eta$ by $1/\eta\leq 10e^{18}/\eps_L(i,j,k)$, which is less than \eqref{eq:cmi_ub:heavy_sc_raw}. In order to use \Cref{lemma:equivalence_l2}, we need individual bounds on $\|\subD{Q}{T}{ABC}\|_2$ and $\gamma(i,j,k)$, which we can bound as follows using the definitions of $\pL{i}{j}{k}$ in \Cref{fig:cmi_partition} and \eqref{eq:cmi_bucketing_guarantee}
        \begin{equation}
            b(i,j,k):= e^9\min\left\{\sqrt{|T|}e^{k+4-(i+j)},\sqrt{e^{k+4-(i+j)}}\right\}\geq \|\subD{Q}{T}{ABC}\|_2,\quad\gamma(i,j,k):=\eta.       
        \end{equation}
        Because of these bounds, we get an additional factor of up to $e^6$ compared to \eqref{eq:cmi_ub:heavy_sc_raw}. 
        
        Due to the additional logarithmic factors, \eqref{eq:cmi_ub:heavy_sc_raw} dominates  $1/\eta$, and the large regime is thus bounded by 
        \begin{equation}
            \ceq \frac{e^{15}\sqrt{|T|}k_{ABC}}{\eps_L}\leq \ceq e^{15}\frac{\sqrt{d_Ad_Bd_C}k_{ABC}}{\eps_L}.
        \end{equation}
\end{proof}

\subsubsection{Sample Complexity of Mixed Categories}
\label{sec:cmi_ub:sc_mixed}
\cmilargemixed*

\begin{proof}
For brevity, we will again write $T$ instead of $\pL{i}{j}{k}$ in the following proof. We also define $|T_C|$ to be the maximal number of different $c$'s appearing in $T$, $T_C:=\{c|\exists a,b: (a,b,c)\in T\}$. We define $T_A$ and $T_B$ analogously. Our proof follows in a similar line as \Cref{cl:mi:ub:mixed}. 
We would again like to upper bound $D_H^2(\subD{P}{T}{ABC},\subD{Q}{T}{ABC})$ by $\|\subD{P}{T}{ABC}-\subD{Q}{T}{ABC}\|_2$. However, we can no longer use a similar approach as in \eqref{eq:dh_indep:heavy_bound} as there is no lower bound on $p_{ab}$. Let us assume without loss of generality that the unbounded coordinate comes from $A$, i.e.\ $p_{ab}\leq e/M$. We denote by $T_{-}$ the set of indices in $T$ for which $(\sqrt{p_{abc}}-\sqrt{p_{ab}p_{bc}/p_b})^2\leq x$, for some $x\in (0,1)$ to be determined, and $T_+:=T\setminus T_-$. Then 
\begin{align}
            &D_H^2(\subD{P}{T}{ABC},\subD{Q}{T}{ABC})
            \\
            &=\frac{1}{2}  \left(\sum_{s \in T_{+}} \left(\sqrt{\subD{P}{T}{ABC}}- \sqrt{\subD{Q}{T}{ABC}}\right)^2 + \sum_{s \in T_{-}} \left(\sqrt{\subD{P}{T}{ABC}}- \sqrt{\subD{Q}{T}{ABC}}\right)^2 \right) \\
            &= \sum_{(a,b,c)\in T_+}\left(\sqrt{p_{abc}}-\sqrt{\frac{p_{ab}p_{bc}}{p_b}}\right)^2+
            \sum_{(a,b,c)\in T_-}\left(\sqrt{p_{abc}}-\sqrt{\frac{p_{ab}p_{bc}}{p_b}}\right)^2
            \\
            &\leq \sum_{(a,b,c)\in T_+}\frac{1}{x}\left(p_{abc}-\frac{p_{ab}p_{bc}}{p_b}\right)^2+|(T)_-|x\\
            &\leq \frac{1}{x}\|\subD{P}{T}{ABC}-\subD{Q}{T}{ABC}\|_2^2+|T|x.
            \label{cmi:heavy:mixedbound}
        \end{align}

Since we want to test if $D_H^2(\subD{P}{T}{ABC},\subD{Q}{T}{ABC}) \geq \eps_L/(2k_{ABC})$,
from \eqref{cmi:heavy:mixedbound}, we have the following:
\begin{align}
    & D_H^2(\subD{P}{T}{ABC},\subD{Q}{T}{ABC}) \geq \frac{\eps_L}{2k_{ABC}} \\
    & \implies \|\subD{P}{T}{ABC}-\subD{Q}{T}{ABC}\|_2 \geq \sqrt{\frac{x\eps_L}{2k_{ABC}}-|T|x^2}\label{cmi:heavy:mixed2}
\end{align}
In order for the first term in the square root to dominate the second, we choose $x= {\eps_L}/{(4k_{ABC}|T|]})$, and arrive at the following testing problem:
\begin{equation}
\|\subD{P}{T}{ABC}-\subD{Q}{T}{ABC}\|_2\geq \tilde\eps_L:= \frac{\eps_L}{4k_{ABC}|T|^{1/2}} \quad  \text{or}\quad \subD{P}{T}{ABC}=\subD{Q}{T}{ABC}.
\label{eq:cmi_upper:eps}
\end{equation}

Similar to the proof of the previous claim, we now need to bound $\|\subD{Q}{T}{ABC}\|_2$. Here we will use \Cref{fact:prelim:l2}, which allows us to bound $\|\subD{Q}{T}{ABC}\|_2$ in the following two ways:
\begin{align}
\|\subD{Q}{T}{ABC}\|_2=\sqrt{\sum_{(a,b,c)\in T}\left(\frac{p_{ab}p_{bc}}{p_b}\right)^2}\leq 
\begin{cases} 
\sqrt{|T|}\max\limits_{(a,b,c)\in T}\left\{\frac{p_{ab}p_{bc}}{p_b}\right\} &\text{(Case 1)}\\
\sqrt{\max\limits_{(a,b,c)\in T}\left\{\frac{p_{ab}p_{bc}}{p_b}\right\}} &\text{(Case 2)}
\end{cases}
\label{eq:cmi_upper:l2}
\end{align}

From the different terms present in the sample complexity of equivalence testing (see \Cref{lemma:equivalence_l2}), we only consider $\|\subD{Q}{T}{ABC}\|_2/\tilde\eps_L^2$ (we will find that this term dominates over $1/\tilde\eps_L$ and $\sqrt{|T|}$).
Depending on which case is more favorable, we obtain:

\begin{enumerate}
    \item \textbf{Case $1$}:  From \eqref{eq:cmi_upper:eps} and \eqref{eq:cmi_upper:l2}, we have:
\begin{align}
    \frac{\|\subD{Q}{T}{ABC}\|_2}{\tilde \eps_L^2} &\leq  \frac{\sqrt{|T|}\max\limits_{(a,b,c)\in T}\left\{\frac{p_{ab}p_{bc}}{p_b}\right\}}{\tilde \eps_L^2} 
    \leq 16k_{ABC}^2\frac{|T|^{3/2}}{\eps_L^2}\max\limits_{(a,b,c)\in T}\left\{\frac{p_{ab}p_{bc}}{p_b}\right\}.
\end{align}
Note that $\forall (a,b,c)\in T: p_{bc}/p_b\leq e^3/|T_C|$ (since all $p_{bc}$ for a fixed $b$ from $(a,b,c)\in T$ differ by no more than a factor $e^3$, and their sum is bounded by $p_b$). Also, by assumption on $A$, $p_{ab}\leq e^{-(k_A+1)}=e/M$, and $|T|\leq |T_A||T_B||T_C|$, so 

\begin{align}
    &\frac{\|\subD{Q}{T}{ABC}\|_2}{\tilde \eps_L^2} \leq 4^2k_{ABC}^2\frac{e^6(|T_A||T_B|)^{3/2}|T_C|^{1/2}}{\eps^2 M}
\end{align}
The number of samples, $N_{L,\text{mixed}}$, we need to upper bound this. For this, assume that we are in a regime where $N_{L,\text{mixed}}\geq \max\{N_{L,\text{heavy}},N_{L,\text{light}}\}$, such that $M=N_{L,\text{mixed}}$. Then 
\begin{align}
    &4^2k_{ABC}^2\frac{e^6(|T_A||T_B|)^{3/2}|T_C|^{1/2}}{\eps_L^2 N_{L,\text{mixed}}}\leq N_{L,\text{mixed}}
    \implies N_{L,\text{mixed}}\geq \frac{(d_Ad_B)^{3/4}d_C^{1/4}}{\eps_L}\cdot 4e^3k_{ABC}^{}
\end{align}
samples are sufficient.

\item \textbf{Case $2$:}
Similarly, using  \eqref{eq:cmi_upper:eps} and \eqref{eq:cmi_upper:l2}, we can say the following:
\begin{align}
    \frac{\|\subD{Q}{T}{ABC}\|_2}{\tilde \eps_L^2} &\leq \frac{\sqrt{\max\limits_{(a,b,c)\in T}\left\{\frac{p_{ab}p_{bc}}{p_b}\right\}}}{\tilde \eps_L^2} \leq \sqrt{\frac{e^6}{N_{L,\text{mixed}}|T_C|}}\frac{16k_{ABC}^2|T|}{\eps_L^2}
    \overset{!}\leq N_{L,\text{mixed}},
\end{align}
such that we find that
\begin{equation}
    N_{L,\text{mixed}}\geq \frac{(d_Ad_B)^{2/3}d_C^{1/3}}{\eps_L^{4/3}}\cdot 4^2e^2k_{ABC}^2
\end{equation}
samples are sufficient as well.
\end{enumerate}

Note that an analogous argument holds for the case where the roles of $A$ and $C$ are switched. We explicitly know $\tilde\eps_L$, and bounded $\|\subD{Q}{T}{ABC}\|_2\leq \min\{\sqrt{|T|}e^{2-k_A},e^{1-k_A/2}\}$. 
Since we can always choose which of the two cases we want to use, the proof of the claim is complete. Also note that the contribution of always dominates the contribution by $1/\tilde\eps_L$ and $\sqrt{|T|}$.    
\end{proof}

\subsubsection{Sample Complexity of Light Categories}
\label{sec:cmi_ub:sc_light}
In the following we will bound the sample complexity for testing a light category $T$. Instead of using \cref{lemma:equivalence_l2} for equivalence testing between $\subD{P}{T}{ABC}$ and $\subD{Q}{T}{ABC}$ directly, we take an intermediate step. Let $R(T):=\{(a,b,c)|b\in T_B\}$ in the following. First, we take samples from $P_{ABC}$ and $Q_{ABC}$ and immediately reject those with $b\notin T_B$. This allows us to condition on $b$ coming from $T_B$, and the remaining samples can be seen as coming from the (normalized) distributions $\subD{P}{R(T)}{ABC}/\ell$ and $\subD{Q}{R(T)}{ABC}/\ell$, for $\ell:=\sum_{i\in R(T)}p_i\leq 1$. We can then apply \Cref{lemma:equivalence_l2} to test for equivalence between $\subD{P}{T}{ABC}/\ell$ and $\subD{Q}{T}{ABC}/\ell$, which we see as `subdistributions' of $\subD{P}{R(T)}{ABC}/\ell$ and $\subD{Q}{R(T)}{ABC}/\ell$.

In our calculation, this results in a tighter bound than using \Cref{lemma:equivalence_l2} directly to perform the testing in a black box manner. The improvement comes into play in \eqref{eq:bound_nl_ell}, where we use the normalization factor $\ell$ to bound a term of the form $\ell/p_b$. Using \Cref{lemma:equivalence_l2} without the conditioning would result in  a worse factor of $1/p_b$ instead. Note that $1/p_b$ does not need to be bounded in the same way in the mixed and heavy categories, such that the approach we use here for the light regime would not provide any advantage in the other regimes.
\cmilargelight*

\begin{proof}
In this proof, our approach follows a similar line as that of the analysis of mixed categories. First, as mentioned previously, we define $R(T):=\{(a,b,c)|b\in T_B\}$ and $\ell:=\sum_{i\in R(T)}p_i$. Since both $a$ and $c$ are light here, we can no longer use the fact that $p_{bc}/p_b\leq O(1/|T_C|)$. Instead, we use the fact that $p_{ab}, p_{bc} \leq 1/N_L$.  Note that we define $T_+$ and $T_-$ slightly differently in this case. We denote by $T_{-}$ the set of indices in $T$ for which $(\sqrt{p_{abc}}-\sqrt{p_{ab}p_{bc}/p_b})^2/\ell \leq y$ holds, for some $y\in (0,1)$ to be determined. Further, let $T_+:=T\setminus T_-$. Note that
\begin{equation}
    \left\|\frac{\subD{P}{R(T)}{ABC}}{\ell}\right\|_1=\sum_{(a,b,c)\in R(T)}\frac{p_{abc}}{\ell}=1,\qquad \left\|\frac{\subD{Q}{R(T)}{ABC}}{\ell}\right\|_1=\sum_{(a,b,c)\in R(T)}\frac{p_{ab}p_{bc}}{p_b \ell}=1.
\end{equation}

We can now write
\begin{align}
    &D_H^2(\subD{P}{T}{ABC},\subD{Q}{T}{ABC})
    \\
    &=\frac{1}{2}  \left(\sum_{t \in T_{+}} (\sqrt{P_{ABC}[t]}- \sqrt{Q_{ABC}[t]})^2 + \sum_{t \in T_{-}} (\sqrt{P_{ABC}[t]}- \sqrt{Q_{ABC}[t]})^2 \right) 
    \\
    &= \sum_{(a,b,c)\in T_+}\left(\sqrt{p_{abc}}-\sqrt{\frac{p_{ab}p_{bc}}{p_b}}\right)^2+\sum_{(a,b,c)\in T_-}\left(\sqrt{p_{abc}}-\sqrt{\frac{p_{ab}p_{bc}}{p_b}}\right)^2
    \\
    &\leq \ell \left[\sum_{(a,b,c)\in T_+}\left(\sqrt{\frac{p_{abc}}{\ell}}-\sqrt{\frac{p_{ab}p_{bc}}{p_b \ell}}\right)^2+\sum_{(a,b,c)\in T_-}\left(\sqrt{\frac{p_{abc}}{\ell}}-\sqrt{\frac{p_{ab}p_{bc}}{p_b \ell}}\right)^2\right]
    \\
    &\leq \ell \left[\sum_{(a,b,c)\in T_+}\frac{1}{y}\left(\frac{p_{abc}}{\ell}-\frac{p_{ab}p_{bc}}{p_b \ell}\right)^2+|T_-|y\right]
    \\
    &\leq \ell \left[\frac{1}{y}\Big\|\frac{\subD{P}{T}{ABC}}{\ell}- \frac{\subD{Q}{T}{ABC}}{\ell}\Big\|_2^2+|T|y\right].
\end{align}
If we take $N_{L,\text{light}}$ samples from both $\subD{P}{R(T)}{ABC}$, in expectation we receive $\ell N_{L,\text{light}}$ samples from $\subD{P}{R(T)}{ABC}$, and analogously for $\subD{Q}{R(T)}{ABC}$. To ensure that we receive no less than $\ell N_{L,\text{light}}/2$ samples with high probability, we impose the requirement that $\ell N_{L,\text{light}}\geq 8\log(10^3k_{ABC})>8\log(10^3k_B)$. Then, \Cref{lemma:learn_approx} ensures that with probability no less than $1-2/(10^3k_{ABC})$, we receive at least $\ell N_{L,\text{light}}/2$ samples from both $\subD{P}{R(T)}{ABC}$ and $\subD{Q}{R(T)}{ABC}$. 

As mentioned previously, the remaining samples can be seen as coming directly from $R(T)$, and we can apply \Cref{lemma:equivalence_l2} to test the equivalence between $\subD{P}{R(T)}{ABC}/\ell$ and $\subD{Q}{R(T)}{ABC}/\ell$ on the subset $T$. We now determine the precision to which we perform the equivalence testing,
\begin{align}
    &\frac{\eps_L}{k_{ABC}}\leq D_H^2(\subD{P}{T}{ABC},\subD{Q}{T}{ABC})\leq \ell \left[\frac{1}{y}\Big\|\frac{\subD{P}{T}{ABC}}{\ell}-\frac{\subD{Q}{T}{ABC}}{\ell}\Big\|_2^2+|T|y\right]
    \\
    &\quad \implies \Big\|\frac{\subD{P}{T}{ABC}}{\ell}-\frac{\subD{Q}{T}{ABC}}{\ell}\Big\|_2\geq \sqrt{\frac{\eps_L y}{k_{ABC}\ell}-|T|y^2}=:\tilde \eps_L.
\end{align}
We choose $ y=\eps/(2|S|\ell k_{ABC})$, resulting in 
\begin{equation}
    \tilde\eps_L=\frac{\eps_L}{2\ell k_{ABC}\sqrt{|T|}}\geq \frac{e\eps_L}{2k_{ABC}\max_{b\in T_k}\{p_b\}|B_k|\sqrt{|T|}}.
\end{equation} 

In the following, note that $p_{ab},p_{bc}\leq e/M$ each. We now use Case 2 of \Cref{fact:prelim:l2},
    \begin{equation}
    \label{eq:bound_q_l2_large_both_last}
        \|\subD{Q}{T}{ABC}/\ell\|_2=\sqrt{\sum_{a,c\in T_{AC}}\left(\frac{p_{ab}p_{bc}}{p_b \ell}\right)^2}\leq \sqrt{\max\left\{\frac{p_{ab}p_{bc}}{p_b \ell}\right\}}\leq \sqrt{\max\left\{\frac{1}{p_b \ell}\right\}}\frac{e}{M}.
    \end{equation}
    The sample complexity for equivalence testing using samples from $\subD{Q}{R(T)}{ABC}$ to precision $\tilde\eps_L$ in $\ell_2$ is $O(\max\{\|\subD{Q}{T}{ABC}\|_2/\tilde\eps_L^2,1/\tilde\eps_L,\sqrt{d_A|T_B|d_C}\})$,  according to \Cref{lemma:equivalence_l2}. We will see that the first term dominates, which can be bounded by
    \begin{equation}
        \ceq\frac{\|\subD{Q}{T}{ABC}/\ell\|_2}{\tilde\eps_L^2}\leq \ceq\frac{|T|\ell^2k_{ABC}^2e}{\sqrt{\min_{b\in T_B}\{p_b\}\ell}\eps^2M}.
    \end{equation}
    This number of samples needs to be upper bounded by the number of samples we see from $\subD{Q}{R(T)}{ABC}$, which we argued above to be at least ${\ell N_{L,\text{light}}}/{2}.$ with high probability.
    
    Assuming again the worst case, where $M=N_{L,\text{light}}$, such that this case dominates the sample complexity, results in the requirement
    \begin{equation}
    \label{eq:bound_nl_ell}
        N_{L,\text{light}}\geq 2e\sqrt{\ceq}\frac{\ell^{1/4}|T|^{1/2}k_{ABC}}{\eps_L(\min_{b\in S_B}\{p_b\})^{1/4}},
    \end{equation}
    which we can ensure by choosing $N_{L,\text{light}}$ such that
    \begin{equation}
        N_{L,\text{light}}\geq\frac{|T_B|^{1/4}|T|^{1/2}}{\eps_L}10\sqrt{\ceq}k_{ABC}.
    \end{equation}
    We also required $\ell N_{L,\text{light}}\geq 8\log(10^3k_B)$. Using $\ell\geq \min_{b\in B_k}\{p_b\}|T_B|\geq \nu|T_B|$, such that we find the bound
    \begin{equation}
        N_{L,\text{light}}\geq 10\max\left\{\frac{|T_B|^{1/4}|T|^{1/2}}{\eps_L}\sqrt{\ceq}k_{ABC}, \frac{\log(c_1k_b)}{\nu |T_B|}\right\}.
    \end{equation} 
\end{proof}

\begin{remark}
    In the above calculations, there seem to be alternative options to bound certain terms. Here we will briefly argue why they do not lead to an improvement over our sample complexity reported in \Cref{lem:cmi:large}.
    First, in \eqref{eq:bound_q_l2_large_both_last}, we chose to use Case 2 of \Cref{fact:prelim:l2}. What happens if we use Case 1 instead? We then obtain
    \begin{align}
        \frac{\|\subD{Q}{T}{ABC}\|_2}{\tilde \eps^2} &\leq  \frac{\sqrt{|T|}\max\limits_{(a,b,c)\in T}\left\{\frac{p_{ab}p_{bc}}{p_b\ell}\right\}}{\tilde \eps^2}\leq \frac{\ell|T|^{3/2}}{ p_b\eps^2N_L^2}\overset{!}\leq \ell N_L\quad \implies N_L\geq \frac{|T|^{1/2}}{p_b^{1/3}\eps^{2/3}}.
    \end{align}
    We are working in the regime of large $p_b$, and we know from \Cref{sec:cmi_ub:small_regime} that $p_b\geq \max\{\eps/d_B^{7/8}, \eps^{8/7}/d_B^{6/7}\}$ (note that including the scaling in $d_A$ and $d_C$ would only decrease $p_b$). Thus we obtain $N_L \geq \max\{|T|^{1/2}d_B^{7/24}/\eps,|T|^{1/2}d_B^{6/21}/\eps^{22/21}\}$, which is larger than the sample complexity from Case $2$, $|T|^{1/2}d_B^{1/4}/\eps$, in all settings.

    It might also seem that in \eqref{eq:bound_nl_ell}, instead of using $\ell=\Theta(p_b|T_B|)$ to cancel $p_b$, we might try to bound $\ell$ by 1, and bound $p_b$ directly by the minimal value of $p_b$ in the large regime, as guaranteed by \Cref{sec:cmi_ub:small_regime}. 
    Instead of $d_B^{1/4}|T|^{1/2}/\eps$, we would get $x^{1/4}|T|^{1/2}/\eps$, where 
    \begin{equation}
        x:=\min\left\{d_B,d_B^{6/7}(d_Ad_C)^{2/7}/\eps^{8/7},d_B^{7/8}(d_Ad_C)^{1/4}/\eps\right\}.
    \end{equation}
    Suppose $x_{\min}:=d_B^{6/7}(d_Ad_C)^{2/7}/\eps^{8/7}$ is the minimum, implying $d_B\geq (d_Ad_C)^2/\eps^8$. Now let us check if this would change the sample complexity. This would only happen if the contribution of the third regime, which is also $x_{\min}$, was smaller than the original sample complexity: $(d_Ad_C)^{1/2}d_B^{3/4}/\eps>d_B^{6/7}(d_Ad_C)^{2/7}/\eps^{8/7}$ implies $(d_Ad_C)^2\eps^{4/3}\geq d_B$, a contradiction! In the parameter range where we could get an improvement, another regime is always dominant. Similarly, if $x_{\min}=d_B^{7/8}(d_Ad_C)^{1/4}/\eps$ is the minimum, again $d_B\geq (d_Ad_C)^2/\eps^8$ holds, and $x_{\min}$ is the contribution by the small regime. A change in sample complexity would occur if $(d_Ad_C)^{1/2}d_B^{3/4}/\eps>d_B^{7/8}(d_Ad_C)^{1/4}/\eps$, but this implies $(d_Ad_C)^2\geq d_B$, again a contradiction.

\end{remark}

\section{Lower Bounds for Conditional Independence Testing}\label{sec:cmi_lb} 
In the following, we derive lower bounds for the sample complexity of conditional mutual information testing. 

\begin{theorem}[Formalized \Cref{res:cmi-lower}]
For the sample complexity of conditional independence testing, \Cref{prob:CI_DH2}, where w.l.o.g. $d_A\geq d_C$, it holds that
    \begin{align}
        \textnormal{SC}_{\textnormal{CI},H}(\eps, d_A, d_B, d_C)
        &=\Omega\left(\max\{f_{\textnormal{sym}}(\eps, 1, d_B, 1), f_{\textnormal{asym}}(\eps, d_A, d_B, d_C)\}\right),
    \end{align}
    where, as defined in \eqref{eq_def_f_sym} and \eqref{eq_def_f_asym},
    \begin{align}
        f_{\textnormal{sym}}(\eps, 1, d_B, 1)&= \min\left\{\frac{d_B^{7/8}}{\eps},\frac{d_B^{6/7}}  {\eps^{8/7}}\right\},
        \\
        f_{\textnormal{asym}}(\eps, d_A, d_B, d_C)&=\min\left\{\frac{d_A^{3/4}d_B^{3/4}d_C^{1/4}}{\eps},\frac{d_A^{2/3}d_B^{2/3}d_C^{1/3}}{\eps^{4/3}}\right\}.
    \end{align}
\end{theorem}
The lower bounds for the first term,
\begin{equation}
    \textnormal{SC}_{\textnormal{CI},H}(\eps, d_A, d_B, d_C)= \widetilde{\Omega}\left(\min\left\{\frac{d_B^{7/8}}{\eps},\frac{d_B^{6/7}}  {\eps^{8/7}}\right\}\right),
\end{equation}
follows directly from \cite[Remark A.2]{canonne_testing_2018}. The lower bounds for the second regime are obtained by a reduction to independence testing.
Our result in this section is described below.

\begin{lemma}
Testing for conditional independence with respect to $D_H^2$, \Cref{prob:CI_DH2}, requires 
\begin{equation}
    \textnormal{SC}_{\textnormal{CI},H}(\eps, d_A, d_B, d_C)=
\Omega\left(\min\left\{\frac{(d_Ad_B)^{3/4}d_C^{1/4}}{\eps},\frac{(d_Ad_B)^{2/3}d_C^{1/3}}{\eps^{4/3}}\right\}\right)
\end{equation}
samples, assuming w.l.o.g.\ $d_A\geq d_C$.
\end{lemma}

To prove lower bounds for testing conditional independence, we reduce the problem to the lower bounds we derived for independence testing in \Cref{sec:ind_test_lb}. We will construct two sets of hard distributions, such that the distributions in the first set are conditionally independent, and the distributions in the second set are far from being conditionally independent. Based on the outcome of a fair coin toss $X$, we will choose one set and then pick a random distribution $P$ from that set. Finally, we will obtain samples from the distribution $P$. We then argue how a lower bound on the sample complexity of reconstructing $X$ reliably implies a lower bound on the sample complexity of testing for conditional independence. First let us define a set of distributions over $A\times B\times C$, with $|A|=d_A$, $|B|=d_B$, $|C|=d_C$, as follows (note that these are the same as in the proof of \Cref{thm:mi_lower}, where $A\rightarrow AB$). Similar to our approach in \Cref{sec:ind_test_lb}, we will take $\mathsf{Poi}(n)$ samples from the distribution.

\begin{proof}

We construct a distribution as follows. For all $(a,b)\neq (0,0)$:
\begin{itemize}
    \item For a given $(a,b)$, with probability $\alpha:=\min\{n/(d_Ad_B),1/2\}$, we set $p_{abc}=1/(2nd_C)$ for all $c$.
    
    \item Otherwise, we set for each $c$ individually
    \begin{align}
    \label{eq:cmi_lb_def_p_case2}
        p_{abc}=
        \begin{cases}
            \frac{\eps}{d_Ad_Bd_C} &\text{if }X=0,\\
            \text{uniformly random } \frac{\eps}{2d_Ad_Bd_C} \text{ or } \frac{3\eps}{2d_Ad_Bd_C} &\text{if }X=1. 
        \end{cases}
    \end{align}
\end{itemize}
$(0,0,c)$ is then uniform and carries the rest of the weight.

Where $n$ is chosen as determined in the lower bound with $d_A\rightarrow d_Ad_B$. Clearly, \Cref{theo:theo:indtesthellingerlb} about independence testing with respect to $D_H^2$ directly implies lower bounds (and a value of $n$) of 
\begin{equation}
n=\Omega\left(\min\left\{\frac{(d_Ad_B)^{3/4}d_C^{1/4}}{\eps},\frac{(d_Ad_B)^{2/3}d_C^{1/3}}{\eps^{4/3}}\right\}\right).
\end{equation}

However, we could also use a conditional independence tester as a decoder to decide from which set of distributions we received the samples. If $X=0$, that is, if $P_{ABC}=P_{AB}P_C$, then $P_{BC}=P_{B}P_{C}$, and hence $P_{AB}P_{BC}/P_B=P_{AB}P_C$. On the other hand, for $X=1$, we will argue in \Cref{lemma:cmi_lower_bounds_farness_regime_1} below that $D_H^2(P_{ABC},P_{AB}P_{BC}/P_B)\geq \Omega(\eps)$. This implies that a possible solution to the decision problem would be to use an algorithm taking samples from $P_{ABC}$ and deciding whether $D_H^2(P_{ABC},P_{AB}P_{BC}/P_B)=0$ or at least $\Omega(\eps)$. The lower bound from \Cref{thm:mi_lower} implies then a lower bound on how efficient any conditional independence tester can perform the task. This completes our reduction.
\end{proof}

We will now prove that a distribution constructed as above when $X=1$, is far from being conditionally independent.

\begin{lemma} \label{lemma:cmi_lower_bounds_farness_regime_1}
With high probability, the distributions constructed above with $X=1$ satisfy 
\begin{equation}
D_H^2(P_{ABC},P_{AB}P_{BC}/P_B)\geq \Omega(\eps).
\end{equation}
\end{lemma}

\begin{proof}
For simplicity, we consider only entries for which $P_{ABC}(a,b,c)=\eps/(2d_Ad_Bd_C)$. Let $S$ denote the set of all triplets $(a,b,c)$ for which $p_{abc}=\eps/(2d_Ad_Bd_C)$, excluding elements with $a=0$ or $b=0$. 
By construction of our distributions, we know that with high probability, $|S|$ is within a constant factor of $d_Ad_Bd_C/8$. Then, 
\begin{align}
   D_H^2(P_{ABC},P_{AB}P_{C|B})&\geq D_H^2(\subD{P}{S}{ABC},\subD{(P_{AB}P_{C|B})}{S}{})=\sum_{(a,b,c)\in S}\left(\sqrt{p_{abc}}-\sqrt{p_{ab}p_{bc}/p_b}\right)^2.
\end{align}
\begin{itemize}
    \item We have $p_{ab}=\sum_{c'}p_{abc'}$ and $p_b= \sum_{a',c'}p_{a'bc'}$. Recall that, for simplicity, we excluded the case where $a$ or $b$ are zero. Let $X$ indicate for how many $(a,b)$ we set $p_{abc}$ according to \eqref{eq:cmi_lb_def_p_case2}, and by $Y_{ic}$, $i\in \{0,..., X\}$ whether we select $\eps/(2d_Ad_Bd_C)$ or $3\eps/(2d_Ad_Bd_C)$ for the $i$th choice and coordinate $c$. Then
    \begin{align}
        p_b&=(d_A-X) d_C \frac{1}{2nd_C} +\sum_{i=1}^X\left(\frac{d_C}{2}+\sum_{c=1}^{d_C}Y_{ic}\right)\frac{\eps}{d_Ad_Bd_C},
        \\
        p_{bc}&=(d_A-X) \frac{1}{2nd_C} +\sum_{i=1}^X\left(\frac{1}{2}+Y_{ic}\right)\frac{\eps}{d_Ad_Bd_C}.
    \end{align}
    Since $1-\alpha\geq 1/2$, with high probability, $X=\Theta(d_A)$. Using a Chernoff bound, we can thus argue that with high probability, $\sum_{i=1}^XY_{ic}$ is close to its expectation value, $X/2$. Thus, with high probability,
    \begin{equation}
        0.9d_C\left[(d_A-X)\frac{1}{2nd_C} +X\frac{\eps}{d_Ad_Bd_C}\right]\leq p_{b}\leq 1.1d_C\left[(d_A-X)\frac{1}{2nd_C} +X\frac{\eps}{d_Ad_Bd_C}\right],
    \end{equation}
    and analogously for $p_{bc}$ (without the factor $d_C$)
    Using a Chernoff bound, we can argue that with high probability, $p_b$ is close to it's expectation value, up to a multiplicative factor close to one.  
    \item For $p_{ab}=\sum_{c'}p_{abc'}$, we can argue analogously with a Chernoff bound that with high probability $p_{ab}$ is close to $\mathbb{E}[p_{ab}]=\eps/(d_Ad_B)$. 
\end{itemize}
Combined, it holds with high probability that for a constant fraction of the elements in $S$,
\begin{equation}
    \frac{p_{ab}p_{bc}}{p_b}\geq
    \frac{3}{4}\frac{\eps}{d_Ad_Bd_C}.
\end{equation}
Using the facts that $p_{abc}= \eps/(2d_Ad_Bd_C)$ and $|S|=\Omega(d_Ad_Bd_C)$, we have
\begin{equation}
    D_H^2(P_{ABC},P_{AB}P_{C|B})|_S\geq \Omega\left(d_Ad_Bd_C\left(\sqrt{\frac{1}{2}}-\sqrt{\frac{3}{4}}\right)^2\frac{\eps}{d_Ad_Bd_C}\right)=\Omega(\eps).
\end{equation}
This completes the proof.
\end{proof}

\section{Equivalence Testing of Distributions}
\label{sec:equiv_testing_general}
In this section, we show that when applied to equivalence testing (\Cref{prob:equiv}), our estimator \eqref{eqn:estimator_intro} is able to recover the optimal sample complexity for testing in $\ell_2$ distance, see \cite[Thm.\ 2]{chan_optimal_2013} and \cite[Lemma 2.3]{diakonikolas_new_2016} under the guarantee that $\max\{\|P\|_2,\|Q\|_2\}\leq b$:
\begin{align}
    &\textnormal{SC}_{\textnormal{EQIV}}(\ell_1, \eps,b,n) = O\left(\frac{bn}{\eps^2}\right),\label{equiv:eq:sc_l1}\\
    & \textnormal{SC}_{\textnormal{EQIV}}(\ell_2, \eps,b,n) = O\left(\frac{b}{\eps^2}\right). \label{equiv:eq:sc_l2}
\end{align}

Our equivalence tester (\Cref{alg:eqvalence_test}) takes $N=200b/\eps^2$ samples of $P$ and $Q$ each to distinguish whether $P=Q$ or $\|P-Q\|_2\geq\eps$.  
To test with respect to $\ell_1$ distance instead of $\ell_2$ distance, we apply \Cref{alg:eqvalence_test} with precision $\eps/\sqrt{n}$ and use $\|P-Q\|_1\leq \sqrt{n}\|P-Q\|_2$ (see \Cref{fact:relations_distances}). 
\begin{theorem}\label{theo:eqn_testing_main}
Consider the problem of equivalence testing, \Cref{prob:equiv}. Then the estimator from \eqref{eqn:estimator_intro} can be used to solve \Cref{prob:equiv} using
\begin{equation}
    \textnormal{SC}_{\textnormal{EQIV}}(\ell_1, \eps,b,n) = O\left(\frac{bn}{\eps^2}\right), \qquad\textnormal{SC}_{\textnormal{EQIV}}(\ell_2, \eps,b,n) = O\left(\frac{b}{\eps^2}\right).
\end{equation}
samples. Replacing the condition $\max\{\|P\|_2,\|Q\|_2\}\leq b$ by $\|Q\|_2\leq b$ increases the sample complexity by an additive term in $O(\sqrt{n})$.
\end{theorem}

\begin{algorithm}[H]
\LinesNumbered
\DontPrintSemicolon
\setcounter{AlgoLine}{0}
\caption{Equivalence Testing $(P,Q)$}
\label{alg:eqvalence_test}
\KwIn{Parameter $b$, two multisets $\mathcal{S}_P$ and $\mathcal{S}_Q$ with $N$ elements from $[n]$ each, parameter $\eps \in (0,1)$.\Comment*[r]{$N=200b/\eps^2$}}
\KwOut{`{\bf Yes}' or `{\bf No}'}

$X,X'$ $\leftarrow$ Split $\mathcal{S}_P$ into two subsets of size at least $100b/\eps^2$ each \

$Y,Y'$ $\leftarrow$ split $\mathcal{S}_Q$ into two subsets of size at least $100b/\eps^2$ each \

$\forall i \in [n]$: $X_i \gets \#i \in X$, $X_i' \gets \#i \in X'$, $Y_i \gets \#i \in Y$, $Y'_i \gets \#i \in Y$

$Z\gets\sum_{i \in [n]}Z_{i},\quad Z_{i}:=X_{i}X'_{i}-2X_{i}Y_{i}+Y_{i}Y'_{i}$

\Return `{\bf Yes}' if $Z \leq \frac{\eps^2N^2}{2}$, otherwise \Return `{\bf No}'
                           
\end{algorithm}

To prove the correctness of \Cref{alg:eqvalence_test}, we will use our results from \Cref{sec:cmi_ub}. As shown in \cite[Lemma 3.2]{diakonikolas_new_2016}, the algorithm can easily be modified to work on the weaker assumption as well, increasing the sample complexity by $O(\sqrt{n})$.

\begin{proof}[Proof of \Cref{theo:eqn_testing_main}]
We first calculate the expectation value of $Z$. Here we denote $N=O(\frac{b}{\eps^2})$. Note that $X_i, X_i' \sim \mathsf{Poi}(Np_i)$. Similarly, $Y_i, Y_i' \sim \mathsf{Poi}(Nq_i)$. Thus
\begin{align}
\mathbb{E}[Z] &=\sum_{i \in [n]} \mathbb{E}[X_i]\mathbb{E}[X_i']-2\mathbb{E}[X_i]\mathbb{E}[Y_i]+\mathbb{E}[Y_i]\mathbb{E}[Y_i']\\
&=\sum_{i \in [n]}\left(\mathbb{E}[X_i]-\mathbb{E}[Y_i]\right)^2 \\
&= \sum_{i \in [n]}(Np_i -Nq_i)^2\\
&=N^2 \|P-Q\|_2^2.
\label{eqn:eqivalence_expectation}    
\end{align}
To bound the variance of $Z$, we reuse \eqref{eq:cmi_upper_small_var_bound},
\begin{align}
\Var[Z_i] &\leq 2(\E[X_i^2]+\E[Y_i^2])^2 \leq 4 N^2(p_i^2 + q_i^2).
\end{align}
Because of the Poissonization process, the random variables $\{Z_1, \ldots, Z_n\}$ will be independent of each other, such that
\begin{align}
\Var[Z] &= \sum_{i \in [n]} \Var[Z_i] \leq 4 \sum_{i \in [n]} N^2(p_i^2 + q_i^2) =4N^2(\|P\|_2^2+\|Q\|_2^2).\label{eqn:eqivalence_variance}
\end{align} 
Let us first consider the case when $P=Q$. This implies that $p_i=q_i$ for every $i \in [n]$. Thus, following \eqref{eqn:eqivalence_expectation}, we know that $\E[Z]=0$ and, since $\max\{\|P\|_2,\|Q\|_2\} \leq b$, $\Var[Z] \leq 8 N^2 b^2$ from \eqref{eqn:eqivalence_variance}.

Following Chebyshev's inequality with $t=\eps^2N^2/2$ (our threshold), and using that $N= 100b/\eps^2$, we have the following:
\begin{align}
    \Pr\left[Z \geq \frac{N^2\eps^2}{2} \right] &\leq \frac{4\Var[Z]}{\eps^4N^4} \leq \frac{32 N^2 b^2}{10^4N^2b^2}  \leq \frac{1}{3}.  
\end{align}

Now let us consider the case when $\|P-Q\|_2 \geq \eps$. From \eqref{eqn:eqivalence_expectation}, we can say that $\E[Z] \geq \eps^2N^2 $. Using Chebyshev's inequality with again $t=\eps^2N^2/2$, we have:
\begin{align}
    \Pr\left[\frac{N^2\eps^2}{2}\geq Z\right]&\leq \Pr[\E[Z]/2\geq Z]
    \\&=\Pr[\E[Z]-Z\geq \E[Z]/2]
    \\
    &\leq\Pr\left[|Z-\E[Z]| \geq  \E[Z]/2\right]\\
    &\leq\Pr\left[|Z-\E[Z]| \geq  N^2\eps^2/4\right]\\
    &\leq \frac{16\Var[Z]}{\eps^4N^4}\\
    &\leq \frac{16\cdot 80N^2 b^2}{10^4 N^2b^2}\leq \frac{1}{3}.
\end{align}
If we want to loosen the constraint $\max\{\|P\|_2,\|Q\|_2\}\leq b$ to $\|Q\|_2\leq b$, we first learn $c_P$, an approximation of $\|P\|_2$ up to a constant multiplicative factor $2$ according to \Cref{lemma:get_weight_l2_basic}, which takes $O(\sqrt{d})$ samples. We reject if this test reveals that $\|P\|_2$ and $\|Q\|_2$ cannot be equal, $\|Q\|_2\leq b< c/2$. Otherwise we can proceed as before, since we then know that $\max\{\|P\|_2,\|Q\|_2\}\leq \Theta(b)$.
\end{proof}

\section*{Conclusion}
In this work, we introduce the following contributions:
\begin{enumerate}[label=(\roman*)]
    \item We design the first sample-optimal (up to polylogarithmic factors) mutual information tester and also prove its optimality.

    \item We introduce a novel conditional mutual information tester, and prove that it is optimal in certain regimes. Moreover, we conjecture that our upper bounds are tight.
\end{enumerate}
Along the way, we define a new estimator for equivalence testing of distributions, which works also for the correlated samples generated in our sampling approach. We believe this estimator will be of independent interest to the community.
Throughout, our proof techniques provide an intuitive explanation for the complicated sample complexities we encounter in our results.

\subsection*{Open Questions}
There remain several relevant open questions closely related to our work, which we describe below: 

\begin{enumerate}
    \item[(i)] For the problem of conditional mutual information testing, we proved lower bounds in the setting of $f_{\textnormal{sym}}(\eps, 1, d_B, 1), f_{\textnormal{asym}}(\eps, d_A, d_B, d_C)$ as mentioned in \Cref{res:cmi-lower}. For $f_{\textnormal{asym}}(\eps, d_A, d_B, d_C)$, this matches with the corresponding upper bounds in \Cref{res:cmi-upper}. Proving matching lower bounds for the other three terms in the sample complexity remains an interesting open problem. These regimes also have not yet been resolved for conditional independence testing in variation distance.

    \item[(ii)] As we mentioned in the introduction, the authors in \cite{canonne_testing_2018} followed an algebraic approach to study the conditional independence testing problem with respect to the $\ell_1$ distance. In this work, we used a combinatorial approach to the conditional mutual information testing problem. Although our approaches to these two different (but related) problems are different, the sample complexities of both problems are very similar
(see \Cref{res:cmi-upper} and \Cref{eq:canonne_cmi_trace_dist_upper_intro}). This motivates the question whether these two different approaches could be combined into a more general framework, or whether the approaches could be applied to the other respective problem.

\item[(iii)] In this work, we studied the problems of mutual and conditional mutual information testing in the non-tolerant setting, in the sense that we wanted to distinguish between the classes if the (conditional) mutual information is zero or at least $\eps$, for some threshold parameter $\eps$. An interesting and more general question is the problem of \textit{tolerant} testing of mutual and conditional mutual information. Here, given two threshold parameters $\eps_1, \eps_2 \in (0,1)$ with $\eps_1< \eps_2$, the goal is to distinguish between classes where the (conditional) mutual information is either at most $\eps_1$, or at least $\eps_2$. Often, tolerant testing requires new techniques. Over the last decade, there has been significant progress in this direction (see \cite{valiant2011testing,valiant2011estimating,valiant2011power,wu2016minimax,wu2019chebyshev,DBLP:conf/colt/CanonneJKL22}). However, to the best of our knowledge, the tolerant variant of (conditional) mutual information testing has not been explored yet. It would be very interesting to study this problem in the tolerant setting.

\item[(iv)] Finally, we would like to note that in several instances, for various properties of distributions, there is some inherent tolerance even in the algorithms designed for non-tolerant testing.  For example, for the setting when $\eps_1 \leq \eps_2/(2\sqrt{n})$, the non-tolerant testing algorithm for identity testing of distributions~\cite{batu_testing_2001,diakonikolas2019collision} based on the $\ell_2$ norm of the unknown distribution is sufficient, due to the reduction from $\ell_2$ distance to $\ell_1$ distance. The equivalence tester in $\ell_2$ \cite{chan_optimal_2013} which we use as a subroutine has inherent robustness as well, implying that some robustness may be extracted directly from our approach.

\end{enumerate}

\section*{Acknowledgements}
The authors would like to thank the anonymous reviewers for their suggestions, which improved the presentation of the paper. JS would like to thank Christopher Chubb for discussions on an early version of this problem. The authors thank Josep Lumbreras for proof reading the draft. This project is supported by the National Research Foundation, Singapore through the National Quantum Office, hosted in A*STAR, under its Centre for Quantum Technologies Funding Initiative (S24Q2d0009).

\newpage

\bibliographystyle{alpha} 
\bibliography{biblio,biblio_mt}

\newpage

\appendix

\section{Calculations}\label{sec:calculation}

\subsection{\texorpdfstring{Remaining Proof from \Cref{sec:prelim}}{Remaining Proof from the Preliminaries}}
\label{app:prelim}

\boundexp*
\begin{proof}
Let us first consider the first inequality. Since $0\leq x<1$, we know that $e^{-x}\leq 1-x/2$ holds. The lower bound holds because $f(x):=e^{-x}$ and $g(x):=1-x+x^2/4$ satisfy $f(0)=g(0)$, and for the first derivative, we have that $f'(x)=-e^{-x}\geq g'(x)=-1+x/2$, which holds since $e^{-x}\leq 1-x/2$.    
    
For the second inequality, since $0<x$, $e^{-x}\leq 1-x+x^2/2$ holds from Taylor's expansion.
For the lower bound, we consider $f(x):=e^{-x}$ and $h(x):=1-x+x^2/2-x^3/6$ which satisfy $f(0)=g(0)$, and for the first derivative, we have that $f'(x)=-e^{-x}\geq h'(x)=-1+x-x^2/2$, which follows from the upper bound.
\end{proof}

\subsection{\texorpdfstring{Remaining Proof from \Cref{sec:kl_hellinger_connection}}{Remaining Proof from the Reduction D to DH2}}\label{sec:reduction_app}

\continuitylem*

\begin{proof}
    We define $\zeta:=\eps/(48\log(d_Ad_C/\eps)d_Ad_C)$ and $\alpha=\zeta^k$, for $k$ to be determined. Note that $\zeta\log(1/\zeta)\leq \eps/(24d_Ad_C)$. Without loss of generality, we assume that $\alpha<1/2$. Further, we can write 
    \begin{equation}
    \label{eq:d_continuity:d_split}
        D(P_{AC}\|P_AP_C)=\sum_{a,c}p_{ac}\log\left(\frac{p_{ac}}{p_ap_c}\right)=\sum_{a,c}p_{ac}\log(p_{ac})-\sum_{a,c}p_{ac}\log(p_ap_c).
    \end{equation}
    We will choose $\alpha$ such that each of the $2d_Ad_C$ terms in \eqref{eq:d_continuity:d_split} changes by at most $\eps/4d_Ad_C$ when introducing $T_{AC}$, which then implies the desired result. Let us fix arbitrary $a$ and $c$. We will show 
    \begin{align}
        |((1-\alpha)p_{ac}+\alpha t_{ac})\log((1-\alpha)p_{ac}+\alpha t_{ac})-p_{ac}\log(p_{ac})|&\leq \frac{\eps}{4d_Ad_C}
        \\
        |((1-\alpha)p_{ac}+\alpha t_{ac})\log((1-\alpha)p_{a}p_c+\alpha t_{ac})-p_{ac}\log(p_{a}p_c)|&\leq \frac{\eps}{4d_Ad_C}
    \end{align}
    \begin{itemize}
        \item $p_{ac}\log(p_{ac})$: 
         We will do a case distinction. Let us first assume that $p_{ac}\geq \zeta$. Then for $k\geq 2$, we have 
        \begin{align}
            &|((1-\alpha)p_{ac}+\alpha t_{ac})\log((1-\alpha)p_{ac}+\alpha t_{ac})-p_{ac}\log(p_{ac})|\\
            &=\left|((1-\alpha)p_{ac}+\alpha t_{ac})\left(\log(1-\alpha)+\log(p_{ac})+\log\left(1+\frac{\alpha t_{ac}}{(1-\alpha)p_{ac}}\right)\right)-p_{ac}\log(p_{ac})\right|
            \\
            &= \left|\alpha(t_{ac}-p_{ac})\log((1-\alpha)p_{ac})+[(1-\alpha)p_{ac}+\alpha t_{ac}]\log\left(1+\frac{\alpha t_{ac}}{(1-\alpha)p_{ac}}\right)\right|\\
            &\leq \alpha\log\left(\frac{2}{p_{ac}}\right)+(p_{ac}+\alpha)\frac{2\alpha}{p_{ac}}
            \\
            &\leq 2\frac{\alpha}{p_{ac}}+2\alpha+2\frac{\alpha^2}{p_{ac}}\leq 6\frac{\alpha}{p_{ac}}\leq 6\zeta.
        \end{align}
        Now consider the case when  $p_{ac}\leq \zeta\ll 1$, then $|p_{ac}\log(p_{ac})|\leq \zeta\log(1/\zeta)$. Note that $(1-\alpha)p_{ac}+\alpha t_{ac}\leq p_{ac}+\alpha\leq 2\zeta$, such that we can bound 
        \begin{equation}
            \left|[(1-\alpha)p_{ac}+\alpha t_{ac}]\log((1-\alpha)p_{ac}+\alpha t_{ac})\right|\leq 2\zeta\log(1/\zeta).
        \end{equation}
        Together,
        \begin{equation}
            |((1-\alpha)p_{ac}+\alpha t_{ac})\log((1-\alpha)p_{ac}+\alpha t_{ac})-p_{ac}\log(p_{ac})|\leq 3\zeta\log(1/\zeta)\leq \eps/(4d_Ad_C).
        \end{equation}
        
        \item $p_{ac}\log(p_a p_c)$: 
        Similar to the above, we perform a case distinction. Let us first assume that $p_{ac}\geq \zeta$. Since $p_a,p_c\geq p_{ac}$, this implies $p_ap_c\geq \zeta^2$. Then, for $k\geq 2$, we have 
        \begin{align}
            &|((1-\alpha)p_{ac}+\alpha t_{ac})\log((1-\alpha)p_ap_c+\alpha t_{ac})-p_{ac}\log(p_ap_c)|\\
            &=|((1-\alpha)p_{ac}+\alpha t_{ac})\left(\log\left((1-\alpha)p_ap_c\right)+\log\left(1+\frac{\alpha t_{ac}}{(1-\alpha)p_ap_c}\right)\right)-p_{ac}\log(p_ap_c)|\\
            &= \left|\alpha(t_{ac}- p_{ac})\log\left((1-\alpha)p_ap_c\right)+[(1-\alpha)p_{ac}+\alpha t_{ac}]\log\left(1+\frac{\alpha t_{ac}}{(1-\alpha)p_ap_c}\right)\right|\\
            &\leq \left|2\alpha(t_{ac}- p_{ac})\log(\sqrt{(1-\alpha)p_ap_c})\right|+[(1-\alpha)p_{ac}+\alpha t_{ac}]\log\left(1+\frac{\alpha t_{ac}}{(1-\alpha)p_ap_c}\right)\\
            &\leq \left|2\alpha\log\left(\frac{2}{\sqrt{p_ap_c}}\right)\right|+(p_{ac}+\alpha)\log\left(1+\frac{2\alpha t_{ac}}{p_ap_c}\right)
            \\
            &\leq \left|2\alpha\log\left(\frac{1}{\zeta}\right)\right|+(p_{ac}+\alpha)\frac{2\alpha}{p_{ac}^2}\leq 6\zeta.
        \end{align}
        Now let us consider the other case. If $p_{ac}\leq \zeta$, then $|p_{ac}\log(p_ap_c)|\leq |2p_{ac}\log(p_{ac})|\leq  2\zeta\log(1/\zeta)$. With
        \begin{equation}
            \log\left(\frac{1}{(1-\alpha)p_ap_c+\alpha t_{ac}}\right)\leq 2\log\left(\frac{2}{(1-\alpha)p_{ac}+\alpha t_{ac}}\right),
        \end{equation}
        we find
        \begin{equation}
            |((1-\alpha)p_{ac}+\alpha t_{ac})\log((1-\alpha)p_{a}p_c+\alpha t_{ac})-p_{ac}\log(p_{a}p_c)|\leq 6\zeta\log(1/\zeta)\leq \eps/(4d_Ad_C).
        \end{equation}
    \end{itemize}
    This implies that each term in \eqref{eq:d_continuity:d_split} changes by at most $\eps/(2d_Ad_C)$ when introducing $T_{AC}$, which shows the desired bound.
\end{proof}

\subsection{Remaining Proofs from \texorpdfstring{\Cref{sec:ind_test_lb}}{MI-Testing (Lower Bounds)}}\label{sec:milbfar_app}
\milbfar*

\begin{proof}
When bounding 
\begin{equation}
    D_H^2(P_{AC},P_AP_C)=\sum_{(a,c)}\left(\sqrt{p_{ac}}-\sqrt{p_ap_c}\right)^2,
\end{equation}
we will only sum over coordinates $(a,c)$ for which $q(a,c)=\eps/(2d_Ad_C)$.  
Let $S_2^0$ denote the set of all tuples $(a,c)$ for which $p_{ac}=\eps/(2d_Ad_C)$. A Chernoff bound guarantees that with high probability, we have $|S_2^0|\geq\Omega(d_Ad_C)$. Next, we analyze $p_a$ for all $a\in S_2$. For a binomial variable $X_a\sim \text{Bin}(d_A,1/2)$, $\mu_a:=\E[X_a]=d_A/2$ we can write 
\begin{equation}
    p_a=\frac{\eps}{2d_C}+X_a\frac{\eps}{d_Ad_C}.
\end{equation}
We now use a Chernoff bound with $\delta:= \sqrt{3}\log(d_A)/\sqrt{\mu_a}=\log(d_A)/\sqrt{d_A/2}<1$ such that (for $d_A$ sufficiently large)
\begin{equation}
    \Pr[|X-d_A/2|\geq d_A/10]\leq \Pr[|X-d_A/2|\geq \log(d_A)\sqrt{3d_A/2}]\leq e^{-\log(d_A)}=1/d_A.
\end{equation}
With a union bound, we conclude that with high probability, simultaneously, all $p_a$ are within a multiplicative factor $1.2$ of their expectation value of $\eps/d_Ad_C$. Similarly, $p_c$ is lower bounded by $(1-\eps/1.2)/d_C$ since we also sum over $a=0$ here.
\begin{align}
    p_{a|a\in S_2^0}p_{c|a\in S_2^0}&=\left(\sum_{c'\neq c}p_{ac'}+\frac{\eps}{2d_Ad_C}\right)\left(\sum_{a'\neq a}p_{a'c}+\frac{\eps}{2d_Ad_C}\right)
    \\
    &\geq \left(p_a-\frac{\eps}{2d_Ad_C}\right)\left(p_c-\frac{\eps}{2d_Ad_C}\right).
\end{align}

Put together, it holds with high probability that $\forall (a,c)\in S_2^0: p_{a|a\in S_2^0}p_{c|a\in S_2^0}>3\E[p_ap_c]/4=3\eps/4d_Ad_C$. We then bound (note that under our constraints, $p_{ac}<p_{a}p_c$)
\begin{equation}
    D_H^2(P_{AC},P_AP_C)\geq \sum_{(a,c)\in S_2^{0}}\left(\sqrt{p_{ac}}-\sqrt{p_ap_c}\right)^2\geq 
    \Omega(d_Ad_C)\left(\sqrt{\frac{\eps}{2d_Ad_C}}-\sqrt{\frac{3\eps}{4d_Ad_C}}\right)^2\geq \Omega(\eps).
\end{equation}
\end{proof}

\milbzeroone*

\begin{proof}
Note that we have the following: 
\begin{align}
    &\Pr[K_a=0|X=0]
    \\
    &=\Pr[K_a=0|X=0,a\in S_2]\Pr[a\in S_2]+\Pr[K_a=0|X=0,a\in S_1]\Pr[a\in S_1]
    \\
    &=\Pr[K_a=0|X=0,a\in S_2]\underbrace{\Pr[a\in S_2]}_{O(1)}+\underbrace{\Pr[K_a=0|a\in S_1]\Pr[a\in S_1]}_{v}.
\end{align}
Similarly, 
\begin{align}
\Pr[K_a=0|X=1]&=\Pr[K_a=0|X=1,a\in S_2]\Pr[a\in S_2]+v    
\end{align}
and the second term cancels in the numerator. Let us start by proving (i).

\begin{itemize}
    \item[(i)]

\textbf{Case Without Collisions}
(For simplicity, we denote such $K_a$ by `0'), again $u:=\eps n/(d_Ad_C)$,
\begin{equation}
    \Pr[K_a=0|X=0]=(e^{-u})^{d_C}+v,\quad \Pr[K_a=0|X=1]=(e^{-u/2}/2+e^{-3u/2}/2)^{d_C}+v
\end{equation}

Assume $ud_C$ is upper bounded by some constant smaller than 1. Then
\begin{align}
&\frac{(\Pr[K_a=\ell_0|X=0]-\Pr[K_a=\ell_0|X=1])^2}{\Pr[K_a=\ell_0|X=0]+\Pr[K_a=\ell_0|X=1]}
\\
&=\frac{\left((e^{-u})^{d_C}-([e^{-u/2}+e^{-3u/2}]/2)^{d_C}\right)^2}{(e^{-u})^{d_C}+([e^{-u/2}+e^{-3u/2}]/2)^{d_C}+2v}
\\
&=\frac{\left((1-u+u^2/2+O(u^3))^{d_C}-(1-u+5u^2/8+O(u^3))^{d_C}\right)^2}{(1-u+u^2/2+O(u^3))^{d_C}+(1-u+5u^2/8+O(u^3))^{d_C}+2v}
\\
&=\frac{\left(\sum_{j=0}^{d_C}\binom{d_C}{j}\left[(-u+u^2/2+O(u^3))^j-(-u+5u^2/8+O(u^3))^j\right]\right)^2}{\Theta(1)}
\\
&=\frac{\left(\sum_{j=0}^{d_C}(-u)^j\binom{d_C}{j}\left[(1-uj/2+O(j^2u^2))-(1-5uj/8+O(j^2u^2))\right]\right)^2}{\Theta(1)}
\\
&\leq \frac{\left(\sum_{j=1}^{d_C}\binom{d_C}{j}\left[u(-u)^{j}j/8+O(u^{j+2}j^2)\right]\right)^2}{\Theta(1)}
\\
&\leq \frac{\left(\sum_{j=1}^{d_C}\binom{d_C}{j}u^{j+1}j\right)^2}{\Theta(1)}
\\
&=\frac{\left(u^2d_C\sum_{j=1}^{d_C}\binom{d_C-1}{j-1}u^{j-1}\right)^2}{\Theta(1)}
\\
&=\frac{u^4d_C^2(1+u)^{2d_C}}{\Theta(1)}=\Theta(u^4d_C^2).
\end{align}
\end{itemize}

Now let us prove (ii).

\begin{itemize}
    \item[(ii)]

\textbf{Single Collision:}
Here the procedure is the same, but now one of the $d_C$ results has to result in a hit. 
\begin{align}
&\sum_{\ell, \|\ell\|_1=1}\frac{(\Pr[K_a=\ell|X=0]-\Pr[K_a=\ell|X=1])^2}{\Pr[K_a=\ell|X=0]+\Pr[K_a=\ell|X=1]}
\\
&=d_C\frac{\left(u(e^{-u})^{d_C}-([e^{-u/2}u+3e^{-3u/2}u]/4)([e^{-u/2}+e^{-3u/2}]/2)^{d_C-1}\right)^2}{\Theta(u+v)}
\\
&=d_Cu\frac{\left((1-u+u^2/2+O(u^3))^{d_C}-(1-5u/4+O(u^2))(1-u+5u^2/8+O(u^3))^{d_C-1}\right)^2}{\Theta(u+v)}
\\
&=d_Cu^2\frac{\left(\sum_{j=0}^{d_C-1}\binom{d_C-1}{j}\left[(-u+u^2/2+O(u^3))^j(1+O(u))-(-u+5u^2/8 +O(u^3))^j(1+O(u))\right]\right)^2}{\Theta(u+v)}
\\
&=d_Cu^2\frac{\left(\sum_{j=0}^{d_C-1}(-u)^j\binom{d_C-1}{j}\left[(1-u/2+O(u^2))^j(1+O(u))-(1-5u/8+O(u^2))^j(1+O(u))\right]\right)^2}{\Theta(u+v)}
\\
&=d_Cu^2\frac{\left(\sum_{j=0}^{d_C-1}(-u)^j\binom{d_C-1}{j}\left[(1-ju/2+O(j^2u^2))(1+O(u))-(1-5ju/8+O(j^2u^2)))(1+O(u))\right]\right)^2}{\Theta(u+v)}
\\
&=d_Cu^2\frac{\left(\sum_{j=0}^{d_C-1}\binom{d_C-1}{j}\left[(-u)^{j}uj+(-u)^{j}O(u^2j^2)\right]\right)^2}{\Theta(u+v)}
\\
&\leq d_Cu^2\frac{\left((d_C-1)u^2\sum_{j=1}^{d_C-1}u^{j-1}\binom{d_C-2}{j-1}\right)^2}{\Theta(u+v)}=\Theta\left(d_C^3u^6(1+u)^{d_C}/v\right).
\end{align}
Note that $v=\alpha/(2d_C)$.
\end{itemize}

\end{proof}

\subsection{\texorpdfstring{Remaining Proofs from \Cref{sec:cmi_ub}}{Remaining Proofs from CMI-Testing (Upper Bounds)}}\label{sec:cmi_ub_app}

\lemprocessproperties*

\begin{proof}
Since the proof for $Y_{abc}$ is identical to the proof of $X_{abc}$, obtained by replacing $p_{ac|b}$ with $q_{ac|b}$, we will only prove the statements for $X_{abc}$. Let us start by bounding the $\E[X_{abc}]$.
\begin{align}
    \E[X_{abc}]&=\sum_{\ell=0}^{\infty}\pr[X_b=\ell]\E[X_{abc}|X_b=\ell]
    \\
    &=\sum_{\ell=0}^{\infty}\pr[X_b=\ell]\ell p_{ac|b}
    \\
    &=p_{ac|b}\E[X_b].
\end{align}
Similarly, we can calculate $\E[X_{abc}^2]$ as well. We have 
\begin{align}
    \E[X_{abc}^2] &=\sum_{\ell=0}^{\infty}\E[X_{abc}^2|X_b=\ell]\pr[X_b=\ell]
    \\
    &=\sum_{\ell=0}^{\infty}\left[\sum_{k=0}^{\ell}k^2\binom{\ell}{k}p_{ac|b}^{\ell}(1-p_{ac|b})^{k-\ell}\right]\pr[X_b=\ell]\\
    &=\sum_{\ell=1}^{\infty}\ell p_{ac|b}\left[\sum_{k=1}^{\ell}k\binom{\ell-1}{k-1}p_{ac|b}^{\ell-1}(1-p_{ac|b})^{k-\ell}\right]\pr[X_b=\ell]\\
    &=\sum_{\ell=1}^{\infty}\ell p_{ac|b}\left[1+\sum_{k=1}^{\ell}(k-1)\binom{\ell-1}{k-1}p_{ac|b}^{\ell-1}(1-p_{ac|b})^{k-\ell}\right]\pr[X_b=\ell]\\
    &=\sum_{\ell=1}^{\infty}\ell p_{ac|b}\left[1+(\ell-1)p_{ac|b}\right]\pr[X_b=\ell]\\
    &=p_{ac|b}\left((1-p_{ac|b})\mathbb{E}[X_{b}]+p_{ac|b}\mathbb{E}[X_{b}^2]\right).
\end{align}

In the following steps, we omit the indices $abc$ for simplicity.
\begin{align}
    \sum_{k=0}^{\infty}k\frac{x^{2k}e^{-x}}{(2k)!}&=\frac{x}{2}\sum_{k=1}^{\infty}\frac{x^{2k-1}e^{-x}}{(2k-1)!}=\frac{x}{2}\sum_{k=0}^{\infty}\frac{x^{2k+1}e^{-x}}{(2k+1)!}=\frac{xe^{-x}}{2}\sinh(x)
    \\
    \sum_{k=0}^{\infty}k\frac{x^{2k+1}e^{-x}}{(2k+1)!}&=\frac{1}{2}\sum_{k=0}^{\infty}(2k+1)\frac{x^{2k+1}e^{-x}}{(2k+1)!}-\frac{1}{2}\frac{x^{2k+1}e^{-x}}{(2k+1)!}=\frac{e^{-x}}{2}\left[x\cosh(x)-\sinh(x)\right]
    \\
    \sum_{k=0}^{\infty}k^2\frac{x^{2k}e^{-x}}{(2k)!}&=\sum_{k=0}^{\infty}\left[\frac{(2k)(2k-1)}{4}+\frac{k}{2}\right]\frac{x^{2k}e^{-x}}{(2k)!}=\frac{e^{-x}}{4}\left[x^2\cosh(x)+x\sinh(x)\right]
    \\
    \sum_{k=0}^{\infty}k^2\frac{x^{2k+1}e^{-x}}{(2k+1)!}&=\sum_{k=0}^{\infty}\left[\frac{(2k)(2k+1)}{4}-\frac{k}{2}\right]\frac{x^{2k+1}e^{-x}}{(2k+1)!}
    =\frac{e^{-x}}{4}\left[x^2\sinh(x)-x\cosh(x)+\sinh(x)\right]
\end{align}
Let $x_b:=\tilde N_Sp_b$. We can now calculate:
\begin{align}
    \E[X_{b}]&=\frac{x_{b}}{2}-\frac{e^{-x_{b}}}{2}\sinh(x_{b}),\quad \quad \E[X_{b}^2]=\frac{x_{b}^2}{4}-\frac{x_{b}e^{-2x_{b}}}{4}+\frac{e^{-x_{b}}}{4}\sinh(x_{b}).
\end{align}
Also note that $\E[X_{b}]\leq \E[X_{b}^2]$. To bound this further, note that writing out $\sinh$ and using the fact that $e^x\geq 1+x$, we have 
\begin{align}
    \E[X_b^2]&\leq \frac{2x_b^2+1}{8}-\frac{(1+2x_b)(1-2x_b)}{8}=\frac{3}{4}x_b^2\label{eq:bound_exp_xb_squared}
\end{align}
Moreover, since $\forall b\in [d_B],\forall k\in \mathbb{N}:\E[X_b^k]=\E[Y_b^k]$, $\E[X_{b}]= \E[Y_b]$ as well as $\E[X_{b}^2] = \E[Y_b^2]$.

In the following, assume $(a,c)\neq (a',c')$, but $a=a'$ or $c=c'$ is permitted. Note that for $0<x,y$ with $x+y<1$, we have that $(1-x-y)^z\leq (1-x-y+xy)^z=(1-x)^z(1-y)^z$. Then
\begin{align}
&\E[X_{abc}X_{a'bc'}]
\\
&=\sum_{k,\ell,s=0}k \ell \Pr[X_{abc}=k,X_{abc'}=\ell|X_b=s]\Pr[X_b=s]
\\
&=\sum_{s=0}^{\infty}\left[\sum_{k=0}^{s}\sum_{\ell=0}^{s-k}k \ell \binom{s}{k}\binom{s-k}{\ell}p_{ac|b}^{\ell}p_{a'c'|b}^k(1-p_{ac|b}-p_{a'c'|b})^{s-k-\ell}\right]\Pr[X_b=s]
\\
&=\sum_{s=2}^{\infty}\left[\sum_{k=1}^{s}\sum_{\ell=1}^{s-k}k \ell \binom{s}{k}\binom{s-k}{\ell}p_{ac|b}^{\ell} p_{a'c'|b}^k(1-p_{ac|b}-p_{a'c'|b})^{s-k-\ell}\right]\Pr[X_b=s]
\\
&\leq\sum_{s=2}^{\infty}\left[\sum_{k=1}^{s}k\binom{s}{k}p_{a'c'|b}^k(1-p_{a'c'|b})^{s-k}\sum_{\ell=1}^{s-k}\ell \binom{s-k}{\ell}p_{ac|b}^{\ell}\frac{(1-p_{ac|b})^{s-k-\ell}}{(1-p_{a'c'|b})^{\ell}}\right]\Pr[X_b=s]
\end{align}
Due to symmetry, we can assume without loss of generality that $p_{a'c'|b}<p_{ac|b}$, which in particular means that $p_{a'c'|b}<1/2$. We now note that 
\begin{align}
    &\sum_{\ell=1}^{s-k}\frac{\ell}{{(1-p_{a'c'|b})^{\ell}}} \binom{s-k}{\ell}p_{ac|b}^{\ell}(1-p_{ac|b})^{s-k-\ell}
    \\
    &\leq \sum_{\ell=1}^{s-k}2^{\ell}\ell \binom{s-k}{\ell}p_{ac|b}^{\ell}(1-p_{ac|b})^{s-k-\ell}
    \\
    &\leq 2p_{ac|b}(s-k)\sum_{\ell=1}^{s-k} 2^{\ell-1}\binom{s-k-1}{\ell-1}p_{ac|b}^{\ell-1}(1-p_{ac|b})^{s-k-\ell}
    \\
    &=2p_{ac|b}(s-k)\left(2p_{ac|b}+(1-p_{ac|b})\right)^{s-k-1}
    \\
    &\leq p_{ac|b}(s-k)2^{s-k}.
\end{align}
In the following, bounding $s-k\leq s$ also allows us to simplify the sum over $k$ analogously. In the final steps, we use inequalities from \Cref{lemma:bounds_exp}.
\begin{align}
    \E[X_{abc}X_{a'bc'}]&\leq p_{ac|b}p_{a'c'|b}\sum_{s=2}^{\infty}\left[\sum_{k=1}^{s}\binom{s-1}{k-1}s(s-k)2^{s-k}p_{a'c'|b}^{k-1}(1-p_{a'c'|b})^{s-k}\right]\Pr[X_b=s]
    \\
    &\leq p_{ac|b}p_{a'c'|b}\sum_{s=2}^{\infty}s^22^s\pr[X_b=s]
    \\
    &\leq 
    p_{ac|b}p_{a'c'|b}e^{\sqrt{2}x_{b}}\sum_{s=2}^{\infty}s^2\left(\frac{(\sqrt{2}x_{b})^{2s}e^{-\sqrt{2}x_{b}}}{(2s)!}+\frac{(\sqrt{2}x_{b})^{2s+1}e^{-\sqrt{2}x_{b}}}{(2s+1)!}\right)
\end{align}
We now let $k_b:=\sqrt{2}x_b=\tilde N_S(\sqrt{2}p_b)$, and $K_b$ accordingly. Then
\begin{align}
    \E[X_{abc}X_{a'bc'}]&= p_{ac|b}p_{a'c'|b}e^{\sqrt{2}x_b}(\E[K_b^2]-(k_b^2e^{-k_b}/2+k_b^{3}e^{-k_b}/6))
    \\
    &= p_{ac|b}p_{a'c'|b}e^{\sqrt{2}x_b}\left(\frac{2k_b^2+1}{8}-\frac{2k_b+1}{8}e^{-2x
    _b}-\left(\frac{k_b^2e^{-k_b}}{2}+\frac{k_b^{3}e^{-k_b}}{6}\right)\right)
    \\
    &\leq p_{ac|b}p_{a'c'|b}e^{\sqrt{2}x_b}\Bigg(\frac{2k_b^2+1}{8}-\frac{2k_b+1}{8}\left(1-2k_b+2k_b^2-\frac{4k_b^3}{3}\right)
    \\
    &\quad\quad-\left(\frac{k_b^2}{2}+\frac{k_b^{3}}{6}\right)\left(1-k_b+\frac{k_b^2}{4}\right)\Bigg)
    \\
    &\leq p_{ac|b}p_{a'c'|b}e^{\sqrt{2}x_b}\left(\frac{2k_b^2+1}{8}-\frac{1-2k_b^2}{8}-\frac{k_b^2}{2}+\frac{k_b^{3}}{2}\right)\leq 12p_{ac|b}p_{a'c'|b}x_b^3.
\end{align}

\end{proof}

\lemcmismallvar*

\begin{proof}
Recall that $Z=\sum_{i}Z_{i}$ where $Z_{i}:=X_{i}X'_{i}-2X_{i}Y_{i}+Y_{i}Y'_{i}$ and we would like to prove that 
\begin{align}\label{eqn:cmi_ub_small_var}
\Var[Z]&= \sum_{(a,b,c)\in S}\Var[Z_{abc}] + \Cov\limits_{(a,b,c) \neq (a',b,c')}[Z_{abc}, Z_{a'bc'}]
\\
&\leq 2\cdot 10^3\left(\|\subD{P}{S}{ABC}\|_2^2+\| \subD{Q}{S}{ABC}\|_2^2\right)\tilde N_S^2.  
\end{align}

This result follows directly from the two following claims, \Cref{cl:cmi_ub_small_varzi} and \Cref{cl:cmi_ub_small_covi}, which bound the variance and covariance, respectively.
\end{proof}

\begin{restatable}{claim}{cmiubsmallvarzi}\label{cl:cmi_ub_small_varzi}

$\sum\limits_{(a,b,c)\in S}\Var[Z_{abc}] \leq 8\left(\|\subD{P}{S}{ABC}\|_2^2+\|\subD{Q}{S}{ABC}\|_2^2\right)\tilde N_S^2$.

\end{restatable}

\begin{proof}
Let us first compute $\E[Z_i^2]$ and $\E[Z_i]^2$.
\begin{align}
Z_i^2&=X_i^2(X_i')^2+4X_i^2Y_i^2+Y_i^2(Y_i')^2-4X_i^2X_i'Y_i+2X_iX_i'Y_iY_i'-4X_iY_i^2Y_i'.
\end{align}
So,
\begin{align}
\E[Z_i^2]&=\E[X_i^2]^2+4\E[X_i^2]\E[Y_i^2]+\E[Y_i^2]^2-4\E[X_i^2]\E[X_i]\E[Y_i]+2\E[X_i]^2\E[Y_i]^2-4\E[X_i]\E[Y_i^2]\E[Y_i]    
\end{align}
Similarly, we have:
\begin{align}
\E[Z_i]^2&=\E[X_i]^4-4\E[X_i]^3\E[Y_i]+6\E[X_i]^2\E[Y_i]^2-4\E[X_i]\E[Y_i]^3+\E[Y_i]^4    
\end{align}

Now let us compute $\Var[Z_i]$.
\begin{align}
    \Var[Z_i]&=\E[Z_i^2]-\E[Z_i]^2
    \\
    &=\E[X_i^2]^2+4\E[X_i^2]\E[Y_i^2]+\E[Y_i^2]^2-\E[X_i^4]-\E[Y_i^4]
    \\
    &\quad -4\E[X_i]\E[Y_i]\left(\E[X_i]\E[Y_i]+\E[X_i^2]+\E[Y_i]^2-\E[X_i]^2-\E[Y_i^2]\right)
    \\
    &=\E[X_i^2]^2+4\E[X_i^2]\E[Y_i^2]+\E[Y_i^2]^2-\E[X_i^4]-\E[Y_i^4]
    \\
    &\quad -4\E[X_i]\E[Y_i]\left(\E[X_i]\E[Y_i]+\Var[X_i]+\Var[Y_i]\right)
    \\
    &\leq 2(\E[X_i^2]+\E[Y_i^2])^2\label{eq:cmi_upper_small_var_bound},
\end{align}
such that we can bound
\begin{eqnarray}
\sum_{(a,b,c)\in S}\Var[Z_{abc}] \leq 2\sum_{(a,b,c)\in S}(\E[X_{abc}^2]+\E[Y_{abc}^2])^2     
\end{eqnarray}

From \Cref{lem:cmi:small:xy}, we know that $\E[X_{abc}^2] = p_{ac|b}(1-p_{ac|b})\E[X_{b}]+p_{ac|b}^2\E[X_{b}^2]$, and $\E[Y_{abc}^2]$ analogously with $q$ instead of $p$. Thus, using $\E[X_b^2]=\E[Y_b^2]\leq (Np_b)^2$ (see \eqref{eq:bound_exp_xb_squared}), we have:
\begin{eqnarray}
\sum_{(a,b,c)\in S}\Var[Z_{abc}] &\leq& 2\sum_{(a,b,c)\in S}(\E[X_{abc}^2]+\E[Y_{abc}^2])^2 \\
&\leq & 2\sum_{(a,b,c)\in S}(p_{ac|b}\E[X_{b}]+p_{ac|b}^2\E[X_{b}^2]+q_{ac|b}\E[Y_{b}]+q_{ac|b}^2\E[Y_{b}^2])^2
\\
&\leq& 8 \sum_{(a,b,c)\in S} \left(p_{ac|b}\E[X_b^2] + q_{ac|b}\E[Y_b^2]\right)^2 \\
&\leq & 8\sum_{(a,b,c)\in S}(p_{ac|b}+q_{ac|b})^2\tilde N_S^4p_b^4\\
&\leq& 8(\|\subD{P}{S}{ABC}\|_2^2+\|\subD{Q}{S}{ABC}\|_2^2)\tilde N_S^2.
\end{eqnarray}    
\end{proof}

\begin{restatable}{claim}{cmiubsmallcovi}\label{cl:cmi_ub_small_covi}
$\Cov\limits_{(a,b,c) \neq (a',b,c')}[Z_{abc}, Z_{a'bc'}] \leq 12^3\left(\|\subD{P}{S}{ABC}\|_2^2+\|\subD{Q}{S}{ABC}\|_2^2\right)\tilde N_S^2$.    
\end{restatable}

\begin{proof}
Let us now compute the covariance. From the definition, we have the following:
\begin{eqnarray}
\Cov\limits_{(a,b,c) \neq (a',b,c')}[Z_{abc}, Z_{a'bc'}] &=& \sum_{(a,c)}\sum_{(a',c')\neq (a,c)}\E[(Z_{abc}-\E[Z_{abc}])(Z_{a'bc'}-\E[Z_{a'bc'}])] \\
&=& \sum_{(a,c)}\sum_{(a',c')\neq (a,c)}\E[Z_{abc}Z_{a'bc'}]-\E[Z_{abc}]\E[Z_{a'bc'}] \\
&\leq& \sum_{(a,c)}\sum_{(a',c')\neq (a,c)}\E[Z_{abc}Z_{a'bc'}]
\end{eqnarray}

For ease of reading, let us denote $Z_{abc}$ as $Z_i$ and $Z_{a'bc'}$ as $Z_i'$. We use similar notations for $X_i, X_i'$ and $Y_i,Y_i'$ as well.

Then
\begin{align}
Z_i Z_i' &= (X_iX_i'-2X_iY_i+Y_iY_i')(X_jX_j'-2X_jY_j+Y_jY_j') \\
&= X_iX_i'X_jX_j'-2X_iX_i'X_jY_j+X_iX_i'Y_jY_j'-2X_iY_iX_jX_j'+4X_iY_iX_jY_j-2X_iY_iY_jY_j'
\\
&\quad +Y_iY_i'X_jX_j'-2Y_iY_i'X_jY_j+Y_iY_i'Y_jY_j'    
\end{align}

So, we have:
\begin{align}
    \E[Z_i Z_i'] &= \E[(X_iX_i'-2X_iY_i+Y_iY_i')(X_jX_j'-2X_jY_j+Y_jY_j')] \\
    &= (\E[X_iX_j])^2-2\E[X_iX_j]\E[Y_j]E[X_i]+(\E[X_iY_j])^2-2\E[X_iX_j]\E[Y_i]\E[X_j]
    \\& \quad\quad +4\E[X_iX_j]\E[Y_iY_j]-2\E[Y_iY_j]\E[X_i]\E[Y_j]
    \\
    & \quad\quad +(\E[Y_iX_j])^2-2\E[Y_iY_j]\E[X_j]\E[Y_i]+(\E[Y_iY_j])^2    
\end{align}

and summing up gives 
\begin{align}
    &\sum_{i,j}(\E[X_iX_j])^2-4(\E[X_iX_j]+\E[Y_iY_j])\E[Y_j]\E[X_i]\\
    &\quad +2\E[X_i]^2\E[Y_j]^2+4\E[X_iX_j]\E[Y_iY_j]
    +(\E[Y_iY_j])^2
    \\
    &\leq 2 \sum_{i,j}\left((\E[X_iX_j]+\E[Y_iY_j])^2+2\E[X_i]^2\E[Y_j]^2\right)\\
    &\leq 4 \sum_{(a,b,c)}\sum_{(a',b,c') \neq (a,b,c)}\left(\E[X_{abc}X_{a'bc'}]^2+\E[Y_{abc}Y_{a'bc'}]^2 + 2\E[X_{abc}]^2\E[Y_{a'bc'}]^2\right).
\end{align}

Now let us first bound the term $\sum_{(a,b,c)}\sum\limits_{(a',b,c') \neq (a,b,c)}\E[X_{abc}X_{a'bc'}]^2$ using \Cref{lem:cmi:small:xy}.
    \begin{align}    \sum_b\sum_{(a,c)}\sum_{(a',c') \neq (a,c)}\E[X_{abc}X_{a'bc'}]^2&\leq 12^2\sum_b\sum_{(a,c)}\sum_{(a',c')}p_{ac|b}^2p_{a'c'|b}^2p_b^6\tilde N_S^6
    \\
    &\leq 12^2\sum_b\sum_{(a,c)}p_{abc}^2p_b^4\tilde N_S^6
    \\
    &\leq 12^2\|\subD{P}{S}{ABC}\|_2^2\tilde N_S^2,
\end{align}

Similarly, we can bound $\sum\limits_{(a,b,c)}\sum\limits_{(a',b,c') \neq (a,b,c)} \E[Y_{abc}Y_{a'bc'}]^2 \leq 12^2\|\subD{Q}{S}{ABC}\|_2^2\tilde N_S^2$.

Now let us bound 
\begin{align}
\sum_b\sum_{(a,c)}\sum_{(a',c') \neq (a,c)}\E[X_{abc}]^2\E[Y_{a'bc'}]^2&\leq \sum_b\sum_{(a,c)}\sum_{(a',c')} p_{ac|b}^2(p_{a'|b}p_{c'|b})^2\tilde N_S^8p_b^8\\
&\leq \sum_b\sum_{(a,c)} p_{ac|b}^2\tilde N_S^8p_b^8\leq \|\subD{P}{S}{ABC}\|_2^2\tilde N_S^2.
\end{align}

Combining all the above, we can say that
\begin{equation}
    \Cov\limits_{(a,b,c) \neq (a',b,c')}[Z_{abc}, Z_{a'bc'}] \leq 12^3\left(\|\subD{P}{S}{ABC}\|_2^2+\|\subD{Q}{S}{ABC}\|_2^2\right)\tilde N_S^2.     
\end{equation}
This completes the proof.   
\end{proof}

\end{document}